\documentclass[11pt]{article}
\usepackage[margin=1in]{geometry}
\usepackage[english]{babel}
\usepackage{xr}
\usepackage{booktabs}
\usepackage{setspace}
\usepackage{natbib}

\usepackage{tikz}
\usetikzlibrary{arrows,automata,positioning}

\usepackage{mathrsfs,hyperref, algorithm}%, algorithmic}
\usepackage[]{amsmath}
\usepackage{amsthm}
\usepackage{fix-cm}
\usepackage[]{amssymb}
\usepackage[]{latexsym}
\usepackage[right]{eurosym}
\usepackage[T1]{fontenc}
\usepackage[]{graphicx}
\usepackage[]{epsfig}
\usepackage{fancyhdr}
\usepackage{pstricks}
\usepackage{multirow}
\usepackage{subcaption}

\allowdisplaybreaks

\makeatletter

\newcommand{\bbr}{\mathbb{R}}
\newcommand{\E}{\mathbb{E}}

\newcommand{\bbn}{\mathbb{N}}
\renewcommand{\P}{\mathbb{P}}

\newcommand{\bbq}{\mathbb{Q}}
\newcommand{\Q}{\bbq}
\newcommand{\EQ}{\E^\Q}
\newcommand{\bbg}{\mathbb{G}}

\newcommand{\bbx}{\mathbb{X}}

\newcommand{\bbt}{\mathbb{T}}
\newcommand{\bbd}{\mathbb{D}}

\newcommand{\fcal}{\mathcal{F}}
\newcommand{\gcal}{\mathcal{G}}
\newcommand{\dcal}{\mathcal{D}}

\newcommand{\bcal}{\mathcal{B}}

\newcommand{\ncal}{\mathcal{N}}

\newcommand{\ccal}{\mathcal{C}}

\newcommand{\acal}{\mathcal{A}}
\newcommand{\mcal}{\mathcal{M}}
\newcommand{\xcal}{\mathcal{X}}
\newcommand{\ycal}{\mathcal{Y}}

\newcounter{modcount}
\newcommand{\modulo}[2]{%
\setcounter{modcount}{#1}\relax
\ifnum\value{modcount}<#2\relax
\else\relax
\addtocounter{modcount}{-#2}\relax
\modulo{\value{modcount}}{#2}\relax
\fi}
\newcommand{\tablepictures}[4][c]{\begin{tabular}[#1]{@{}c@{}}#2\vspace{0.5cm}\\(\alph{#4}) #3\end{tabular}}
\newcounter{gridsearch}
\newcommand{\tabpic}[2]{
    \stepcounter{gridsearch}
    \modulo{\thegridsearch}{2}
    \ifnum\value{modcount}=0
        \tablepictures[t]{#1}{#2}{gridsearch}\\[2.0cm]
    \else
        \tablepictures[t]{#1}{#2}{gridsearch}&~&
    \fi
}

\makeatother
\newtheorem{lemma}{Lemma}[section]
\newtheorem{proposition}[lemma]{Proposition}
\newtheorem{theorem}[lemma]{Theorem}
\newtheorem{corollary}[lemma]{Corollary}
\newtheorem{definition}[lemma]{Definition}
\newtheorem{example1}[lemma]{Example}
\newtheorem{rem1}[lemma]{Remark}
\newtheorem{assumption}[lemma]{Assumption}
\newtheorem{alg1}[lemma]{Algorithm}
\newtheorem{me1}[lemma]{Mechanism}

\newenvironment{remark}{\begin{rem1}\rm}{\end{rem1}}
\newenvironment{example}{\begin{example1}\rm}{\end{example1}}

\usepackage{color}
\newcommand{\T}{\top}
\newcommand{\diag}{\operatorname{diag}}

\DeclareMathOperator*{\argmax}{arg\,max}

\newcommand\ind[1]{\mathbb{I}_{\{#1\}}}

\begin{document}

\title{Pricing of debt and equity in a financial network with comonotonic endowments}
\author{Tathagata Banerjee \thanks{Washington University in St.\ Louis, Department of Electrical and Systems Engineering, St.\ Louis, MO 63130, USA.} \and
Zachary Feinstein \thanks{Stevens Institute of Technology, School of Business, Hoboken, NJ 07030, USA. \tt{zfeinste@stevens.edu}}}
\date{\today}
\maketitle
\abstract{
In this paper we present formulas for the valuation of debt and equity of firms in a financial network under comonotonic endowments.  We demonstrate that the comonotonic setting provides a lower bound and Jensen's inequality provides an upper bound to the price of debt under Eisenberg-Noe financial networks with bankruptcy costs.  Such financial networks encode the interconnection of firms through debt claims.  The proposed pricing formulas consider the realized, endogenous, recovery rate on debt claims.  
%Special consideration is given to the CAPM setting in which firms invest in correlated portfolios so as to provide analytical stress testing formulas.  
We endogenously construct the comonotonic endowment setting from a equity maximizing standpoint with capital transfers.  We conclude by, numerically, comparing the network valuation problem with two single firm baseline heuristics which can, respectively, approximate the price of debt and equity.
}

\section{Introduction}\label{sec:intro}
Valuation adjustments, especially credit valuation adjustment [CVA], have become an important part of derivative valuation by any bank since the 2007-2009 financial crisis.  CVA proposes an adjustment to the ``traditional'' price of a derivative, as in the seminal paper of \cite{merton1974}, to account for counterparty risk of the instrument.  However, CVA does not take into account the full network of interconnections that exist within a financial system. As evidenced by the 2007-2009 financial crisis, considering the risk of a single firm alone can cause gross misspecification in firm health.  In this work, we will focus on interconnections through correlated assets as well as interbank debt claims. These interconnections effectively link the balance sheets of different banks and make the value of a firm dependent on the performance of other firms. These shared connections might open up avenues for shared prosperity but also introduce potential channels for contagion. These avenues of contagion become particularly significant during a financial crisis where the default of one firm might cause the failure in other firms.  This effect is also referred to as cascading defaults.

The rest of the paper is organized as follows. A review of relevant literature is provided in Section~\ref{sec:literature}. Detailed motivation for the comonotonic endowment setting, which is central to this work, is provided in Section~\ref{sec:motivation}. The primary innovations of this paper are highlighted in Section~\ref{sec:contributions}. In Section~\ref{sec:setting} we provide a description of our mathematical setting and necessary background on the Eisenberg-Noe framework.  Section~\ref{sec:price} considers network clearing when firms have comonotonic endowments. We provide the expectation of the equilibrium payments, equity, and wealth.  
Further, we prove that these expected values can provide upper and lower bounds for the general random endowment setting of, e.g., \cite{gourieroux2012,barucca2016valuation} in Section~\ref{sec:bounds}.  
%We wish to highlight the special case in which bank endowments follow a CAPM setting with idiosyncratic risks to provide analytical results of the NVA bounds for endowments in Section~\ref{sec:capm}.  
Section~\ref{sec:statics} considers simple comparative statics of the provided valuations with respect to the different system parameters under {a lognormal} setting for clear comparisons to \cite{merton1974}.  
Section~\ref{sec:conclusion} concludes. All proofs can be found in Online Appendix~\ref{sec:proofs}.

\subsection{Literature review}\label{sec:literature}
Since the 2007-2009 financial crisis there have been significant efforts to study the effects of interconnection within the financial system, with particular emphasis on modeling and quantifying systemic risk. One of the primary approaches for modeling systemic risk is fundamentally network-based. The seminal paper \cite{EN01} models these financial dependencies as a directed graph. In this approach, firms must meet their full liabilities by transferring their assets, otherwise they are deemed insolvent and are unable to pay out in full. This inability to make the full payment can, in turn, make other banks default, thus resulting in cascading failures. This interdependency of realized (clearing) payments is modeled as a fixed point problem. \cite{EN01} prove the existence and uniqueness of the clearing payments and provides an elegant algorithm for the computation of the same. This baseline setting of Eisenberg-Noe has been extended in many directions to account for more realistic situations. Interconnections through cross-holdings have been considered in \cite{suzuki2002,gourieroux2012, gourieroux2013}. 
Realistic default mechanisms in the form of bankruptcy cost has been studied in \cite{E07,RV13,GY14,AW_15,CCY16,veraart2017distress}.  Central banks and regulatory bodies have incorporated these network models into their stress tests of the financial system (see, e.g., \cite{Anand:Canadian,HK:Modeling,ELS13,U11,GHK2011}).  As such, valuing claims that take the full network effects into account is imperative so that the true risk of the claims are also taken into account.  And, in fact, \cite{SS18pricing} concludes, empirically, that at present financial contagion is \emph{not} priced into interbank markets; as such, the focus of this work is on a network valuation adjustment scheme [NVA] to determine what prices should be if markets were to accurately price such contagion.  The NVA approach is closely related to the work of~\cite{CS07network} which constructs a network of obligations with infinite maturity and in which debts are refinanced at time of default.  In this work we construct a tractable NVA approach specifically in an Eisenberg-Noe framework.

The NVA approach in an Eisenberg-Noe framework was, to our knowledge, first considered in \cite{gourieroux2012} and extended to more general clearing mechanisms in \cite{barucca2016valuation}.  We wish to note that the general methodology of NVA in the Eisenberg-Noe framework, as studied in \cite{gourieroux2012}, requires partitioning the endowment space into $2^n$ possible default scenarios in a system with $n$ banks. Therefore computing the expectation suffers from the curse of dimensionality and is typically computationally intractable for realistic systems.  In fact, a common thread of the existing literature is that explicit, analytical solutions are considered only for cases with either no direct interconnection between firms or where the number of firms in the system is very low.  Further complications arise if we consider bankruptcy costs along the lines of \cite{RV13}. Without bankruptcy costs, the $2^n$ default scenarios result in the partition of the bank endowment space into $2^n$ mutually exclusive and convex regions. However, in the presence of bankruptcy costs, these partitioned regions are not, in general, convex.  We refer the reader to Online Appendix \ref{sec:general} for elucidation of this point. Providing any analytical solution in this case becomes very challenging even in small systems (as has been highlighted in past works such as \cite{suzuki2002}). Hence a numerical approach has been generally followed, i.e., via Monte Carlo simulations.

\subsection{Motivation for comonotonic endowments}\label{sec:motivation}
%{\rd ADD MORE DETAILS HERE???}
We note, first, that the problem of finding the expectation of the clearing wealths and payments under (integrable and nonnegative) random endowments $X \in (L^1_+)^n$ was considered in~\cite{gourieroux2012} in the case of no bankruptcy costs ($\alpha_x = \alpha_L = 1$) but with cross-ownership.  We replicate those results and extend them to consider the case with bankruptcy costs in Online Appendix~\ref{sec:general}.  
However, in this work we focus on a financial system in which banks hold comonotonic endowments; this is used directly in Section~\ref{sec:price} and as bounds for the general setting in Section~\ref{sec:bounds}.  
Briefly, a random vector $X \in (L^1)^n$ is called \emph{comonotonic} if it is equal in distribution to $f(q)$ for some random variable $q \in L^0$ and nondecreasing function $f: \bbr \to \bbr^n$; this is formalized in Section~\ref{sec:comonotone}.
There are three complementary motivating justifications for the comonotonic endowment setting fundamental to this work: theoretical, empirical, and computational.
\begin{enumerate}
\item \emph{Portfolio optimization}: 
First, as is standard in the literature, if all firms are portfolio optimizers and do not take any other firm's investments into account, the chosen endowments will all be countermonotonic to the pricing kernel~(see, e.g., \cite{PY75}). 
Formally, we wish to consider the risk-sharing equilibrium problem of~\cite{buhlmann1980economic,buhlmann1984general}; this problem, {under a finite probability space}, is equivalent to the Arrow-Debreu equilibrium~\cite{arrow1954existence} (see, e.g.,~\cite{anthropelos2017equilibrium,bichuch2020endogenous}).  Let $(\Omega,\fcal,\P)$ denote a probability space describing the financial market of $n$ utility maximizing agents.  Agent $i$ has some initial (uniformly bounded random) endowment $X_i \in L^\infty$ which she can trade with the other agents so as to maximize her own expected utility $\E[u_i]$.  A \emph{B\"uhlmann equilibrium} is a pair $(Y^*,\Q)$ so that
    \begin{enumerate}
    \item \emph{utility maximizing}: for every agent $i$, $Y_i^* \in \argmax_{Y_i \in L^\infty} \left\{\E[u_i(Y_i)] \; | \; \E^\Q[Y_i] \leq \E^\Q[X_i]\right\}$ and
    \item \emph{market clearing}: $\sum_{i = 1}^n Y_i^* = \sum_{i = 1}^n X_i$.
    \end{enumerate} 
    As proven in Lemma 3 of~\cite{tsanakas2006risk}, any {equilibrium vector of} portfolio holdings {$Y^*$} will be comonotonic.  {The details of this argument can be found in Online Appendix~\ref{sec:buhlmann}.}  In fact, \cite{tsanakas2006risk} proves the comonotonicity of the portfolio holdings even with (potentially heterogeneous) distortions on the probability measure $\P$ which can account for, e.g., ambiguity aversion of the economic agents.

\item \emph{Empirical correlations}: 
The empirical evidence shows that the bank assets exhibit a very high degree of rank correlation. The Spearman correlations of the daily returns (adjusted for dividends and stock splits) for five U.S.\ banks (JP Morgan Chase and Co, Goldman Sachs Group, Bank of America Corporation, Morgan Stanley, and Citigroup Inc.) are shown in Table~\ref{tab:Correlation} from January 1, 2015 to December 31, 2020. This notion of a homophilic financial system is reported in US and German banking systems by~\cite{elliott2019systemic}: ``Banks often have similar real exposures to their financial counterparties.''
\begin{table}
\centering
\begin{tabular}{@{}|c|c||c|c|c|c|c|@{}}
\hline
Name                        & Ticker & JPM  & GS   & BAC  & MS   & C    \\ \hline\hline
JP Morgan Chase and Co      & JPM    & 1.00 & 0.82 & 0.88 & 0.85 & 0.87 \\ \hline
Goldman Sachs               & GS     & 0.82 & 1.00 & 0.81 & 0.86 & 0.81 \\ \hline
Bank of America Corporation & BAC    & 0.88 & 0.81 & 1.00 & 0.84 & 0.87 \\ \hline
Morgan Stanley              & MS     & 0.85 & 0.86 & 0.84 & 1.00 & 0.83 \\ \hline
Citigroup Inc               & C      & 0.87 & 0.81 & 0.87 & 0.83 & 1.00 \\ \hline
\end{tabular}
\caption{Spearman correlations of the daily returns for five large U.S.\ institutions.}
\label{tab:Correlation}
\end{table}

\item \emph{Computational and analytical tractability}: 
The general random endowment setting with $X \in (L^1_+)^n$, as considered in~\cite{gourieroux2012}, requires computing the measure of $2^n$ regions in $\bbr^n$.  This would typically require Monte Carlo simulation and suffer greatly from the curse of dimensionality.  The comonotonic framework, in contrast, requires computing the measure of only $n$ intervals in $\bbr$; this is tractable both computationally and analytically.  As such, the comonotonic framework allows for resilience and stability analysis as considered in~\cite{AOT13}.  That work restricts the network topologies to those constructed from ring and completely connected networks with i.i.d.\ Bernoulli shocks; by considering comonotonic endowments instead, we are able to analytically study resilience and stability for general network topologies and more general shock types. 
In particular, in Section~\ref{sec:bounds}, we will demonstrate that we are able to provide lower and upper bounds on the expectation of the system behavior under any random endowment using the comonotonic setting.  
This is considered in the special case of a common, systematic, shock with idiosyncratic shocks in {Online Appendix}~\ref{sec:capm}.  
By focusing on the comonotonic framework, we are able to take advantage of the tractable analytical results with at most $n$ intervals for consideration.
\end{enumerate}

\subsection{Primary contributions}\label{sec:contributions}

In light of the aforementioned analytical and computational limitations to NVA in the current literature, and with the highlighted motivations in Section~\ref{sec:motivation}, we define and study the price of debt and equity under comonotonic endowments.  
Within this work, we broadly equate prices with expectations.  That is, up to discounting, the expected payments $\E[p_i]$ of a total obligation $\bar p_i$ is viewed as the price of debt and the expected equity $\E[E_i]$ is viewed as the price of equity.  This is made explicit in Section~\ref{sec:statics} with the use of discounting and a risk-neutral measure $\Q$.
The primary contributions of this paper are as follows:
\begin{enumerate}
\item We formulate an \textbf{analytical formulation for NVA} under the comonotonic endowment setting in Section~\ref{sec:price}.
In this setting, the default regions can be characterized by at most $n$ intervals on $\bbr_+$ and become tractable analytically. This tractability extends to the case where bankruptcy costs are considered;  this is in contrast to prior considerations of NVA, such as~\cite{gourieroux2012}, in which $2^n$ nonconvex regions need to be evaluated (see Online Appendix~\ref{sec:general}). 
Thus the comonotonic setting allows us to explore the network effects from an analytical perspective.  
\item Under the comonotonic framework, we are able to provide \textbf{bounds on the expectation of the system behavior} under any general random endowment. The lower bound assumes particular importance from a stress-testing perspective and for assessing systemic risk in a financial network.  While it may appear contradictory that banks would choose the riskiest scenario (as considered in the portfolio optimization motivation above and revisited specifically under the Eisenberg-Noe network setting in Corollary~\ref{cor:comonotonic}) the upside benefits for individual bank and sector equities outweigh, for the institutions themselves, the downside systemic risks for the system.
The upper bound, by conditioning on a systematic factor, also follows a comonotonic endowment setting and thus the computational improvements can be applied for that bound as well.
\item We provide special consideration to the analytical bounds for the setting in which the shocks to the banking sector can be decomposed into a systematic component felt (heterogeneously) by all firms and an idiosyncratic component. The lower bound, in particular, can be used to \textbf{perform stress-testing} and assess the health of the system. In doing so, we are able to quantify the effects of diversity of investment strategies on pricing and, therefore, also systemic risk.  By deconstructing the returns of any bank into the systematic or market component and the idiosyncratic component, we can recover the tradeoffs between systematic risk and idiosyncratic risks on systemic risk.  
This so-called diversity versus diversification problem is well-studied in the price-mediated contagion literature; we refer the interested reader to, e.g.,~\cite{CW19,detering2020suffocating}.
\item From an application standpoint, we wish to note the \textbf{comparison of the NVA setting to that taken in the single-firm setting} by \cite{merton1974} (see, also, Online Appendix~\ref{sec:merton}). Numerical case studies are presented to study the comparative statics for the performance of the system with respect to important system parameters and highlight the difference with respect to \cite{merton1974}.  These illustrative exercises permit us to study the extent to which interbank networks impact prices through direct comparison to equivalent balance sheets without the network contagion effects.  In fact, we construct \textbf{two baseline heuristic balance sheets without financial networks} in Section~\ref{sec:statics} which approximate the network effects on the price of debt and equity.
\end{enumerate}
Though not undertaken in this work, our results allow for consideration of financial networks as in \cite{AOT13} to compare the health and stability of various network topologies to random shocks.  In contrast to \cite{AOT13}, this setting allows for the formulation of stability results under any network topology and not only the two stylized networks (ring and completely connected) undertaken in that work.  We wish to highlight that, though the comonotonic setting may appear restrictive for this purpose, \cite{AOT13} imposes i.i.d.\ shocks to symmetric systems in order to obtain analytical results. 
As far as the authors are aware, one other work has analytically studied the stability and resilience of general network topologies; that work, \cite{amini2021optimal}, is based on the results presented herein in order to solve the optimal network compression problem under systematic shocks.

\section{Setting}\label{sec:setting}
We begin with some simple notation that will be consistent for the entirety of this paper.  Let $x,y \in \bbr^n$ for some positive integer $n$, then 
\[x \wedge y = \left(\min(x_1,y_1),\min(x_2,y_2),\ldots,\min(x_n,y_n)\right)^\T,\] $x^- = -(x \wedge 0)$, and $x^+ = (-x)^-$.  Further, to ease notation, we will denote $[x,y] := [x_1,y_1] \times [x_2,y_2] \times \ldots \times [x_n,y_n] \subseteq \bbr^n$ to be the $n$-dimensional compact interval for $y - x \in \bbr^n_+$.  Similarly, we will consider $x \leq y$ if and only if $y - x \in \bbr^n_+$.  We will also make wide use of the vectors $e_j \in \{0,1\}^n$ for $j = 1,2,...,n$.  This vector is defined so that $e_j$ has a 1 in its $j^{th}$ element and 0 in all other elements.

\subsection{Financial networks}\label{sec:EN}
Throughout this paper we will consider a network of $n$ financial institutions.  
Often we will consider an additional node $n+1$, which encompasses the entirety of the financial system outside of the $n$ banks; this node $n+1$ will also be referred to as society or the societal node. 
We refer to \cite{feinstein2014measures,GY14} for further discussion of the meaning and concepts behind the societal node. 

In this paper, we consider obligations with a single maturity date, as considered in \cite{EN01}. Any bank $i \in \{1,2,...,n\}$ may have obligations $L_{ij} \geq 0$ to any other firm or society $j \in \{1,2,...,n+1\}$.  We will assume that no firm has any obligations to itself, i.e., $L_{ii} = 0$ for all firms $i \in \{1,2,...,n\}$, and the society node has no liabilities at all, i.e., $L_{n+1,j} = 0$ for all firms $j \in \{1,2,...,n+1\}$.  Thus the \emph{total liabilities} for bank $i \in \{1,2,...,n\}$ is given by $\bar p_i := \sum_{j = 1}^{n+1} L_{ij} \geq 0$ and \emph{relative liabilities} from bank $i \in \{1,2,...,n\}$ to bank $j \in \{1,2,...,n\}$ is given by $\pi_{ij} := \frac{L_{ij}}{\bar p_i}$ if $\bar p_i > 0$ and arbitrary otherwise; for simplicity, in the case that $\bar p_i = 0$, we will let $\pi_{ij} = 0$ for all $j \in \{1,2,...,n\}$.  Note that, for any firm $i$, we recover the property that $\sum_{j = 1}^n \pi_{ij} \leq 1$.  Throughout this work we will consider the square matrix $\Pi \in [0,1]^{n \times n}$; the relative liabilities from firm $i$ to the societal node $n+1$ can be defined as being $1 - \sum_{j = 1}^n \pi_{ij} \geq 0$. On the other side of the balance sheet, all firms are assumed to begin with some endowments $x_i \geq 0$ for all $i \in \{1,2,...,n\}$. 

The central question explored in the network models is the determination of the firm wealths after network clearing. Let the clearing wealths be given by $V \in \bbr^n$. In this paper to determine the clearing wealths, we assume the following stylized rules, in adherence to the standard literature, i.e.\ \cite{EN01,RV13}: 
\begin{enumerate}
\item \emph{Limited liabilities}: the total payment made by any firm will never exceed the total assets available to the bank.
\item \emph{Priority of debt claims}: the shareholders of a firm receive no value unless all its debts are paid in full.
\item \emph{All debts are of the same seniority}: in case a bank defaults, debts are paid out in proportion to the size of the nominal claims.
\end{enumerate}
Throughout this work we consider a system with some exogenous recovery rates in case of default, i.e.\ the model proposed in~\cite{RV13}. This means if bank $i$ has negative wealth $V_i < 0$ then it is defaulting and its assets are reduced with recovery rates $\alpha_x \in [0,1]$ on its external assets and $\alpha_L \in [0,1]$ on its interbank assets.

We will briefly define this setting mathematically, the details can be found in~\cite{EN01,RV13} and are replicated for our discussion in Online Appendix~\ref{sec:network}.
With the rules set, we formalize the clearing process $\Psi^*: \bbr^n \to \bbr^n$ in wealths to describe this system as
\begin{align}
\label{eq:wealth2} \Psi^*_i(V) &:= \left(\ind{V_i \geq 0} + \ind{V_i < 0}\alpha_x\right) x_i + \left(\ind{V_i \geq 0} + \ind{V_i < 0}\alpha_L\right) \sum_{j = 1}^n \pi_{ji} (\bar p_j - V_j^-) - \bar p_i
\end{align}
As such, the clearing procedure $\Psi^*$ implies: if bank $i$ has nonnegative wealth $V_i \geq 0$ then it is solvent and its wealth is equal to its total assets minus its total liabilities; if bank $i$ has negative wealth $V_i < 0$ then it is defaulting and its assets are reduced by the recovery rates $\alpha_x,\alpha_L$.  From~\cite{RV13}, we immediately recover a greatest and least clearing solution to $V = \Psi^*(V)$ within the lattice $[-\bar p , x + \Pi^\T \bar p]$.
We note that with $\alpha_x = \alpha_L = 1$ (i.e.\ under no bankruptcy costs) we recover the model of~\cite{EN01}; if additionally all firms have obligations to the societal node $n+1$ (i.e.\ $\sum_{j = 1}^n \pi_{ij} < 1$ with $\bar p_i > 0$ for all firms $i$), then there exists a unique clearing solution $V = \Psi^*(V)$ in this setting. 

\begin{assumption}\label{ass:network}
For the remainder of this paper we will assume that all firms have obligations to the societal node $n+1$ (i.e.\ $\sum_{j = 1}^n \pi_{ij} < 1$ with $\bar p_i > 0$ for all firms $i$).
\end{assumption}

Throughout this work we will focus on the greatest clearing wealths solution $V^\uparrow$.  We choose this equilibrium as all firms and regulators, if given the choice, would prefer these clearing wealths to all others as no firm can improve on their performance beyond that given by $V^\uparrow$.  We wish to note that if the least clearing wealths were desired instead, all subsequent results of this paper would follow comparably.

\begin{definition}\label{defn:maximal}
Define the mapping $V: \bbr^n_+ \to \bbr^n$ so that $V(x)$ is the maximal clearing wealth solution under endowments $x \in \bbr^n_+$.
Further, define $x \mapsto p(x) := \bar p - V(x)^-$ and $x \mapsto E(x) := V(x)^+$ to be the associated payments and equity.
\end{definition}

\subsection{Comonotonicity}\label{sec:comonotone}
For the remainder of this paper, we will consider a probability space $(\Omega,\fcal,\P)$.  Denote by $L^0 := L^0(\Omega,\fcal,\P)$ all measurable random variables.  Further, denote by $L^1 \subseteq L^0$ those random variables that have finite absolute expectation, i.e.\ $X \in L^1$ if $X$ is $(\Omega,\fcal,\P)$ measurable and $\E[|X|] < \infty$.  We will denote by $L^1_+$ those random variables that are almost surely nonnegative.

As highlighted in Section~\ref{sec:motivation}, within this work we will focus on comonotonic random vectors.
\begin{definition}[Definition 4 of~\cite{dhaene2002concept}]\label{defn:comonotonic}
$X \in (L^1)^n$ is comonotonic if it has comonotonic support, i.e.\ either $x \leq y$ or $y \leq x$ for any $x,y \in \operatorname{supp}(X)$ the smallest Borel set $A \subseteq \bbr^n$ such that $\P(X \in A) = 1$.
%$\P(X_1 \leq x_1,...,X_n \leq x_n) = \min_{i \in \{1,...,n\}} \P(X_i \leq x_i)$ for any $x \in \bbr^n$.
\end{definition}
Rather than using the formal definition of comonotonicity given above, we will focus on a widely utilized equivalent formulation which comes from, e.g., Proposition 7.18 of~\cite{MFE15}.
\begin{proposition}[Proposition 7.18 of~\cite{MFE15}]\label{prop:comonotonic-defn}
$X \in (L^1)^n$ is comonotonic if and only if $X \overset{(d)}{=} f(q)$ (i.e.\ equal in distribution) for some random variable $q \in L^0$ and $f: \bbr \to \bbr^n$ nondecreasing.
\end{proposition}

\section{Expectations of debt and equity under comonotonic endowments}\label{sec:price}
Consider the motivation expressed in Section~\ref{sec:motivation} for a financial system with comonotonic endowments.  With the justification that banks would choose comonotonic endowments in theory as well as the data to support this being (approximately) accurate, a general endowment space is not necessary for understanding systemic risk and financial contagion.  As such, in this section we will consider comonotonic endowments. 
We will present herein the expectations and probability distributions under general comonotonic endowments. We wish to emphasize that in this paper, we consider a single maturity model along the lines of \cite{EN01,RV13}, i.e.\ the network is formed and fixed at time $0$ and all claims mature (and solvency is determined) at time $T > 0$.

Though all results of this section are presented as the expectation of clearing solutions under the probability measure $\P$, we consider these as pricing formulas.  That is, up to discounting, we view the term $\E[p_i(X)]$ as the price of debt and $\E[E_i(X)]$ as the price of equity for bank $i$ with system endowments $X \in (L^1_+)^n$.  In Section~\ref{sec:statics}, we explicitly consider a market model with risk-neutral measure $\Q$ under which prices can be given.
Throughout this work we, additionally, consider an effective interest rate as an equivalent measure for the price of debt; this interest rate $R_i(X)$ for bank $i$, under some pricing measure $\Q$, is such that $\bar p_i e^{-R_i(X) T} = \EQ[p_i(X)]e^{-rT}$ with risk-free rate $r \geq 0$ and network maturity $T > 0$. 
\begin{definition}\label{defn:R}
Consider a financial network with maturity $T > 0$ and a market with risk-free rate $r \geq 0$ and prices determined by the probability measure $\Q$.  Define the \emph{effective interest rate} on firm $i$'s debt by:
\begin{equation}\label{eq:interest}
R_i(X) = \frac{1}{T}\left[\log(\bar p_i) - \log(\EQ[e^{-rT}p_i(X)])\right].
\end{equation}
\end{definition}
For bank $i$, we can define $R_i(X)-r$ as the risk premium along the lines of \cite{merton1974}. Note that either the effective interest rate or the risk premium is often taken as \emph{the} measure of the price of debt due to the monotonic relation between the expectation and the interest rate (see, e.g.,~\cite{merton1974}).

\begin{assumption}\label{ass:comonotonic}
Throughout this section we will restrict our consideration to comonotonic nonnegative random vectors of endowments $X \in (L^1_+)^n$ that are equal in distribution to $f(q)$ for some nonnegative random variable $q \in L^0_+$ and nondecreasing map $f: \bbr_+ \to \bbr^n_+$.
\end{assumption}

\subsection{Piecewise linear formulation of clearing wealths}\label{sec:piecewise}
From the fictitious default algorithm of~\cite{RV13} (see Corollary~\ref{cor:FDA}) we are able to give a linear construction for the clearing vector provided the defaulting set is known.  This is given by the following construction.  We compare this linear structure to the directional derivative proposed in~\cite{LS10} and the ``network multipliers'' from~\cite{CLY14} when considering only the model of~\cite{EN01}, i.e.\ with full recovery ($\alpha_x = \alpha_L = 1$).
For the remainder of this paper we will use the following definitions:
\begin{align}
\label{eq:Delta} \Delta(z) &:= \left(I - \left(I - (1-\alpha_L)\diag(z)\right)\Pi^\T \diag(z)\right)^{-1} \left(I - (1-\alpha_x)\diag(z)\right)\\
\label{eq:delta} \delta(z) &:= \left(I - \left(I - (1-\alpha_L)\diag(z)\right)\Pi^\T \diag(z)\right)^{-1} \left[I - \left(I - (1-\alpha_L)\diag(z)\right)\Pi^\T\right]\bar p
\end{align}
for $z \in \{0,1\}^n$ denoting the set of defaulting institutions as in the fictitious default algorithm.
Thus 
\begin{align*}
V(x) := \Delta(\ind{V(x) < 0})x - \delta(\ind{V(x) < 0})   
\end{align*} 
for any endowment $x \in \bbr^n_+$ by construction.

\begin{remark}\label{rem:crossownership}
The results of this work can be compared to those from~\cite{gourieroux2012} with cross-ownership of equity.  In that setting the ownership of equity is denoted by $\Gamma \in [0,1]^{n \times n}$ where bank $j$ owns $\gamma_{ij}$ of bank $i$'s equity.  Under the assumption that no firm has sold off all of its equity to other firms within the financial system, we consider only the case that $\sum_{j = 1}^n \gamma_{ij} < 1$.  Though~\cite{gourieroux2012} does not consider bankruptcy costs, we allow for these within our generalized framework through the use of the recovery rates $\alpha_x,\alpha_L \in [0,1]$ as presented above.  That is, the clearing wealths satisfy 
\[V_i = \left(\ind{V_i \geq 0} + \ind{V_i < 0}\alpha_x\right)x_i + \left(\ind{V_i \geq 0} + \ind{V_i < 0}\alpha_L\right)\sum_{j = 1}^n \left[\pi_{ji}\left(\bar p_j - V_j^-\right) + \gamma_{ji}V_j^+\right] - \bar p_i\]  
for every bank $i$.
We can replicate all results of this section -- the formulations of the expected debt payments and equity values under comonotonic endowments -- for the setting of~\cite{gourieroux2012} with bankruptcy costs by simply redefining $\Delta,\delta$ as: 
\begin{align*}
\Delta(z) &:= \left(I - \left(I - (1-\alpha_L)\diag(z)\right)\left[\Pi^\T \diag(z) + \Gamma^\T (I-\diag(z))\right]\right)^{-1}\\
    &\qquad\qquad \times \left(I - (1-\alpha_x)\diag(z)\right)\\
\delta(z) &:= \left(I - \left(I - (1-\alpha_L)\diag(z)\right)\left[\Pi^\T \diag(z) + \Gamma^\T (I-\diag(z))\right]\right)^{-1}\\
    &\qquad\qquad \times \left[I - \left(I - (1-\alpha_L)\diag(z)\right)\Pi^\T\right]\bar p
\end{align*}
without any other modification.
The comonotonic cases that we consider in this work would solve the curse of dimensionality issues that exists in the work of~\cite{gourieroux2012}.  We wish to highlight that, though the comonotonic endowment setting presented in this section can be approached with equity cross-ownership, the bounds presented in Section~\ref{sec:bounds} do \emph{not} hold for this more general setting.
\end{remark}

{
\subsection{Defaulting regions}\label{sec:defaulting}
}
For any random endowment satisfying Assumption~\ref{ass:comonotonic}, we can now consider the defaulting regions, i.e.\ the regions in which different combinations of banks are deemed to be defaulting on a portion of their liabilities.  In fact, under the comonotonic setup considered herein, all such regions in the $q$-space must be convex intervals in $\bbr_+$.  This is in contrast to a general endowment space in which the regions need not be convex if the recovery rate is strictly less than 1 ($\min\{\alpha_x,\alpha_L\} < 1$); we refer to~\cite{gourieroux2012} and Online Appendix~\ref{sec:general} for more on this analysis.  Thus, with the comonotonicity assumption we can uniquely define the regions of $q$ under which different firms default, which we will do so with the vector $q^* \in \bbr^n_+$.  {Notably, this construction of the defaulting regions is fully characterized by the monotonic map $f$ defining the comonotonic endowments without consideration for the underlying random variable.}
\begin{definition}\label{defn:q*}
Fix some random endowments satisfying Assumption~\ref{ass:comonotonic}.  Define $q^* \in \bbr^n_+$ so that $q_i^*$ is the minimal value such that firm $i$ is solvent, i.e.\
\[q_i^* = \inf\left\{q \geq 0 \; | \; V_i(f(q)) \geq 0\right\}.\]
\end{definition}
{A discussion on the computation of $q^*$ is provided within Online Appendix~\ref{sec:q*}.}
The values of $q^*$ have a clear financial meaning if the construction $f$ is considered to be a single factor model.  This is expanded upon in Remark~\ref{rem:sfm}. 
{Specifically, the value $q_i^*$ is, in some sense, indicative of the financial stability of bank $i$.}
\begin{assumption}\label{ass:q*order}
Without loss of generality we will assume for the remainder of this paper (except where explicitly mentioned otherwise) that the banks are placed in descending order of $q^*$, i.e.\ so that $q_1^* \geq q_2^* \geq ... \geq q_n^*$.  Additionally, define $q_0^* = \infty$ and $q_{n+1}^* = 0$.
\end{assumption}

\subsection{Price of debt and equity}
Immediately with {the} construction of minimal values $q^*$ for which each firm is solvent {(as defined in Section~\ref{sec:defaulting} and with computation provided in Online Appendix~\ref{sec:q*})}, we are able to deduce formulations for the defaulting probabilities for each bank as well as the expectations of the wealth, payments, and equity for each firm; such a construction of minimal threshold prices can only exist in the comonotonic setting assumed herein.  As such, and in contrast to the general formulation in~\cite{gourieroux2012} which requires a partition of the endowment space into $2^n$ defaulting regions, the formulations given in the below theorem require only a partition into $n$ intervals due to comonotonicity.  
Further, as demonstrated in Lemma~\ref{lemma:bound} below, in the setting of~\cite{RV13}, this formulation provides a tractable bound on the general expectations given in, e.g.,~\cite{gourieroux2012,barucca2016valuation} for large networks. 
\begin{theorem}\label{thm:comonotonic}
Let the endowments be defined by $X = f(q)$ satisfy Assumption~\ref{ass:comonotonic}.
The probability of default, the expected wealth, the expected payment, and the expected equity for firm $i$ are given, respectively, by:
\begin{align*}
\P(V_i(X) < 0) &= \P(q < q_i^*)\\
\E[V_i(X)] &= e_i^\T \sum_{k = 0}^n \left[\Delta_k \E[f(q)\ind{q \in [q_{k+1}^*,q_k^*)}] - \delta_k \P(q \in [q_{k+1}^*,q_k^*))\right]\\
\E[p_i(X)] &= \bar p_i + e_i^\T \sum_{k = i}^n \left[\Delta_k \E[f(q)\ind{q \in [q_{k+1}^*,q_k^*)}] - \delta_k \P(q \in [q_{k+1}^*,q_k^*))\right]\\
\E[E_i(X)] &= e_i^\T \sum_{k = 0}^{i-1} \left[\Delta_k \E[f(q)\ind{q \in [q_{k+1}^*,q_k^*)}] - \delta_k \P(q \in [q_{k+1}^*,q_k^*))\right]
\end{align*}
where we define $\Delta_k$ and $\delta_k$ by:
\begin{align*}
\Delta_k &:= \begin{cases}\Delta\left(\sum_{j = 1}^k e_j\right) &\text{if } k = 1,2,...,n\\ I &\text{if } k = 0\end{cases} \text{ and } \delta_k := \begin{cases}\delta\left(\sum_{j = 1}^k e_j\right) &\text{if } k = 1,...,n\\ (I - \Pi^\T)\bar p &\text{if } k = 0\end{cases}
\end{align*}
{with $\Delta,\delta$ defined in~\eqref{eq:Delta} and~\eqref{eq:delta} respectively and $q^*$ as in Definition~\ref{defn:q*}.}
\end{theorem}
These expectation formulas implicitly encode the network and clearing model of~\cite{RV13} in two key points: (i) through the piecewise linear constructions for $\Delta_k,\delta_k$ and (ii) through the minimal solvency prices $q^*$.  With these components encoding the network, the various expectations are found by conditioning on all possible default sets; in this comonotonic framework that requires partitioning the $q$-space into $n$ intervals with endpoints determined by $q^*$.  In {Online Appendix}~\ref{sec:capm}, we make these formulas more explicit in a lognormal setting akin to that taken by~\cite{merton1974} (summarized in Online Appendix~\ref{sec:merton}) for a single firm only.

\begin{remark}\label{rem:R}
Theorem~\ref{thm:comonotonic} provides an explicit representation for the (discounted) price of debt $e^{-rT}\E[p_i(X)]$ and equity $e^{-rT}\E[E_i(X)]$ under comonotonic endowments $X = f(q)$ for a financial network with maturity $T > 0$, risk-free rate $r \geq 0$, and with pricing measure $\P$.  The effective interest rates $R(X)$ can, likewise, be computed explicitly through~\eqref{eq:interest}.
\end{remark}

\section{Comonotonic pricing bounds}\label{sec:bounds}

In this section we will utilize comonotonic endowments to consider upper and lower bounds on the pricing of debt with general random endowments.
We wish to note that in the following the {threshold values} $q^*$ for the comonotonic versions {(see Definition~\ref{defn:q*})} may not have a physical interpretation.  In the special case that a single factor model is considered, the $q^*$ regains a clear meaning with regards to this factor analysis.  This is expanded upon in Remark~\ref{rem:sfm} below {and considered explicitly in Online Appendix~\ref{sec:capm}}.

\subsection{General upper and lower bounds for wealth and payments}\label{sec:gen-bounds}
We now wish to consider how the formulas above for the expectations of wealth and debt under comonotonic endowments can provide a bound for the more general random endowments.  As previously mentioned, the expectations of debt and equity were studied in~\cite{gourieroux2012,barucca2016valuation}, but the formulations required suffer from the curse of dimensionality. More generally, if the correlations between firm endowments is unknown, the following lemma is useful from a stress-test viewpoint as we find that the comonotonic case is a lower bound on the health of the system.  We wish to note that these bounds are to be considered mathematically and not as a statement on portfolio rebalancing. {In Online Appendix~\ref{sec:strict-bound}, we demonstrate that these bounds can be binding.}
\begin{lemma}\label{lemma:bound}
Let $X \in (L^1_+)^n$ and $Z$ be its comonotonic copula, i.e., $Z = (F_{X_1}^{-1}(U),...,F_{X_n}^{-1}(U))$ for uniform random variable $U$ on the support $[0,1]$ and marginal distributions $F_{X_1},...,F_{X_n}$ for $X_1,...,X_n$ respectively.\footnote{For non-continuous distributions we define $F_{X_i}^{-1}(u) := \inf\{x_i \in \bbr_+ \; | \; F_{X_i}(x_i) \geq u\}$ as in~\cite{MFE15}.}  
Further, let {$\bbg = \left\{\gcal \subseteq \fcal \; | \; \gcal \text{ is a $\sigma$-algebra}, \; \E[X \; | \; \gcal] \text{ is comonotonic}\right\}$ be the set of sub-$\sigma$-algebras such that $\E[X \; | \; \gcal]$ is a comonotonic projection of $X$.}  Then
\begin{align*}
\E[V_i(\alpha_x Z;\alpha_L \Pi,\bar p,1,1)] &\leq \E[V_i(X;\Pi,\bar p,\alpha_x,\alpha_L)] \leq \inf_{{\gcal \in \bbg}} \E[V_i(\E[X \; | \; {\gcal}];\Pi,\bar p,1,1)] \leq V_i(\E[X];\Pi,\bar p,1,1),\\
\E[p_i(\alpha_x Z;\alpha_L \Pi,\bar p,1,1)] &\leq \E[p_i(X;\Pi,\bar p,\alpha_x,\alpha_L)] \leq \inf_{{\gcal \in \bbg}} \E[p_i(\E[X \; | \; {\gcal}];\Pi,\bar p,1,1)] \leq p_i(\E[X];\Pi,\bar p,1,1)
\end{align*}
for any bank $i$. 
{The bound on wealth holds, also, for the societal node where the relative liabilities owed to society are fixed by $\Pi$, i.e.\ such that the societal wealth is explicitly defined by $V_{n+1}(Y;\tilde\Pi,\bar p,\alpha_x,\alpha_L) := \sum_{i = 1}^n \pi_{i,n+1} p_i(Y;\tilde\Pi,\bar p,\alpha_x,\alpha_L)$.}
\end{lemma}
\begin{remark}\label{rem:bound-R}
As discussed previously in, e.g., Remark~\ref{rem:R}, we view the expectation $\E[p_i(X)]$ as the price of debt (up to modification by discounting).  Thus, we can view the bounds on the expectation provided in Lemma~\ref{lemma:bound} as a bound on the price of debt.  Similarly, due to the monotonic relationship between the effective interest rate $R(X)$ and the price of debt, we can determine that
\[R_i(\E[X];\Pi,\bar p,1,1) \leq \sup_{{\gcal \in \bbg}} R_i(\E[X \; | \; {\gcal}];\Pi,\bar p,1,1) \leq R_i(X;\Pi,\bar p,\alpha_x,\alpha_L) \leq R_i(\alpha_x Z;\alpha_L \Pi,\bar p,1,1)\]
for every bank $i$ utilizing the notation from Lemma~\ref{lemma:bound}.
\end{remark}

The conditional upper bounds with endowments $\E[X \; | \; {\gcal}]$ for some {sub-$\sigma$-algebra $\gcal \in \bbg$} in Lemma~\ref{lemma:bound} are comonotonic.  This means that each possible objective value over which we are optimizing is computationally tractable via the formulations in Theorem~\ref{thm:comonotonic}. {Though we present this upper bound as the infimum over all such comonotonic projections, in practice it can be challenging to determine the set of all appropriate sub-$\sigma$-algebras.  As presented in Lemma~\ref{lemma:bound}, the trivial $\sigma$-algebra is always an element of $\bbg$ which provides a bound to this infimum.}  In Remark~\ref{rem:sfm} we consider {another} specific, financially meaningful, choice of {sub-$\sigma$-algebra} when considering a single factor model.

We wish to briefly discuss some intuition surrounding Lemma~\ref{lemma:bound} before continuing.  When banks have interbank relationships, these assets serve as a method for diversification.  However, if banks are all subject to a common shock, this diversification, in effect, disappears and the interconnectedness becomes irrelevant from a risk perspective.  From a computational standpoint, we construct upper bounds for the expected wealths and payments via the comonotonic setting as well, and thus this is not only useful from a worst-case scenario. {}

\begin{remark}\label{rem:sfm}
For simplicity of exposition, consider the full recovery setting $\alpha_x = \alpha_L = 1$ of~\cite{EN01}.
Consider a single factor model for endowments that allows for errors or idiosyncratic terms in which every firm is long the factor: $X = f(q,\epsilon)$ where $f$ and $q$ are as in Assumption~\ref{ass:comonotonic} and $\epsilon \in (L^0)^n$ is independent from $q$ ($f$ need not be monotonic in $\epsilon$).  The conditional upper bound considered in Lemma~\ref{lemma:bound} implies $\E[V_i(X)] \leq \E[V_i(\E[X | q])]$ and $\E[p_i(X)] \leq \E[p_i(\E[X | q])]$ for all banks $i$ ({i.e., with $\gcal = \sigma(q)$}).  We wish to note that this upper bound is again with respect to comonotonic endowments and, thus, can be computed via the formulations in Theorem~\ref{thm:comonotonic}. {Such a setting is of particular interest when $q$ is a ``systematic factor'' and $\epsilon$ defines idiosyncratic risks; we refer the interested reader to Online Appendix~\ref{sec:capm} for a detailed example of this setting.}
\end{remark}

\begin{remark}\label{Acemoglu}
The lower bound discussed in this section assumes particular importance from a stress-testing perspective and for assessing systemic risk in a financial network. This can be used to construct measures of stability and resilience as discussed in \cite{AOT13}. However, the results of stability and resilience discussed in \cite{AOT13} are only proven for ring and completely connected network. In contrast our comonotonic setting allows for these results to be extended (analytically) to any general network topology. {These bounds can, furthermore, be applied to both scalar systemic risk measures~(\cite{chen2013axiomatic,kromer2013systemic}) and set-valued ones~(\cite{feinstein2014measures,AR16}) via the robust representation; this is detailed in Online Appendix~\ref{sec:sysrisk}.}
\end{remark}

\subsection{Price bounds for the market capitalization}\label{sec:equity}
Intriguingly, though the payments and wealth attain their worst-case under comonotonic endowments and best case under expected endowments, the equity of the different firms in the financial system typically is higher under the comonotonic endowments than other correlation structures and lower under the expected endowments.  (We wish to note that the societal node would exhibit the same bounding properties as given in Lemma~\ref{lemma:bound} since its equity is equal to its wealth by construction.)  In fact, as proven in Corollary~\ref{cor:bound} below, the total market capitalization of the financial sector is bounded in the reverse order from those given in Lemma~\ref{lemma:bound}.
\begin{corollary}\label{cor:bound}
%Consider the setting of~\cite{EN01}, i.e.\ $\alpha_x = \alpha_L = 1$.  
Consider the notation and setting of Theorem~\ref{lemma:bound}.
%That is, let $X \in (L^1_+)^n$ and $Z = (F_{X_1}^{-1}(U),...,F_{X_n}^{-1}(U))$ for uniform random variable $U$ on the support $[0,1]$ and marginal distributions $F_{X_1},...,F_{X_n}$ for $X_1,...,X_n$ respectively. %\footnote{For non-continuous distributions we define $F_{X_i}^{-1}(u) := \inf\{x_i \in \bbr_+ \; | \; F_{X_i}(x_i) \geq u\}$ as in~\cite{MFE15}.}  
Then the total expected equity of the financial sector is bounded from above and below by the comonotonic endowment setting, i.e., 
\begin{align*}
\sum_{i = 1}^n &E_i(\alpha_x \E[X];\alpha_L \Pi,\bar p,1,1) \leq \sup_{{\gcal \in \bbg}}\sum_{i = 1}^n \E[E_i(\alpha_x\E[X \; | \; {\gcal}];\alpha_L\Pi,\bar p,1,1)]\\
    &\leq \sum_{i = 1}^n \E[E_i(X;\Pi,\bar p,\alpha_x,\alpha_L)] \leq \sum_{i = 1}^n \E[E_i(Z;\Pi,\bar p,1,1)].
\end{align*}
\end{corollary}

In fact, as in, e.g.,~\cite{acharya2009theory,elliott2019systemic}, these market capitalization bounds allow us to consider the portfolio choice problem endogenously in our network framework with full recovery $\alpha_x = \alpha_L = 1$, i.e., that of~\cite{EN01}.\footnote{{As the upper bound for market capitalization is considered with full recovery, the below results are only guaranteed to hold in that setting.}}  We undertake this in order to reiterate (a different) portfolio optimization motivation of the comonotonic endowment framework.  In order to study this problem, 
consider banks who seek to maximize their own expected equity in the worst-case subject to some budget and regulatory constraints $\bcal_i \subseteq L^1_+$, i.e., bank $i$ seeks to solve the minimax problem relative to all other institutions
\begin{equation}\label{eq:max-equity}
\sup_{X_i \in \bcal_i} \inf_{X_{-i} \in \prod_{j \neq i} \bcal_j} \E[E_i(X_i,X_{-i})].
\end{equation}
We will assume that $\bcal_i$ is law-invariant for every bank $i$, i.e., if $X_i \in \bcal_i$ and $X_i \overset{(d)}{=} Y_i$ (equal in distribution) then $Y_i \in \bcal_i$ as well.
Consider the setting in which banks can coordinate and jointly strategize to maximize the total equity of the coalition whose market capitalization is then shared, i.e., coalition $\ccal$ seeks to optimize the modified minimax problem
\[v(\ccal) := \sup_{X_\ccal \in \prod_{i \in \ccal} \bcal_i} \inf_{X_{-\ccal} \in \prod_{j \not\in \ccal} \bcal_j} \sum_{i \in \ccal} \E[E_i(X_\ccal,X_{-\ccal})].\]
The below corollary proves that the grand coalition will invest comonotonically and is stable against individual defections and, thus, the banks would endogenously strategize to invest in the comonotonic endowment space.
\begin{corollary}\label{cor:comonotonic}
Consider a financial system with full recovery (i.e., $\alpha_x = \alpha_L = 1$) in which all banks are equity maximizers according to~\eqref{eq:max-equity} in which they can coordinate investment strategies.
The Shapley value 
\[\phi_i := \sum_{\ccal \subseteq \ncal\backslash\{i\}} \frac{|\ccal|!(n-|\ccal|-1)!}{n!}\left[v(\ccal \cup \{i\}) - v(\ccal)\right], \quad \forall i = 1,2,...,n\]
of the grand coalition $\ncal {:= \{1,2,...,n\}}$ is an imputation, i.e., it is an efficient equity sharing arrangement ($\sum_{i = 1}^n \phi_i = v(\ncal)$) which is individually rational ($\phi_i \geq v(\{i\})$ for every bank $i$) so that no bank will unilaterally choose to leave the grand coalition under this arrangement.
Furthermore, all banks in the grand coalition will invest comonotonically. 
\end{corollary}
As a result of this corollary, though we consider general space of random endowments to be bounded by comonotonic endowments, the comonotonic setting can arise endogenously in an Eisenberg-Noe framework due to a simple equity maximization principle.

\section{Comparative statics}\label{sec:statics}
In this section we provide the comparative statics for the performance of the system with respect to important system parameters through numerical examples. In these numerical examples we will assume {a} lognormal setting {as expressed in Assumption~\ref{ass:lognormal}. 
\begin{assumption}\label{ass:lognormal}
Let bank endowments be given by $X = sq$ for vector of holdings $s \in \bbr^n_+$ and such that the risky investment has \emph{lognormal} payout $q = \exp([r - \frac{\sigma^2}{2}]T + \sigma Z\sqrt{T})$ with risk-free rate $r \geq 0$, volatility $\sigma > 0$, and maturity $T > 0$ where $Z$ is some standard normal random variable. To simplify the setting further, we will assume $r = 0$ and $T = 1$ throughout this section.
\end{assumption} 
Assumption~\ref{ass:lognormal} presents a simplified setting of Online Appendix~\ref{sec:capm} (under the risk-neutral measure) such that the comonotonic condition $\rho_{iM} = 1$ with the additional restriction that $\sigma_i = \sigma_M =: \sigma$ for all firms $i$ and risk-free rate $r = 0$.}  We utilize these illustrative numerical examples to emphasize the impacts of default contagion through the interbank network by comparing the prices formulated in this work to those in~\cite{merton1974} without interbank liabilities (see, also, Online Appendix~\ref{sec:merton}).
We note that the comonotonic construction from the prior sections would allow us to consider more complicated underlying market models, e.g.\ a jump diffusion model; this would be accomplished by considering the appropriate distribution of $q$ at maturity $T$. But for simplicity and due to its use in the seminal work by Merton, we will restrict ourselves to the lognormal setting.

As mentioned above, in the following numerical examples we wish to compare the pricing of debt and equity with network adjustments to the formulation of~\cite{merton1974}, which is presented also in Online Appendix~\ref{sec:merton}.  In particular, we wish to study two simple baseline heuristics in the single bank Merton setting which we numerically find approximate the prices of debt and equity (respectively) in markets with high recovery $\alpha_x,\alpha_L \approx 1$ rates.
\begin{enumerate}
\item\label{baseline:risky} \textbf{Risky approximation for debt}: Assume all interbank assets -- reduced by the recovery rate in some way (herein we choose to consider $\frac{1}{2}[\alpha_x + \alpha_L]$) -- are fully invested in risky assets (i.e.\ exhibiting the lognormal distribution and not capped by the total obligations).  In particular, as demonstrated below, this setting provides an approximation for the price of debt of the firms for high recovery rates.  Heuristically, due to the comonotonicity assumption utilized in this work, the default of the banks occur under overlapping market conditions with payments comonotonic to the risky asset.  Therefore, taking these ideas to the extreme, any bank is defaulting only if all others are as well with those banks paying out proportionally to their obligations and to the (low) asset value; if the (future) asset value is high then these interbank payments are valued comparably, which results in full payments and no defaults.  
\item\label{baseline:riskfree} \textbf{Risk-free approximation for equity}: Assume all interbank assets are paid off in full in units of the risk-free asset.  In particular, as demonstrated below, this setting provides an (upper) approximation for the market capitalization of the firms.  Heuristically, due to the comonotonicity assumption utilized in this work, the solvency of the banks occur under overlapping market conditions.  Therefore, taking this idea to the extreme, any bank has positive equity only in the cases in which all other banks pay their liabilities in full.
\end{enumerate}
Notably, these heuristics have the added advantage in that they can be computed under only aggregate network information (i.e., knowing the total interbank assets) without granular network information. As expressed in the network reconstruction literature (see, e.g.,~\cite{UW04,GV16}) and seen in Section~\ref{sec:EBA}, though total interbank assets and liabilities are known, granular information on the interbank obligations is often unavailable.  As such, these heuristics take on a larger importance as they can approximate the network valuation adjustments without requiring the exact network construction.

\subsection{Two bank system}\label{sec:2bank}
\begin{figure}[h!]
\centering
\begin{tikzpicture}
\tikzset{node style/.style={state, minimum width=0.36in, line width=0.5mm, align=center}}
\node[node style] at (0,0) (x1) {Bank $1$\\$s_1 = 3$};
\node[node style] at (10,0) (x2) {Bank $2$\\$s_2 = 4$};
\node[node style] at (5,0) (x3) {Societal $3$};

\draw[every loop, auto=right, line width=0.5mm, >=latex]
(x1) edge node {$L_{13} = 3$} (x3)
(x2) edge node {$L_{23} = 3$} (x3)
(x1) edge[line width=1.167mm, bend right=20] node {$L_{12} = 7$} (x2)
(x2) edge[bend right=20] node {$L_{21} = 3$} (x1);
\end{tikzpicture}
\caption{A simple network topology for Section~\ref{sec:2bank}.}
\label{fig:2bank}
\end{figure}
Consider the financial system with 2 banks and an additional societal node as depicted in Figure~\ref{fig:2bank}. 
This two bank system with an additional societal node is such that bank $1$ owes $7$ units to bank $2$ and $3$ to the societal node and bank $2$ owes $3$ units to both bank $1$ and the societal node.
Additionally, for simplicity and where otherwise we are not varying that parameter, we consider the risky asset to have volatility $\sigma = 1$ and the claims to have maturity at time $T = 1$.  Further, {recall that} the risk-free rate is assumed to be $r = 0$.  We consider this simple, illustrative, example so as to demonstrate the effects of the financial network (in comparison to the same system in two baseline systems without interbank debt as in~\cite{merton1974}).  For a clear comparison we will take this system without bankruptcy costs ($\alpha_x = \alpha_L = 1$)\footnote{{Except where otherwise noted, the choice of bankruptcy costs do not alter any conclusions within this case study.}} and with a common risky asset exhibiting a lognormal distribution at maturity.  Specifically, we will consider the same comparative statics on the effective interest rate as undertaken by~\cite{merton1974}, i.e.\ by varying the debt-firm value ratios
and the maturity of the debt claims.  We omit a consideration of varied volatility as the conclusions from that study are directly comparable to that of modifying the maturity of the debt claims. We wish to note that since we have assumed $r=0$, the risk premium $R_i(sq) - r$ utilized by~\cite{merton1974} is equivalent to the effective interest rate herein.  We wish to highlight two key insights from this case study.  First, we wish to study the impacts that network effects have on the price of debt and equity. This comparison with the single firm setting of~\cite{merton1974} demonstrates that qualitative conclusions no longer hold in general once the network effects and systemic risk are taken into account; namely the debt-firm value ratio does \emph{not} uniquely define the price of debt and equity (Section~\ref{sec:2bank-d}) nor does it define the shape of the term structure, i.e., the dependence of the effective interest rate on maturity $T$ (Section~\ref{sec:2bank-T}). Second, we consider the performance of our aforementioned baseline heuristics to demonstrate their relevance in estimating the price of debt and equity.

\subsubsection{Impact of debt-firm value ratio}\label{sec:2bank-d}
First we will consider the impact of the debt-firm value ratios $d = \diag\left(s + \Pi^\T \bar p\right)^{-1}\bar p$ on the effective interest rate and thus the price of debt for each firm.  In~\cite{merton1974} in which no firm holds any interbank assets ($\sum_{j = 1}^n \pi_{ji} \bar p_j = 0$), it was shown that an individual firm's debt-firm value ratio can completely determine its own interest rate $R(sq)$.  However, herein we consider explicitly the effects of the interbank assets. In our case, we can vary $d$ by either:
\begin{enumerate}
\item altering the liabilities $\bar p$ and keeping investments $s$ constant, i.e.
\[\bar p = \left[\begin{array}{cc} 1 & -\pi_{21}d_1 \\ -\pi_{12}d_2 & 1 \end{array}\right]^{-1} \left[\begin{array}{c} s_1d_1 \\ s_2d_2 \end{array}\right]\]
given the desired debt-firm value ratio $d \in \bbr^2_+$ constrained by $d_1d_2 \leq \pi_{12}\pi_{21}$; or
\item altering the assets $s$ and keeping liabilities $\bar p$ constant, i.e.
\[s = \diag(d)^{-1}\left[\bar p - \diag(d)\Pi^\T \bar p\right]\]
given the desired debt-firm value ratio $d \in \bbr^2_+$ so that $d \leq \diag(\Pi^\T \bar p)^{-1} \bar p$. 
\end{enumerate}
The distinction between the two approaches to varying $d$ is important because we find that the manner in which the debt-firm value ratio is modified can greatly affect the price of debt as measured by the effective interest rate. The contour plots of the interest rate of bank 2 with respect to debt-firm 1 value ratio $d_1$ and debt-firm 2 value ratio $d_2$ are shown in Figure~\ref{fig:R_pdd2} and Figure~\ref{fig:R_xdd2} for varying the debt-firm values by altering liabilities and altering assets respectively. To provide further clarity on how the individual debt-firm values affect each other, we consider three slices of this data, by fixing the level of $d_1$ and varying $d_2$ through either altering the liabilities or the assets in Figure~\ref{fig:R_pd2} and Figure~\ref{fig:R_xd2} respectively. Notably, if firm 1 has a lower debt-firm value ratio $d_1$ constructed through the change in assets, then firm 2 consistently has a lower effective interest rate for any debt-firm ratio chosen.  However, there is no such monotonicity when the debt-firm value ratios are constructed through changes in the liabilities.
\begin{figure}[h!]
\centering
\begin{subfigure}[b]{.46\textwidth}
\includegraphics[width=\textwidth]{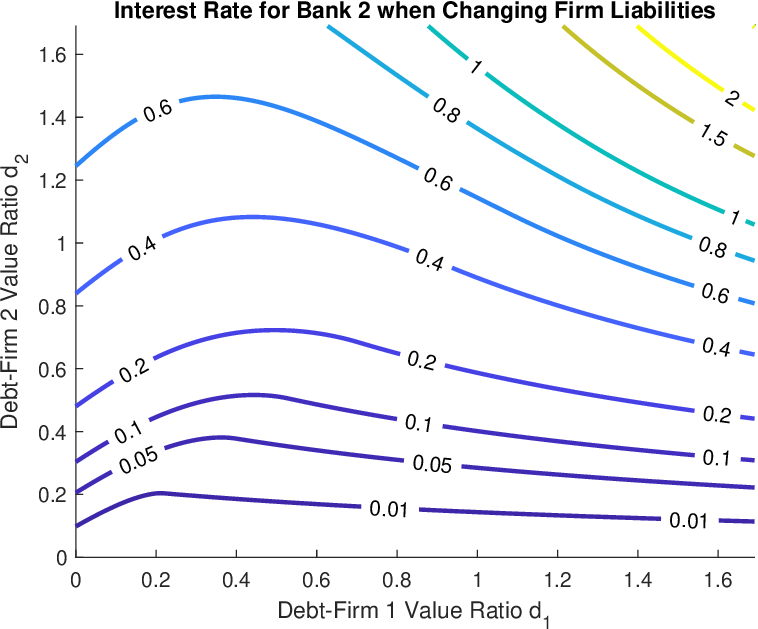}
\caption{Contour plot of the effective interest rate for firm 2 under different debt-firm value ratios when determined by the total liabilities.}
\label{fig:R_pdd2}
\end{subfigure}
~
\begin{subfigure}[b]{.46\textwidth}
\includegraphics[width=\textwidth]{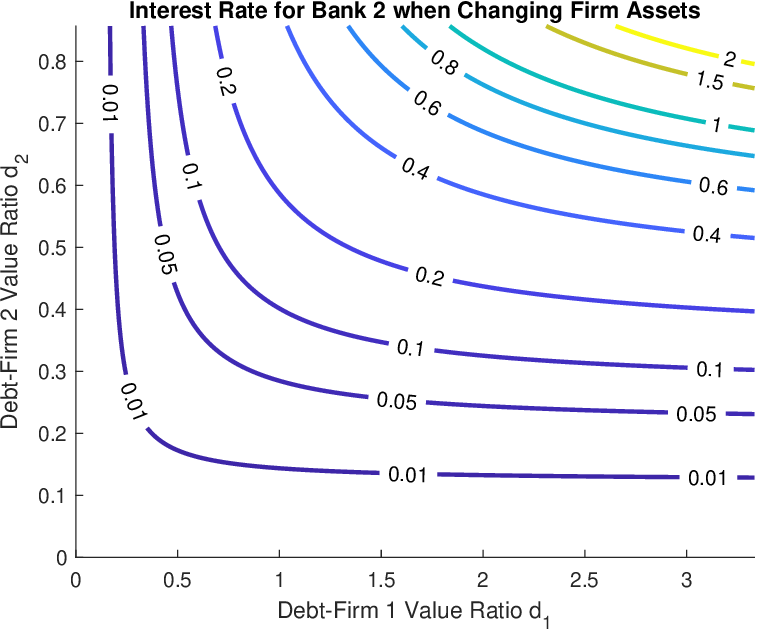}
\caption{Contour plot of the effective interest rate for firm 2 under different debt-firm value ratios when determined by the total endowments.}
\label{fig:R_xdd2}
\end{subfigure}
\\~\\
\begin{subfigure}[b]{.46\textwidth}
\includegraphics[width=\textwidth]{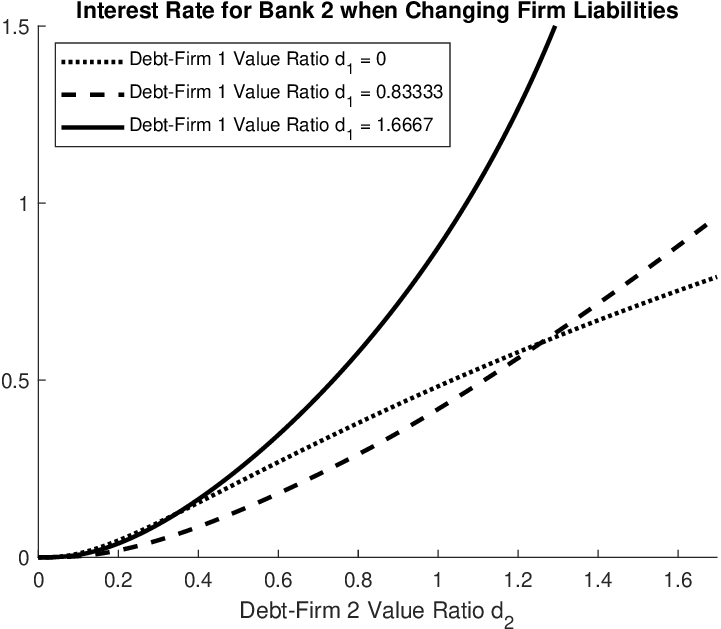}
\caption{Cross-section of the effective interest rate for firm 2 under changes in its own debt-firm value ratio and with three levels of debt-firm 1 value ratio when changes are accomplished through modifications in total obligations.}
\label{fig:R_pd2}
\end{subfigure}
~
\begin{subfigure}[b]{.46\textwidth}
\includegraphics[width=\textwidth]{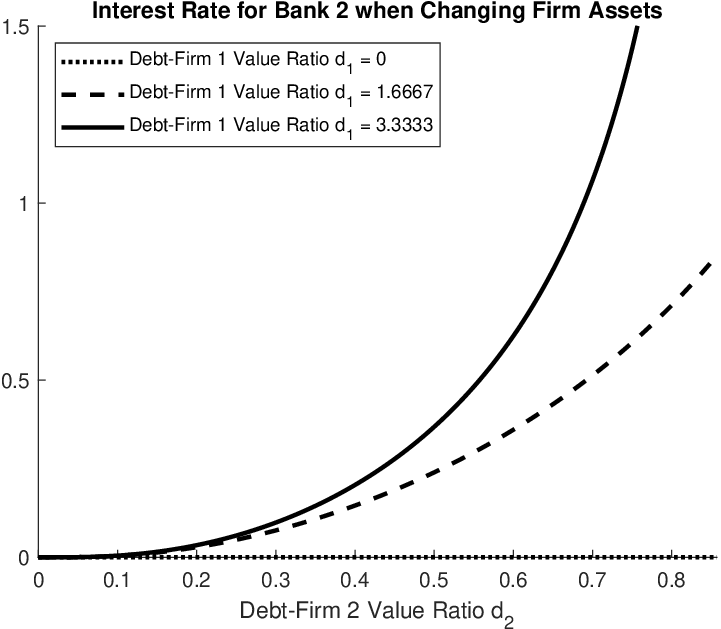}
\caption{Cross-section of the effective interest rate for firm 2 under changes in its own debt-firm value ratio and with three levels of debt-firm 1 value ratio when changes are accomplished through modifications in total assets.}
\label{fig:R_xd2}
\end{subfigure}
\caption{Section~\ref{sec:2bank-d}: The effective interest rate of firm 2 versus changes in the debt-firm value ratios.}
\label{fig:R_d}
\end{figure}

\subsubsection{Impact of maturity}\label{sec:2bank-T}
Second we will consider the impact of the maturity $T$ for the claims on the effective interest rate (and thus the price of debt) and the market capitalization for each firm.  
In this consideration we wish to compare the network effects with the two baseline models without a network presented above.  %That is, in order to consider the system without network effects we consider two settings: 
%\begin{enumerate}
%\item with the assumption that all interbank assets are treated no differently than other risky assets (i.e.\ exhibiting the lognormal distribution and not capped by the total obligations); and
%\item with the assumption that all interbank assets are paid off in full in units of the risk-free asset.
%\end{enumerate}
Due to the risk of the interbank assets in a network, it is clear that the effective interest rate would be higher and market capitalization lower when including the network effects than when all interbank assets are treated as in baseline model \eqref{baseline:riskfree}.  As depicted in Figure~\ref{fig:R_T1}, firm 1 has similar effective interest rate in all three scenarios, though the market capitalization is nearly double when interbank assets are treated as the market asset (baseline model \eqref{baseline:risky}).  Firm 2, as depicted in Figure~\ref{fig:R_T2}, has orders of magnitude higher effective interest rates under the network effects than if they had no counterparty risk and noticeably higher effective interest rate when full network effects are taken into account than if interbank assets are treated no differently than other risky assets (i.e.\ following the market model).  As expected due to their added risks, the network effects greatly reduce the market capitalization compared to the two single-firm scenarios considered herein.  Notably, the results found herein match the heuristic expectations we have for our 2 baseline models; however, and as depicted in Figure~\ref{fig:EBA_beta} for the next case study, these heuristics lose much of their predictive power for low recovery rates $\alpha_x,\alpha_L < 1$.

We wish to conclude by considering the shapes of the interest rates as a function of the maturity of the claims under network effects.  In~\cite{merton1974} the hyperbolic shape of firm 1's effective interest rate would only occur if its debt-firm value ratio was greater than or equal to 1; similarly the shape exhibited by firm 2's effective interest rate would only occur if its debt-firm value ratio was strictly less than 1.  However, as discussed above, the debt-firm value ratio does not have as unique a property under network effects as it did in~\cite{merton1974} without counterparty risk.  Thus we find that the change in shape need not (and in this numerical example, does not) {occur} at the individual debt-firm value ratios of 1.
\begin{figure}[h!]
\centering
\begin{subfigure}[b]{.46\textwidth}
\includegraphics[width=\textwidth]{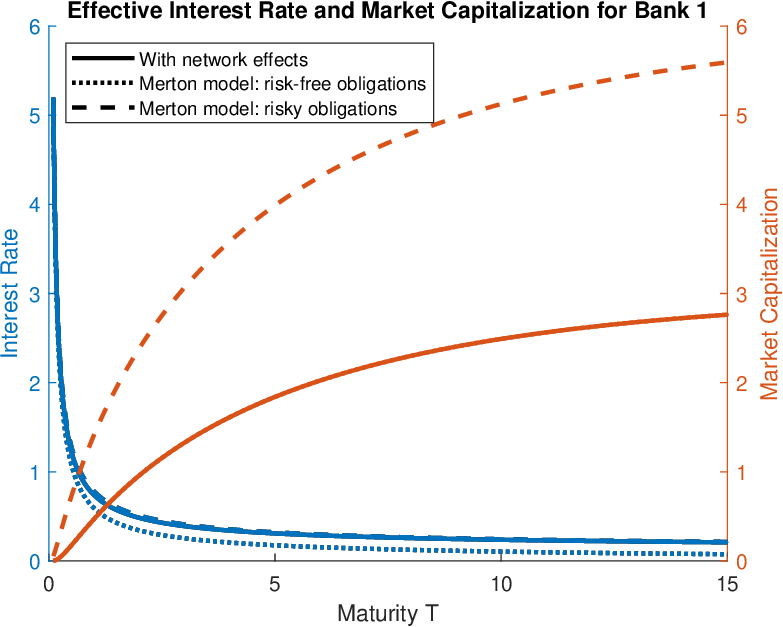}
\caption{The effective interest rate and market capitalization for firm 1 under changes to the maturity of the debt claims with network effects and under single firm effects only.}
\label{fig:R_T1}
\end{subfigure}
~
\begin{subfigure}[b]{.46\textwidth}
\includegraphics[width=\textwidth]{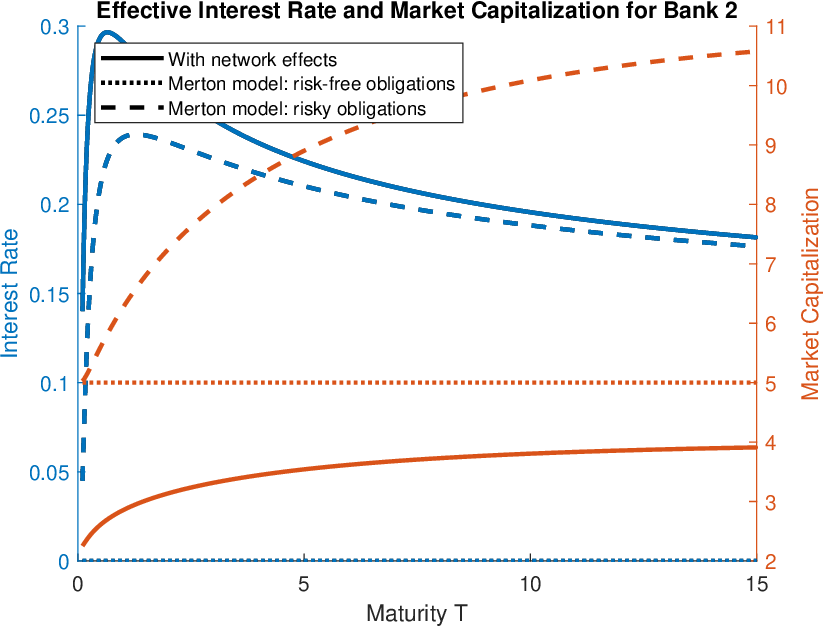}
\caption{The effective interest rate and market capitalization for firm 2 under changes to the maturity of the debt claims with network effects and under single firm effects only.}
\label{fig:R_T2}
\end{subfigure}
\caption{Section~\ref{sec:2bank-T}: The effective interest rate and market capitalization versus changes in the maturity of the debt claims.}
\label{fig:R_T}
\end{figure}

\subsection{European banking system}\label{sec:EBA}
We will now consider a larger financial network consisting of $n = 87$ banks.  This large network provides clear reasoning for considering the comonotonic approach taken within this paper.  As previously discussed, with 87 banks, there are $2^{87} > 10^{26}$ potential combinations of defaulting banks $z \in \{0,1\}^{87}$.  As such, the general framework for considering expected payments from \cite{gourieroux2012} would be computationally intractable.  However, the comonotonic framework presented herein (and which, under the setting of~\cite{EN01}, provides a worst-case for the general setting as discussed in Section~\ref{sec:gen-bounds}) is computationally tractable as only $87$ defaulting regions need to be considered.  As such, we wish to use this case study to highlight the computational tractability of the comonotonic setting as well as the performance of our baseline heuristics on a larger, more realistic, network.

For this example, we will consider these 87 banks to come from the 2011 European Banking Authority EU-wide stress tests.\footnote{Due to complications with the calibration methodology, we only consider 87 of the 90 institutions. DE029, LU45, and SI058 were not included in this analysis.}  This dataset has been used in multiple prior empirical case studies (e.g.~\cite{GV16,CLY14}) of financial contagion in interbank networks.  To calibrate this system, we will take the same approach from~\cite{feinstein2017currency} {which is provided in Online Appendix~\ref{sec:calibration}}.  We note, however, that though we are calibrating the financial network to a real dataset, the marginal distribution for bank endowments are \emph{not} calibrated and as such this example is for \emph{illustrative purposes only}.  We believe that there would be significant value in a further, detailed, case study to empirically determine the marginal distributions of the bank endowments and, with that result, consider yield rates and bond prices to compare with the realized prices in the market.  This is, however, beyond the scope of the current example.  In fact, the primary purpose of using this dataset in this example, as opposed to a large fictional network, is to demonstrate the order of magnitude that the effective interest rates (i.e., the price of debt) can achieve (in comparison to the values presented in the prior case studies on the 2 bank system).

In order to complete our model, we need to consider the remaining parameters of the system.  First, as all economic data pulled from the EBA EU-wide stress test dataset are already in a consistent unit (millions of euros), we will consider the (risk-neutral) value of the market portfolio $q$ to be $1$ (million euros).  Further, during the period over which this data was collected, central banks were setting a low interest rate environment.  Therefore we estimate that the risk-free interest rate is $r = 0$ {(as is assumed throughout this section)}.  Additionally, as this is data from a single year's stress test, we will consider maturity on all debt claims to be $T = 1$ (year).  Finally, the volatility of the risky asset is estimated to be $\sigma = 20\%$ from comparisons to annualized historical volatility of European markets in 2011.

First, we wish to consider the impact of the full network effects on the effective interest rates and market capitalization in the setting without bankruptcy costs ($\alpha_x = \alpha_L = 1$).  For this analysis we consider the same two baseline heuristic models as presented above. %, i.e.\
%\begin{enumerate}
%\item with the assumption that all interbank assets are treated no differently than other risky assets (i.e.\ exhibiting the lognormal distribution and not capped by the total obligations); and
%\item with the assumption that all interbank assets are paid off in full in units of the risk-free asset.
%\end{enumerate}
The data for these comparisons are provided in Figures~\ref{fig:R_EBA} and~\ref{fig:E_EBA} respectively.  We note that, as with our intuition and as in Section~\ref{sec:2bank} above, the price of debt with full network effects is generally comparable to the single firm effect case with all interbank assets treated as the risky asset.  In fact, the interest rate of debt with full network effects is lower than if all interbank assets are treated as the risky asset, but significantly higher than when interbank assets are treated as the risk-free asset.  In contrast, and again matching our intuition and comparable to that in Section~\ref{sec:2bank} above, the market capitalization for firms is strikingly similar between the full network effects and the single firm effects with interbank assets treated as the risk-free asset.  The single firm effects with interbank assets treated as the risky asset can differ by a large degree from the network effects for the market capitalization of the individual firms.
\begin{figure}[h!]
\centering
\begin{subfigure}[b]{.46\textwidth}
\includegraphics[width=\textwidth]{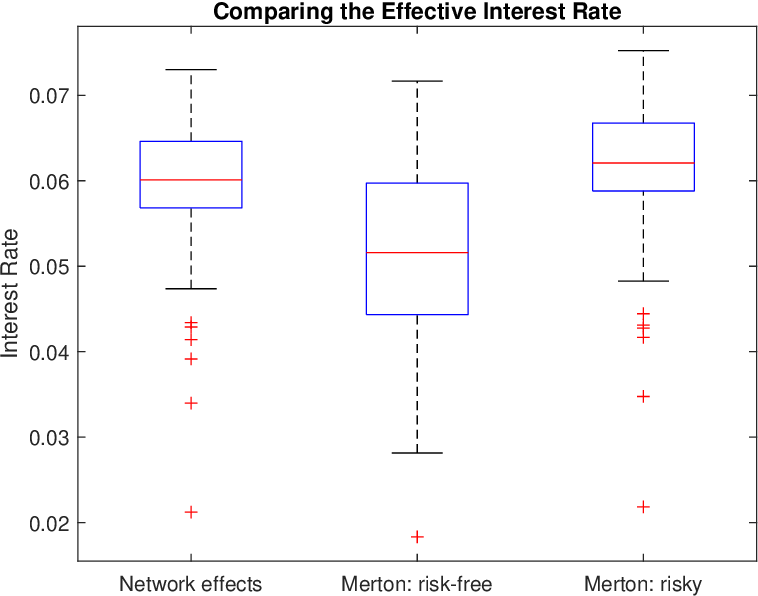}
\caption{Box plot of effective interest rates under network effects and under single firm effects only.}
\label{fig:R_EBA_box}
\end{subfigure}
~
\begin{subfigure}[b]{.46\textwidth}
\includegraphics[width=\textwidth]{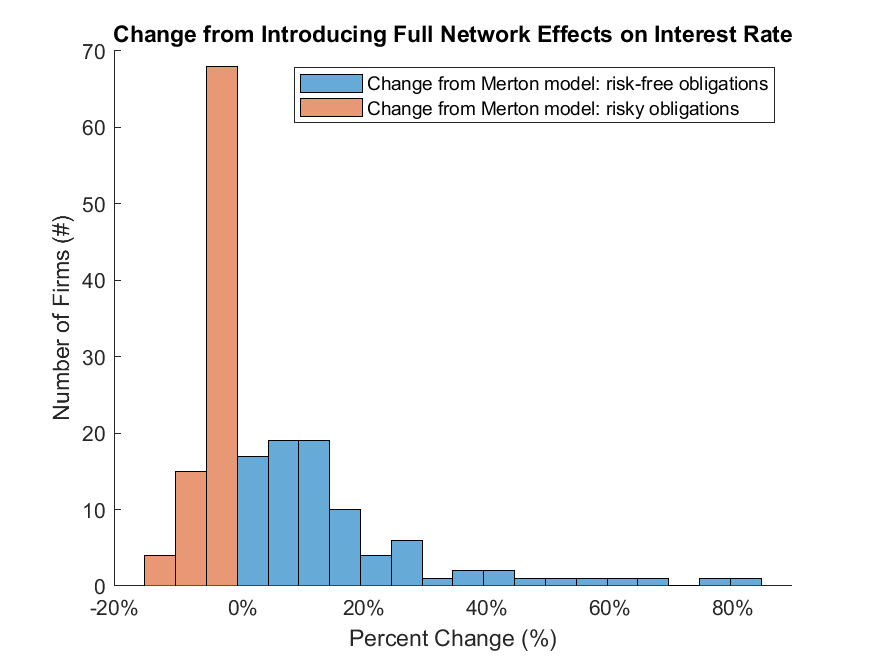}
\caption{Histogram of the relative change in the effective interest rates by including full network effects.}
\label{fig:R_EBA_change}
\end{subfigure}
\caption{Section~\ref{sec:EBA}: Comparison of the effective interest rates under network effects and under single firm effects only without bankruptcy costs ($\alpha_x = \alpha_L = 1$).}
\label{fig:R_EBA}
\end{figure}

\begin{figure}[h!]
\centering
\begin{subfigure}[b]{.46\textwidth}
\includegraphics[width=\textwidth]{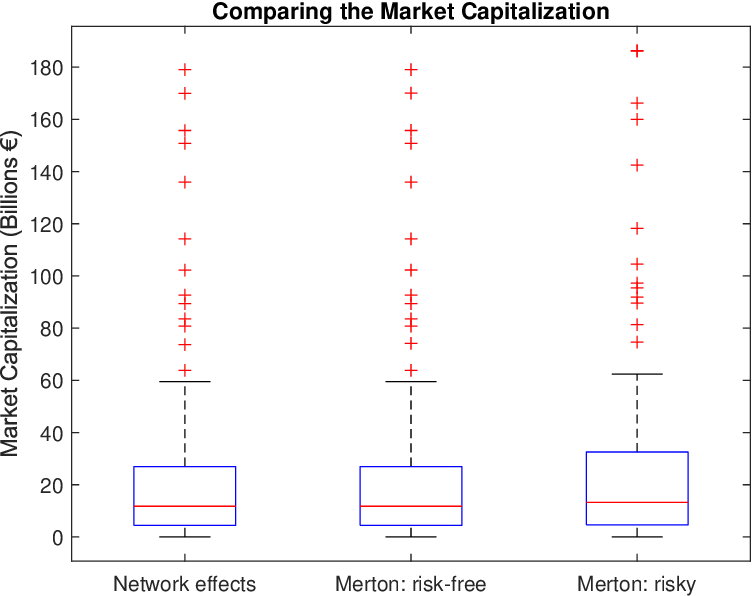}
\caption{Box plot of the market capitalization under network effects and under single firm effects only.}
\label{fig:E_EBA_box}
\end{subfigure}
~
\begin{subfigure}[b]{.46\textwidth}
\includegraphics[width=\textwidth]{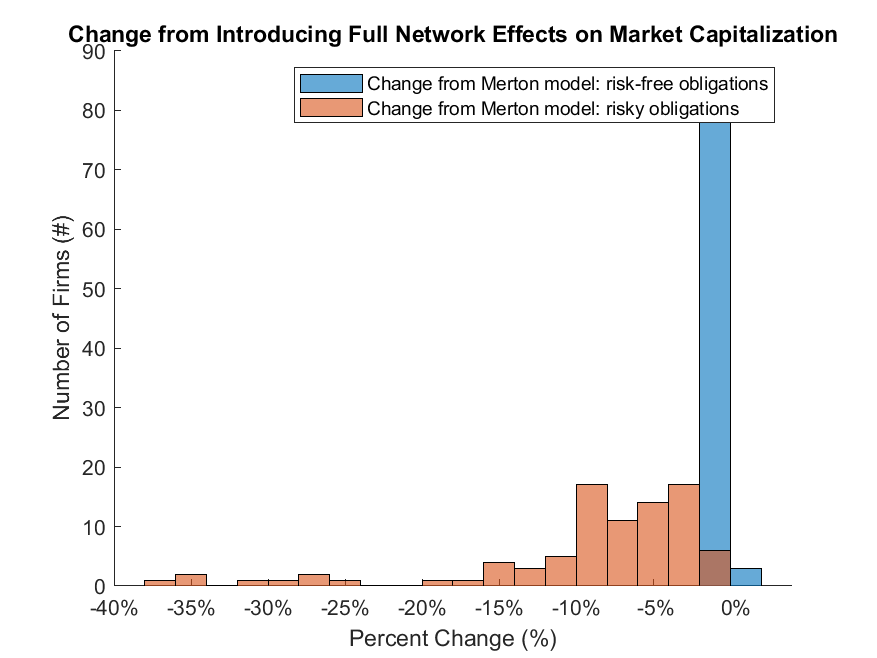}
\caption{Histogram of the relative change in the market capitalizations by including full network effects.}
\label{fig:E_EBA_change}
\end{subfigure}
\caption{Section~\ref{sec:EBA}: Comparison of market capitalization under network effects and under single firm effects only without bankruptcy costs ($\alpha_x = \alpha_L = 1$).}
\label{fig:E_EBA}
\end{figure}

Second, though above we consider the setting without bankruptcy costs, we now wish to consider how the price of debt and equity are affected by the bankruptcy costs.  Analytically, we can conclude before any simulations, that the effective interest rates will decrease and the market capitalization will increase with the recovery rates $\alpha_x$ and $\alpha_L$.  For the purposes of this case study we restrict ourselves to the special case that $\alpha_x = \alpha_L$ as in~\cite{veraart2017distress} and only plot the relevant baseline model for debt (with interbank assets treated as risky assets) and equity (with interbank assets treated as risk-free assets). Figure~\ref{fig:EBA_beta} depicts the median effective interest rate and market capitalization for the 87 banks under consideration; this demonstrates that the heuristics are reasonable at high values of $\alpha_x=\alpha_L$, but lose their predictive power if $\alpha_x = \alpha_L < 1$.  Thus, if bankruptcy costs exist, considering the interbank assets as either the risk-free or risky asset can cause mispricing of risk.  %Given the simulated network below, the full network effects cause higher risk in both debt and equity than if each firm were treated individually.
\begin{figure}[h!]
\centering
\begin{subfigure}[b]{.46\textwidth}
\includegraphics[width=\textwidth]{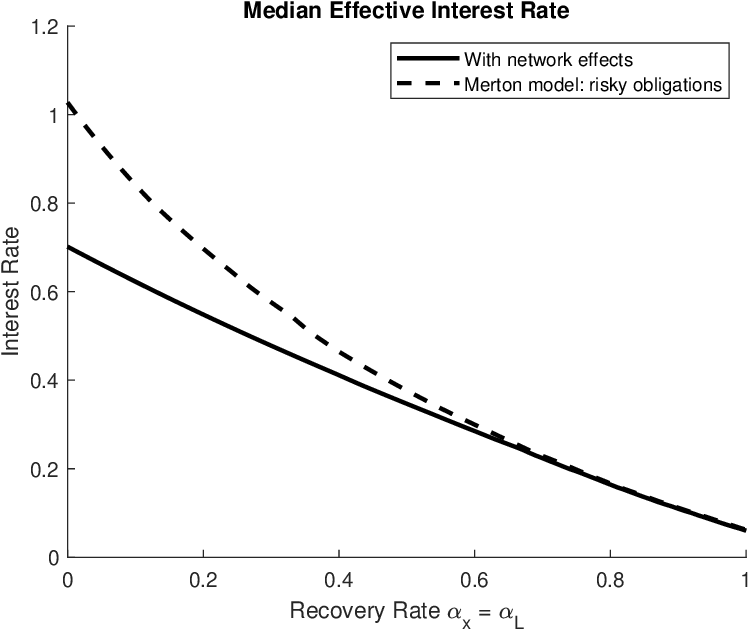}
\caption{Effect of recovery rate on the median effective interest rate.}
\label{fig:R_EBA_beta}
\end{subfigure}
~
\begin{subfigure}[b]{.46\textwidth}
\includegraphics[width=\textwidth]{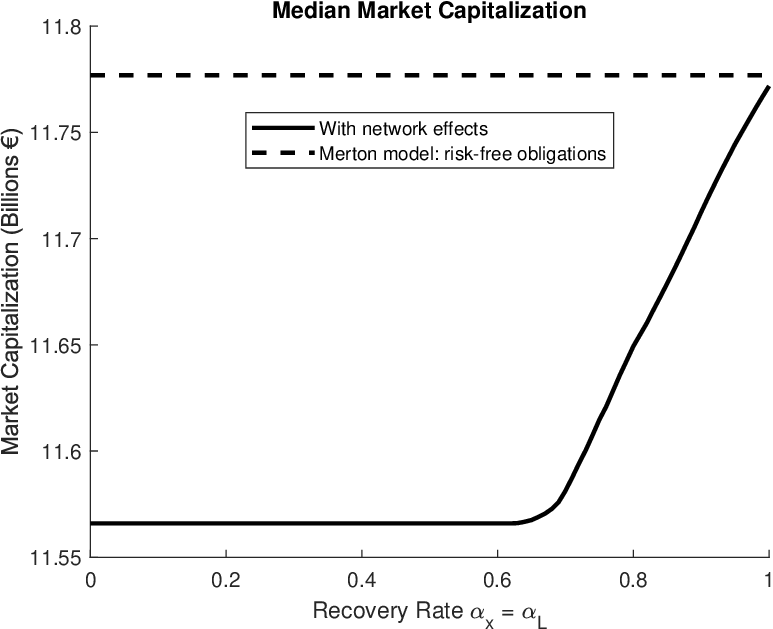}
\caption{Effect of recovery rate on the median market capitalization.}
\label{fig:E_EBA_beta}
\end{subfigure}
\caption{Section~\ref{sec:EBA}: Comparative statics on the recovery rate (with $\alpha_x = \alpha_L$) to the market prices of debt and equity.} 
\label{fig:EBA_beta}
\end{figure}

\section{Conclusion}\label{sec:conclusion}
In this work we present formulas for considering pricing of debt and equity of firms in a financial network under comonotonic endowments.  This methodology extends considerations of CVA for valuation adjustment so that the whole financial network of counterparties is accounted for.   Additionally, the comonotonic framework is theoretically justified, though this is only approximated in market data; under these approximations, we provide upper and lower bound for the price of debt in the Eisenberg-Noe framework.  This is particularly valuable as financial networks are of specific interest in performing stress tests and studying systemic risk.

The models considered herein are simple compared to many modern processes considered for pricing financial instruments.  Though updating the stochastic model of the risky asset would be of interest, herein we will only propose a few extensions for more complete financial network models.  First, we propose utilizing an extension of the Eisenberg-Noe framework in which obligations are neither zero-coupon nor have the same maturities.  Such an underlying financial network has been proposed in~\cite{CC15,KV16,BBF18}.  In particular,~\cite{BBF18} already proposes a setting with stochastic endowments.  We believe using mark-to-market pricing of debt and market capitalization would allow for more realistic determination of default times over an exogenous deficit level.
Second, as during systemic crises the failure of banks and drop in asset prices are inexplicably linked, we believe that including more complicated fire sale dynamics into this system would be of interest.  In particular, we highlight~\cite{AFM16,feinstein2015illiquid,feinstein2016leverage,CW13,CW14} as possible underlying models for use in pricing with price impact dynamics.  
Finally, as highlighted in many empirical works such as~\cite{HalajKok:ECB,HK:Modeling,ELS13,Anand:slides2015}, the network is typically unknown and needs to be estimated from partial information.  Thus sensitivity analysis of pricing under misspecification of the network would be of interest.  In the static, deterministic, setting of~\cite{EN01} this was studied by~\cite{feinstein2017sensitivity}.

\bibliographystyle{plainnat}
\bibliography{bibtex2}

\newpage
\appendix
%\title{Supplemental material}
%\author{Tathagata Banerjee \and Zachary Feinstein}
%\date{\today}
%\maketitle

{
This appendix is organized as follows.  First, in Appendix~\ref{sec:network}, we provide details on the network clearing problem of~\cite{EN01,RV13}.  Then, in Appendix~\ref{sec:general}, we provide comments on computing the expectation of, e.g., the clearing wealths under general random endowments.  In particular, we study the partitioning of the endowment space by the set of defaulting banks.  In Appendix~\ref{sec:q*}, we provide an algorithm for computing the threshold prices $q^*$ introduced in Section~\ref{sec:defaulting} to partition $q$-space by the set of defaulting banks under the comonotonicity assumption.  This algorithm allows for efficient construction of the these threshold prices.  We then consider systemic risk measures in Appendix~\ref{sec:sysrisk} in which we generalize the results of Lemma~\ref{lemma:bound} to find that bounds on these objects can similarly be provided by the comonotonic endowment setting.  In Appendix~\ref{sec:strict-bound} we return to the pure expectation setting presented in the main body of this work to provide simple examples demonstrating that the upper and lower bounds provided in Lemma~\ref{lemma:bound} can be binding.  This is followed by a summary of the Merton model for pricing debt and equity in Appendix~\ref{sec:merton}.  The lognormal setting of the Merton model is then considered in a CAPM setting with idiosyncratic risks and placed in a financial network in Appendix~\ref{sec:capm}.  Following this, we briefly describe the calibration of the interbank network for Section~\ref{sec:EBA} in Appendix~\ref{sec:calibration}.  In Appendix~\ref{sec:buhlmann}, we review details on the risk sharing problem presented in Section~\ref{sec:motivation}.  Finally, the proofs of the results within the main body of this paper are provided in Appendix~\ref{sec:proofs}. 
}

\section{Details of financial networks}\label{sec:network}
In this section we wish to give a formalized construction of all the details of the financial networks considered in this paper.  We begin with the setting proposed in Section~\ref{sec:EN}.

First, we wish to formalize the clearing process $\Psi: \bbr^n \to \bbr^n$ in wealths to describe this system. 
We refer to \cite{veraart2017distress,barucca2016valuation,BBF18,banerjee2017insurance} for detailed discussion of the clearing wealths and their relation to the more typical clearing payments from \cite{EN01,RV13}.
As in \cite{BBF18}, we can define the payments and equity from the wealths $V$ as $p = (\bar p - V^-)^+$ and $E = V^+$ respectively. 
The clearing process is defined for all firms $i$ as
\begin{align}
\label{eq:wealth} \begin{split} \Psi_i(V) &:= \ind{V_i \geq 0} \left[x_i + \sum_{j = 1}^n \pi_{ji} (\bar p_j - V_j^-)^+ - \bar p_i\right]\\
&\qquad + \ind{V_i < 0} \left[\alpha_x x_i + \alpha_L \sum_{j = 1}^n \pi_{ji} (\bar p_j - V_j^-)^+ - \bar p_i\right].\end{split}
\end{align}
As such, the clearing procedure $\Psi$ implies: if bank $i$ has nonnegative wealth $V_i \geq 0$ then it is solvent and its wealth is equal to its total assets minus its total liabilities; if bank $i$ has negative wealth $V_i < 0$ then it is defaulting and its assets are reduced by the recovery rates $\alpha_x,\alpha_L$.  
We note that with $\alpha_x = \alpha_L = 1$ (i.e.\ under no bankruptcy costs) we recover the model of~\cite{EN01}. 

We will now consider existence and uniqueness results on the clearing wealth $V$. In general, we can get existence by applying Tarski's fixed point theorem.

\begin{proposition}\label{prop:existence}
There exists a greatest and least clearing solution to $V = \Psi(V)$ for $V \in \bbr^n$ and any finite clearing solution falls within the lattice $[-\bar p , x + \Pi^\T \bar p - \bar p]$.
\end{proposition}
\begin{proof}
First note that $\Psi$ is nondecreasing in wealths $V$.  Now we will prove that $-\bar p \leq V \leq x + \Pi^\T \bar p - \bar p$ for any $V = \Psi(V) \in \bbr^n$.
\begin{itemize}
\item For any bank $i$: $V_i \geq \ind{V_i \geq 0} [-\bar p_i] + \ind{V_i < 0} [-\bar p_i] = -\bar p_i$ by construction.
\item By monotonicity of the clearing procedure we recover $V \leq \Psi(V^+) = x + \Pi^\T \bar p - \bar p$.
\end{itemize}
The proof is completed by an application of Tarski's fixed point theorem on the lattice $[-\bar p , x + \Pi^\T \bar p - \bar p]$.
\end{proof}

In general, however, the clearing wealth $V$ is not unique. In the special case without bankruptcy costs ($\alpha_x = \alpha_L = 1$), this reduces to the network described in~\cite{EN01}. In that setting we can get uniqueness under very mild assumptions.

\begin{corollary}\label{cor:uniqueness}
Consider a setting with no bankruptcy costs ($\alpha_x = \alpha_L = 1$) and all firms have obligations to the societal node $n+1$ (i.e.\ $\sum_{j = 1}^n \pi_{ij} < 1$ with $\bar p_i > 0$ for all firms $i$), then there exists a unique clearing solution $V = \Psi(V)$.
\end{corollary}
\begin{proof}
The external node $n+1$ implies the system is a regular network \cite[Definition 5]{EN01}. Thus by Theorem 2 of~\cite{EN01} we recover the uniqueness of the clearing solution.
\end{proof}

\begin{proposition}\label{prop:greatest}
Let $V^\uparrow = \Psi(V^\uparrow)$ denote the greatest clearing solution from Proposition~\ref{prop:existence}.  Then $V^\uparrow = \Psi^*(V^\uparrow)$ and is the greatest real-valued fixed point of $\Psi^*$ (defined in~\eqref{eq:wealth2}).
\end{proposition}
\begin{proof}
By Proposition~\ref{prop:existence}, $V^\uparrow \geq -\bar p$ and thus $V^\uparrow = \Psi^*(V^\uparrow)$ as well.  Similarly to the proof of Proposition~\ref{prop:existence}, we can apply Tarski's fixed point theorem to \eqref{eq:wealth2} on the lattice $[-\infty , x + \Pi^\T \bar p - \bar p]$.  Let $V^* = \Psi^*(V^*)$ be the greatest real-valued fixed point of $\Psi^*$ and assume $V^* \geq V^\uparrow$ with $V^*_i > V^\uparrow_i$ for some bank $i$.  Then it must follow that $V^* \geq V^\uparrow \geq -\bar p$, which implies $V^* = \Psi(V^*)$.  However this is a contradiction to $V^\uparrow$ being the greatest clearing solution to $\Psi$.
\end{proof}

We can compute the maximal clearing solution, as discussed in the previous proposition, through an application of the fictitious default algorithm as described in \cite{RV13}. 

\begin{corollary}\label{cor:FDA}
The following algorithm converges to the maximal clearing solution $V^\uparrow = \Psi(V^\uparrow)$:
\begin{enumerate}
\item Initialize $V^{(0)} = x + \Pi^\T \bar p - \bar p$, $z^{(0)} = 0 \in \bbr^n$, and $k = 0$.
\item\label{alg:iterate} Iterate $k = k+1$ and define $z^{(k)} = \ind{V^{(k-1)} < 0} \in \{0,1\}^n$.
\item If $z^{(k)} = z^{(k-1)}$ then $V^\uparrow = V^{(k-1)}$ and terminate.
\item Define $\Lambda = \diag(z^{(k)})$ to be the diagonal matrix with main diagonal defined by $z^{(k)}$ and
    \begin{align*}
    V^{(k)} &= (I - \Lambda)\left[x + \Pi^\T \bar p + \Pi^\T \Lambda V^{(k)} - \bar p\right]\\
    &\qquad + \Lambda \left[\alpha_x x + \alpha_L \Pi^\T \bar p + \Pi^\T \Lambda V^{(k)} - \bar p\right]\\
    &= \left(I - \left(I - (1-\alpha_L)\Lambda\right)\Pi^\T \Lambda\right)^{-1}\left[\left(I - (1-\alpha_x)\Lambda\right)x + \left(I - (1-\alpha_L)\Lambda\right) \Pi^\T \bar p - \bar p\right].
    \end{align*}
\item Go to step~\eqref{alg:iterate}.
\end{enumerate}
\end{corollary}
\begin{proof}
The convergence of this algorithm to the greatest clearing wealth solution follows from the logic of the fictitious default algorithm in \cite{EN01}.  The nonsingularity of the matrix $I - (I - (1-\alpha_L)\Lambda)\Pi^\T \Lambda$ follows from input-output results as detailed in \cite[Theorem~2.6]{feinstein2017sensitivity}.
\end{proof}

{
Before continuing, we wish to recall a notion of monotonicity for fixed points from~\cite{MR94} which we regularly revisit within these appendices.
\begin{theorem}[Theorem 3 of~\cite{MR94}\label{thm:fixedpt-monotone}]
Let $\bbx$ be a complete lattice, $\bbt$ a partially ordered set, and $f: \bbx \times \bbt \to \bbx$.  Suppose $f$ is monotone nondecreasing.  Let $x_L(t) = \inf\{x \; | \; f(x,t) \leq x\}$ and $x_H(t) = \sup\{x \; | \; f(x,t) \geq x\}$.  Then:
\begin{enumerate}
\item $x_L(t)$ and $x_H(t)$ are the least and greatest fixed points of $f(\cdot,t)$,
\item $x_L(\cdot)$ and $x_H(\cdot)$ are nondecreasing, and
\item if for all $x \in \bbx$, $f$ is strictly increasing in $t$ then $x_L(\cdot)$ and $x_H(\cdot)$ are strictly increasing.
\end{enumerate}
\end{theorem}
}

For the remainder of this section we use the notation introduced in Definition~\ref{defn:maximal}.
\begin{proposition}\label{prop:monotonic}
The greatest clearing wealth mapping $V$, and thus also the payment and equity mappings $p$ and $E$, is nondecreasing in the endowments $x$.  
\end{proposition}
\begin{proof}
The monotonicity of the clearing wealths in the endowments follow from Theorem~{\ref{thm:fixedpt-monotone}}.  
The results for the payments and equity follow directly from the definition of those mappings from the clearing wealths.
\end{proof}

\begin{proposition}\label{prop:concave-submodular}
Consider the setting of~\cite{EN01}, i.e.\ $\alpha_x = \alpha_L = 1$.  The greatest clearing wealth mapping $V$ and the payment mapping $p$ are concave and submodular\footnote{{$\phi: \bbr^n \to \bbr$ is submodular if $\phi(x \vee y) + \phi(x \wedge y) \leq \phi(x) + \phi(y)$ for any $x,y$.  $\phi$ is supermodular if $-\phi$ is submodular.}} in the endowments $x$.
\end{proposition}
\begin{proof}
We first note that, under the setting of~\cite{EN01}, we can consider this system as a fixed point in the payments $p(x) = \bar p \wedge (x + \Pi^\T p(x))$ with $V(x) = x + \Pi^\T p(x) - \bar p$.  Thus if $p: \bbr^n_+ \to [0,\bar p]$ is concave (submodular) so is $V$.
\begin{enumerate}
\item Concavity of the clearing payment vector $p$ is given by Lemma 5 of~\cite{EN01}.
\item To prove submodularity of the clearing payment vector $p$, consider that the payment function is the pointwise limit of the mappings $p^k: \bbr^n_+ \to [0,\bar p]$ defined iteratively as: 
\[p^0(x) := \bar p \quad \text{and} \quad p^{k+1}(x) := \bar p \wedge (x + \Pi^\T p^k(x)) \; \forall k \in \bbn, \; \forall x \in \bbr^n_+.\]
As $p(x) = \lim_{k \to \infty} p^k(x)$ by construction (where convergence follows from the monotonicity and boundedness of the arguments $0 \leq p^{k+1}(x) \leq p^k(x)$), if $p^k$ is submodular for all $k$ then the same must be true for the clearing payments $p$.  Trivially $p^0$ is submodular.  Now by induction assume that $p^{k-1}$ is submodular.  Take $x,y \in \bbr^n_+$ and $i \in \{1,2,...,n\}$; there are three cases that must be considered:
\begin{enumerate}
\item If $p_i^k(x) = p_i^k(y) = \bar p_i$ then $p_i^k(x) + p_i^k(y) \geq p_i^k(x \wedge y) + p_i^k(x \vee y)$ by construction.
\item If $p_i^k(x) < p_i^k(y) = \bar p_i$ then $p_i^k(x) \geq p_i^k(x \wedge y)$ and $p_i^k(y) = p_i^k(x \vee y)$ by monotonicity (Proposition \ref{prop:monotonic}); thus $p_i^k(x) + p_i^k(y) \geq p_i^k(x \wedge y) + p_i^k(x \vee y)$. 
\item If $p_i^k(x) < \bar p_i$ and $p_i^k(y) < \bar p_i$ then $p_i^k(x) = x_i + \sum_{j = 1}^n \pi_{ji} p_j^{k-1}(x)$ and $p_i^k(y) = y_i + \sum_{j = 1}^n \pi_{ji} p_j^{k-1}(y)$.  Therefore we find
    \begin{align*}
    p_i^k(x) + p_i^k(y) &= \left[x_i + \sum_{j = 1}^n \pi_{ji} p_j^{k-1}(x)\right] + \left[y_i + \sum_{j = 1}^n \pi_{ji} p_j^{k-1}(y)\right]\\
    &= x_i + y_i + \sum_{j = 1}^n \pi_{ji} \left[p_j^{k-1}(x) + p_j^{k-1}(y)\right]\\
    &\geq x_i \wedge y_i + x_i \vee y_i + \sum_{j = 1}^n \pi_{ji} \left[p_j^{k-1}(x \wedge y) + p_j^{k-1}(x \vee y)\right]\\
    &= \left[(x \wedge y)_i + \sum_{j = 1}^n \pi_{ji} p_j^{k-1}(x \wedge y)\right] + \left[(x \vee y)_i + \sum_{j = 1}^n \pi_{ji} p_j^{k-1}(x \vee y)\right]\\
    &\geq p_i^k(x \wedge y) + p_i^k(x \vee y).
    \end{align*}
\end{enumerate}
\end{enumerate}
\end{proof}

\begin{proposition}\label{prop:alpha-bound}
For any $x \in \bbr^n_+$ and any bank $i$:
\begin{align*}
V_i(x;\Pi,\bar p,\alpha_x,\alpha_L) &\in [V_i(\alpha_x x;\alpha_L \Pi,\bar p,1,1) \, , \, V_i(x;\Pi,\bar p,1,1)] \quad \text{and} \quad\\
p_i(x;\Pi,\bar p,\alpha_x,\alpha_L) &\in [p_i(\alpha_x x;\alpha_L \Pi,\bar p,1,1) \, , \, p_i(x;\Pi,\bar p,1,1)].
\end{align*}
{The bound on wealth holds, also, for the societal node where the relative liabilities owed to society are fixed by $\Pi$, i.e.\ such that the societal wealth is explicitly defined by $V_{n+1}(Y;\tilde\Pi,\bar p,\alpha_x,\alpha_L) := \sum_{i = 1}^n \pi_{i,n+1} p_i(Y;\tilde\Pi,\bar p,\alpha_x,\alpha_L)$.}
\end{proposition}
\begin{proof}
{Note that,} if the result holds for the clearing wealths, then it must also hold for the clearing payments as all bounds are given with respect to the same total obligations $\bar p$.
\begin{enumerate}
\item Consider the proposed upper bound.  Consider $\Psi^*: [-\bar p , x + \Pi^\T\bar p - \bar p] \times [0,1]^2 \to \bbr^n$ with explicit consideration for the recovery rates $\alpha_x,\alpha_L$.  By construction, $\Psi^*$ is (jointly) nondecreasing.  Therefore, by Theorem~{\ref{thm:fixedpt-monotone}}, it must follow that the (maximal) clearing wealths are nondecreasing as a function of the recovery rates as well and the upper bound is proven.
\item Consider the proposed lower bound with fixed $\alpha_x,\alpha_L \in [0,1]$.  Consider $\Psi^\dagger: [-\bar p , x + \Pi^\T\bar p - \bar p] \times [\alpha_x,1] \times [\alpha_L,1] \to \bbr^n$ be the modification of the clearing equation so that
\[\Psi^\dagger(V,(a_x,a_L)) := \left(\ind{V_i \geq 0}a_x + \ind{V_i < 0}\alpha_x\right) x_i + \left(\ind{V_i \geq 0}a_L + \ind{V_i < 0}\alpha_L\right) \sum_{j = 1}^n \pi_{ji} (\bar p_j - V_j^-) - \bar p_i.\]
Notably, by construction of the clearing wealths, $V(x;\Pi,\bar p,\alpha_x,\alpha_L)$ is the maximal fixed point of $\Psi^\dagger(\cdot,(1,1))$ and $V(\alpha_x x;\alpha_L \Pi,\bar p,1,1)$ is the maximal fixed point of $\Psi^\dagger(\cdot,(\alpha_x,\alpha_L))$.  Additionally, by construction, $\Psi^\dagger$ is (jointly) nondecreasing.  Therefore, by Theorem~{\ref{thm:fixedpt-monotone}}, it must follow that the lower bound holds.
\end{enumerate}
{The bounds for the societal node follow directly from the those for the payments $p_i,~i=1,...,n$.}
\end{proof}

\begin{proposition}\label{prop:msb}
Define the random endowments $X \in (L^1_+)^n$.  
The greatest clearing wealths $V(X)$, and thus also the payments and equities $p(X)$ and $E(X)$, is a measurable vector.
\end{proposition}
\begin{proof}
{First, we wish to recall the piecewise linear formulation of the clearing wealths as provided in Section~\ref{sec:piecewise}; in particular, recall $\Delta,\delta$ as defined in~\eqref{eq:Delta} and~\eqref{eq:delta} respectively.}
For any $\omega \in \Omega$ we note that $V(X)[\omega] = \sum_{z \in \{0,1\}^n} \ind{X(\omega) \in \xcal(z)} \left(\Delta(z)X(\omega) - \delta(z)\right)$ where $\xcal(z)$ 
{provides the set of endowments that leads to the defaults encoded in $z$, i.e., $\xcal(z) \subseteq \bbr^n_+$ is defined as:
\[\xcal(z) := \left\{x \in \bbr^n_+ \; \left| \; \begin{array}{l}e_i^\T\Delta(z)x \geq e_i^\T \delta(z) \; \forall i: z_i = 0,\\ e_i^\T\Delta(z)x < e_i^\T \delta(z) \; \forall i: z_i = 1\end{array}\right.\right\} \cap \bigcap_{\bar z \lneq z} \xcal(\bar z)^c.\]
That is, $\xcal(z)$ is constructed as the intersection of a finite number of closed and open halfspaces as well as an additional condition that $x \not\in \bigcup_{\bar z \lneq z} \xcal(\bar z)$.  This additional condition is the one that guarantees that $x \in \xcal(z)$ does not provide a ``better'' (i.e.\ fewer defaulting banks) clearing wealth vector than the maximal clearing solution $V(x)$.
By the use of Tarski's fixed point theorem in the proof of Proposition~\ref{prop:existence}, we are able to guarantee that this construction of $\xcal(z)$ is, in fact, disjoint.}
{As $V(X)[\omega]$ is the summation over a finite set, measurability of the clearing wealths follows.} 
\end{proof}

\section{Expectations under random endowments}\label{sec:general}
In this section we wish to consider a partition of the endowment space $\bbr^n_+$ into regions so that the defaulting set is constant in the $2^n$ subsets.   This problem was considered in great detail in~\cite{gourieroux2012} for the setting without bankruptcy costs, i.e.\ $\alpha_x = \alpha_L = 1$, and with cross-ownership.  Herein we will present a quick extension that allows for bankruptcy costs, i.e.\ for any $\alpha_x,\alpha_L \in [0,1]$.  Notably, when $\min\{\alpha_x,\alpha_L\} < 1$ the partitions need not be convex sets, while they are convex polyhedrons in a system without bankruptcy costs as given in~\cite{gourieroux2012}.  In the below, {we will consider the piecewise linear formulation of the clearing wealths as provided in Section~\ref{sec:piecewise};} if cross-ownership is desired then we refer to Remark~\ref{rem:crossownership} for the modifications necessary to the mappings $\Delta$ and $\delta$.

To consider the partitions, fix $z \in \{0,1\}^n$ to denote the defaulting banks.  By construction, the resulting wealths given endowments $x \in \bbr^n_+$ are provided by $V(x) = \Delta(z)x - \delta(z)$. 
For an endowment vector to be consistent with the defaulting set $z$, it would need to be such that $V_i(x) \geq 0$ if and only if $z_i = 0$.  That is, the set of endowments that generate the defaulting set $z$ is given by the system of inequalities:
\begin{align*}
e_i^\T\Delta(z)x \geq e_i^\T \delta(z) \quad &\forall i: \; z_i = 0\\
e_i^\T\Delta(z)x < e_i^\T \delta(z) \quad &\forall i: \; z_i = 1.
\end{align*}

However, except in the special case that there are no bankruptcy costs ($\alpha_x = \alpha_L = 1$), these regions need not be disjoint.  If an endowment $x$ has two clearing wealth vectors $V^{1} \neq V^{2}$, then it must be that $z^1 \neq z^2$ where $z^k = \ind{V^k < 0}$.  If $z^1 = z^2$ then, by construction of $\Delta(z^k),\delta(z^k)$, it must follow that $V^1 = V^2$.  
In particular, we are interested in the maximal clearing wealth, thus we can construct the partition from the least to greatest number of defaults by considering the set of endowments that lead to $z$ by $\xcal(z) \subseteq \bbr^n_+$ {as provided in the Proof of Proposition~\ref{prop:msb}.  That is, $\xcal(z)$ is} defined as:
\[\xcal(z) := \left\{x \in \bbr^n_+ \; \left| \; \begin{array}{l}e_i^\T\Delta(z)x \geq e_i^\T \delta(z) \; \forall i: z_i = 0,\\ e_i^\T\Delta(z)x < e_i^\T \delta(z) \; \forall i: z_i = 1\end{array}\right.\right\} \cap \bigcap_{\bar z \lneq z} \xcal(\bar z)^c.\]
%That is, $\xcal(z)$ is constructed as the intersection of a finite number of closed and open halfspaces as well as an additional condition that $x \not\in \bigcup_{\bar z \lneq z} \xcal(\bar z)$.  This additional condition is the one that guarantees that $x \in \xcal(z)$ does not provide a ``better'' (i.e.\ fewer defaulting banks) clearing wealth vector than the maximal clearing solution $V(x)$.  By the use of Tarski's fixed point theorem in the proof of Proposition~\ref{prop:existence}, we are able to guarantee that this construction of $\xcal(z)$ is now, in fact, disjoint.
In Figure~\ref{fig:DefaultRegions} we provide an image of the partitioning of the endowment space for a small network with 2 banks plus a societal node.  We note that the societal node in this image can never default, this is due to it having no liabilities and thus always a nonnegative wealth.
\begin{figure}[h!]
\centering
\includegraphics[width=.75\textwidth]{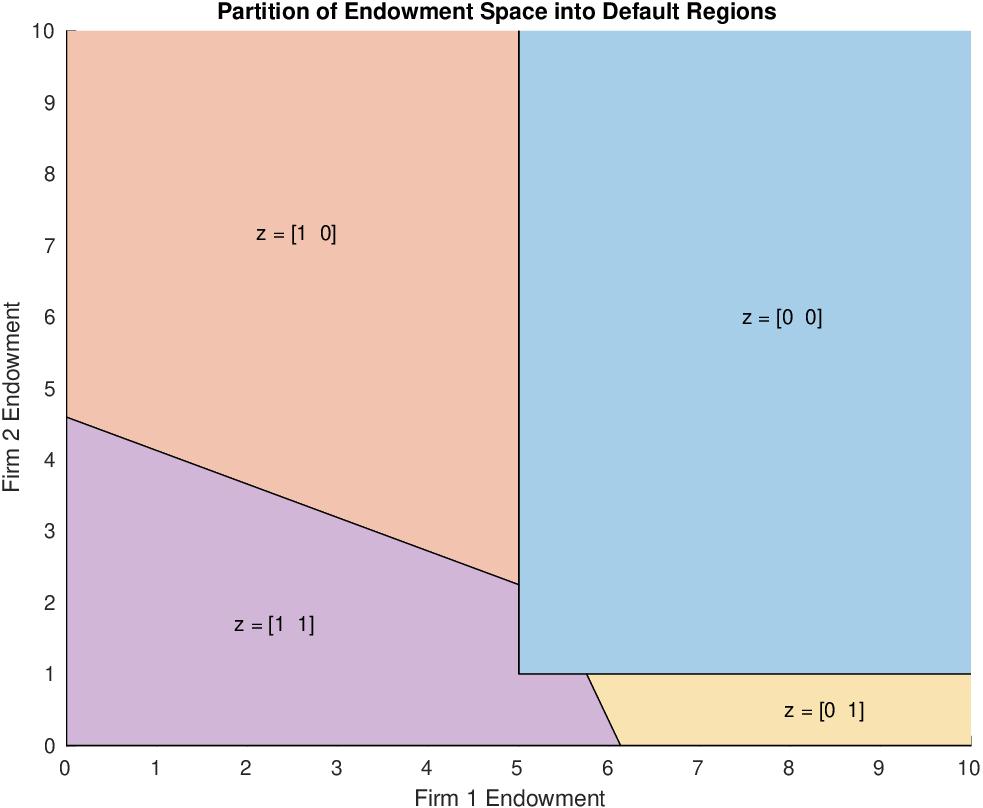}
\caption{Sample partition of the endowment space into default regions for 2 bank plus societal node network with bankruptcy costs.}
\label{fig:DefaultRegions}
\end{figure}

Lastly, we wish to consider the particular case without bankruptcy costs ($\alpha_x = \alpha_L = 1$) as provided by~\cite{gourieroux2012}.  Due to the uniqueness of the clearing wealths (Corollary~\ref{cor:uniqueness}) in this case we obtain only a single consistent default set $z \in \{0,1\}^n$ for every endowment $x \in \bbr^n_+$ and thus do not need to take a secondary intersection as in the case with bankruptcy costs.  Further, in this case, those banks with 0 wealth can be considered equivalently both solvent and defaulting; thus the set of endowments that produce $z \in \{0,1\}^n$ can be considered as the finite intersection of \emph{closed} halfspaces only, i.e.\  the convex polyhedron:
\[\xcal(z) := \left\{x \in \bbr^n_+ \; | \; (I - 2\diag(z))\Delta(z) x \geq (I - 2\diag(z))\delta(z)\right\}.\]
We note that the partition $(\xcal(z))_{z \in \{0,1\}^n}$ is no longer disjoint as boundaries would be shared by the partitions.  However, on this shared boundary the clearing wealths and payments will be equivalent, and this intersection has Lebesgue measure 0, therefore we (and prior authors) have discounted this situation for ease of constructing the sets $\xcal(z)$.

\section{Computing the threshold prices $q^*$}\label{sec:q*}
With {the} comonotonic setting described in Assumption~\ref{ass:comonotonic} we can provide an iterative representation for the lowest prices $q^* \in \bbr^n_+$ such that each firm is solvent {-- as is introduced in Definition~\ref{defn:q*}.  Throughout this section we will consider the notation for order statistics with $q^*_{[k]} \geq q^*_{[k+1]}$ for every bank $k$}. This iterative approach follows as with the fictitious default algorithm of~\cite{EN01,RV13} (see Corollary~\ref{cor:FDA}) insofar as we assume only the first $k$ defaults have occurred to determine the next threshold price $q^*_{[k+1]}$.  That is, to start, assume that that no banks are in default and we find the largest price $q^*_{[1]}$ such that some bank defaults.  Incorporating this default into the system, we find the next highest price $q^*_{[2]}$ at which a default occurs (at most $q^*_{[1]}$ -- in which case this is either a simultaneous or contagious default).  This procedure of adding a single default at a time and finding the next greatest threshold price is repeated sequentially until all bank default levels have been found.  This procedure successfully determines the sequence of defaults and threshold prices $q^*$ due to the monotonicity of the wealths of each bank w.r.t.\ the set of defaulting institutions.
\begin{proposition}\label{prop:q*-gen}
The lowest prices such that the various firms are solvent, defined by $q^*$ in Definition~\ref{defn:q*}, can be defined explicitly by the following iterative relation of decreasing values.  
Initialize $q_{[0]}^* = \infty$ and $z^{(0)} = 0 \in \bbr^n$.  Then for any $k = 1,2,...,n$:
\begin{align*}
[k] &\in \argmax_{i:\; z_i^{(k-1)} = 0} \sup\left\{q \geq 0 \; | \; e_i^\T \Delta(z^{(k-1)})f(q) < \delta_i(z^{(k-1)})\right\}^+\\
q_{[k]}^* &:= \min\left\{q_{[k-1]}^* \; , \; \sup\left\{q \geq 0 \; | \; e_{[k]}^\T \Delta(z^{(k-1)})f(q) < \delta_{[k]}(z^{(k-1)})\right\}^+\right\}\\
z^{(k)} &:= z^{(k-1)} + e_{[k]}
\end{align*}
{with $\Delta,\delta$ defined in~\eqref{eq:Delta} and~\eqref{eq:delta} respectively.}
As a convention, the supremum of the empty set is assumed to be $0$.  If the set definition of $[k]$ has cardinality greater than one, then only a single argument is chosen arbitrarily.  If $f$ is continuous and strictly increasing then $q_{[1]}^* = \left[\max_i f_i^{-1}(\bar p_i - \sum_{j = 1}^n L_{ji})\right]^+$ and $q_{[k]}^*$ can be found via bisection search between $0$ and $q_{[k-1]}^*$.
\end{proposition}
\begin{proof}
This follows directly from the monotonicity of the wealths as given in Proposition~\ref{prop:monotonic} and the construction of $\Delta,\delta$ in~\eqref{eq:Delta} and~\eqref{eq:delta}.  
The level $q_{[k]}^*$ is chosen exactly to be the largest price $q$ so that the $[k]^{th}$ bank would have $0$ wealth (i.e.\ the lowest price so that it is solvent) given that the prior $[1]$ through $[k-1]$ banks have already been deemed insolvent.  The minimum taken with $q_{[k-1]}^*$ is necessary only in the case of contagious defaults, i.e.\ from bankruptcy costs if the jump in payments from bank $[k-1]$ causes bank $[k]$ to also become insolvent at the same time. 
\end{proof}

\section{Bounding systemic risk measures}\label{sec:sysrisk}
Within this section, we wish to extend the main results of this work, i.e.~Lemma~\ref{lemma:bound}, to provide bounds on systemic risk measures.  In fact, we will prove Lemma~\ref{lemma:bound} (within Appendix~\ref{sec:proofs}) by demonstrating that result is a special case of the bounds presented within this section on systemic risk measures.  Herein we will present two notions of systemic risk measures from the literature.  In Section~\ref{sec:insensitive}, we will consider the scalar systemic risk measures as presented in, e.g.,~\cite{chen2013axiomatic,kromer2013systemic} to demonstrate that we can bound the systemic risk via comonotonic endowment settings under certain (weak) conditions.  This is extended in Section~\ref{sec:sensitive} to the case of set-valued risk measures as presented in, e.g.,~\cite{feinstein2014measures,AR16}; in such a setting, the bounding inequalities are, in fact, inclusions.
For simplicity of exposition, within this section, we will assume all random variables are uniformly bounded, i.e., elements of $L^\infty$.

\subsection{Monetary risk measures}\label{sec:riskmsr}

Within this section we will introduce the notion of monetary risk measures.  We refer the interested reader to \cite[Chapter 4]{FS04} for detailed description of these objects; the definitions provided in Definition~\ref{defn:riskmsr} can be found in Chapter 4.1 of~\cite{FS04}.  Briefly, a monetary risk measure is a function which maps the profits and losses of a portfolio into a capital requirement. 
\begin{definition}\label{defn:riskmsr}
A mapping $\rho: L^\infty \to \bbr$ is called a monetary risk measure if it satisfies the following properties.
\begin{itemize}
\item Monotonicity: $\rho(X) \leq \rho(Y)$ for any $X,Y \in L^\infty$ and $X \geq Y$ a.s.
\item Cash invariance: $\rho(X+m) = \rho(X) - m$ for any $X \in L^\infty$ and $m \in \bbr$.
\end{itemize}
A monetary risk measure $\rho: L^\infty \to \bbr$ is called a convex risk measure if it satisfies the following properties.
\begin{itemize}
\item Convexity: $\rho(\lambda X + (1-\lambda)Y) \leq \lambda \rho(X) + (1-\lambda) \rho(Y)$ for any $X,Y \in L^\infty$ and $\lambda \in [0,1]$.
\item Continuity from above: $\rho(X_n) \nearrow \rho(X)$ if $X_n \searrow X$ a.s.
\end{itemize}
A convex risk measure $\rho: L^\infty \to \bbr$ is called a coherent risk measure if it satisfies the following property.
\begin{itemize}
\item Positive homogeneity: $\rho(\lambda X) = \lambda \rho(X)$ for any $X \in L^\infty$ and $\lambda \in \bbr_+$.
\end{itemize}
\end{definition}
Though continuity from above is, often, not assumed as part of the definition of a convex risk measure, that property is necessary for the robust representation considered in Theorem~\ref{thm:dual} below.  It is that representation which we will make extensive use of in Sections~\ref{sec:insensitive} and~\ref{sec:sensitive}.

\begin{example}\label{ex:ES}
One of the most popular monetary risk measures is the \emph{expected shortfall} (also called the ``average value-at-risk'' or ``conditional value-at-risk'').  This risk measure provides the capital necessary so that the tail expectation is nonnegative.  Formally, as defined in Chapter 4.4 of \cite{FS04}, the expected shortfall $ES_{\lambda}(X)$ at level $\lambda \in (0,1]$ and position $X \in L^\infty$ is given by
\[ES_{\lambda}(X) = \frac{1}{\lambda}\int_0^\lambda VaR_{\gamma}(X)d\gamma\]
where $VaR_{\gamma}(X) := \inf\{m \in \bbr \; | \; \P(X+m < 0) \leq \gamma\}$ is the \emph{value-at-risk} for $X$ at level $\gamma \in (0,1]$.  The expected shortfall is a coherent risk measure.  %Furthermore, as we will take advantage in Sections~\ref{sec:insensitive} and~\ref{sec:sensitive}, if $X$ has a continuous distribution then $ES_{\lambda}(X) = -\E[X \; | \; X \leq -VaR_{\lambda}(X)]$ is explicitly given by the (negative) tail expectation (see Corollary 4.49 of~\cite{FS04}).

In fact, as discussed in e.g.~\cite{BT07}, the expected shortfall is an example of an \emph{optimized certainty equivalent} 
\begin{equation}\label{eq:oce}
\rho(X) := \inf_{y \in \bbr} \left(y - \E[u(X+y)]\right)
\end{equation}
for any position $X \in L^\infty$ with concave utility function $u(t) := \frac{1}{\lambda} \min\{0,t\}$.  
In fact, the optimized certainty equivalent~\eqref{eq:oce} for any proper closed concave and nondecreasing utility function $u: \bbr \to \bbr \cup \{-\infty\}$ provides a convex risk measure.
\end{example}

Any monetary risk measure is bijective with the class of \emph{acceptance sets}, i.e., nonempty sets $\acal \subseteq L^\infty$ such that (i) if $X \in \acal$ and $X \leq Y$ a.s.\ then $Y \in \acal$ and (ii) $\inf\{m \in \bbr \; | \; m \in \acal\} > -\infty$.  The representation between risk measures and acceptance sets presented in the following proposition is often called a ``primal representation.''
\begin{proposition}[Propositions 4.6 and 4.7 of \cite{FS04}]\label{prop:primal}
Let $\rho: L^\infty \to \bbr$ be a monetary risk measure, then
\begin{equation}\label{eq:primal-acc}
\acal_{\rho} := \{X \in L^\infty \; | \; \rho(X) \leq 0\}
\end{equation}
is an acceptance set.  If $\rho$ is convex (coherent) then $\acal_{\rho}$ is a weak* closed convex set (weak* closed convex cone).  Moreover, $\rho$ can be recovered from its acceptance set:
\begin{equation*}
\rho(X) = \inf\{m \in \bbr \; | \; X + m \in \acal_{\rho}\} \quad \forall X \in L^\infty.
\end{equation*}

Conversely, let $\acal \subseteq L^\infty$ be an acceptance set, then
\begin{equation*}
\rho_{\acal}(X) := \inf\{m \in \bbr \; | \; X + m \in \acal\} \quad \forall X \in L^\infty
\end{equation*}
defines a monetary risk measure.  If $\acal$ is a weak* closed convex set (weak* closed convex cone) then $\rho_{\acal}$ is a convex (coherent) risk measure.  Moreover, $\acal$ can be recovered from its risk measure:
\begin{equation*}
\acal = \{X \in L^\infty \; | \; \rho_{\acal}(X) \leq 0\}
\end{equation*}
\end{proposition}

\begin{example}\label{ex:ubsr}
Besides the optimized certainty equivalent, there exists another common choice of convex risk measures based on a proper closed concave and nondecreasing utility function $u: \bbr \to \bbr \cup\{-\infty\}$.  The \emph{utility based shortfall risk measure}, presented in Chapter 4.9 of \cite{FS04}, is provided by the acceptance set 
\begin{equation}\label{eq:ubsr}
\acal = \{X \in L^\infty \; | \; \E[u(X)] \geq c\}
\end{equation}
for some threshold utility $c \in \bbr$.  That is, the utility based shortfall risk measure provides the necessary capital so that the expected utility exceeds a given threshold utility.
\end{example}

Finally, we wish to provide a dual or robust representation for convex risk measures.  Such a representation presents a risk measure as a supremum over ``plausible'' probability measures (with penalization for, e.g., divergence from $\P$). 
\begin{theorem}[Theorem 4.31 and Corollary 4.34 of~\cite{FS04}]\label{thm:dual}
Let $\rho: L^\infty \to \bbr$ be a convex risk measure, then
\begin{align*}
\rho(X) &= \sup_{\Q \in \mcal} \left(-\alpha(\Q) - \E^\Q[X]\right) \quad \forall X \in L^\infty\\
\alpha(\Q) &:= \sup_{Y \in \acal} -\E^\Q[Y] \quad \forall \Q \in \mcal
\end{align*}
where $\mcal := \{\Q \ll \P\}$ is the set of all probability measure absolutely continuous w.r.t.\ $\P$.

If $\rho$ is a coherent risk measure then
\[\rho(X) = \sup_{\Q \in \mcal^\Q} -\E^\Q[X] \quad \forall X \in L^\infty\]
where $\mcal^\Q := \{\Q \in \mcal \; | \; \alpha(\Q) = 0\}$.
\end{theorem}

\subsection{Stochastic orders}\label{sec:so}
Before continuing on to the main results of this section, we wish to review some stochastic orders which will be widely used in the following proofs.  Specifically, we will focus on the supermodular and directionally convex orderings.  Notably for the results within this work, a random vector is always dominated by its comonotonic copula under either of these two orders (see Theorem~\ref{thm:sm} and Corollary~\ref{cor:sm-dcx}). 

\subsubsection{The supermodular order}\label{sec:so-sm}
First, we will consider the supermodular order.  This stochastic ordering is provided in much greater detail in Chapter 9.A.4 of~\cite{SS07stochastic}.
\begin{definition}\label{defn:sm}
Let $X,Z$ be two $n$-dimensional random vectors such that 
    \[\E[\phi(X)] \leq \E[\phi(Z)]\]
    for all supermodular functions $\phi: \bbr^n \to \bbr$ (provided the expectations exist).  Then $X$ is said to be smaller than $Z$ in the supermodular order (denoted by $X \leq_{sm} Z$).
\end{definition}
The supermodular order is important to us because of its relation to comonotonicity.  In particular, the comonotonic copula of a random vector always dominates it in the supermodular order.
\begin{theorem}[Theorem 9.A.21 of~\cite{SS07stochastic}]\label{thm:sm}
Let $X = (X_1,...,X_n)$ be a random vector and let $F_{X_i}$ be the marginal distribution of $X_i$, $i = 1,...,n$.  Then, for a uniform random variable $U$ with support $[0,1]$ we have that
\[X \leq_{sm} (F_{X_1}^{-1}(U),...,F_{X_n}^{-1}(U)).\]
\end{theorem}
\subsubsection{The directionally convex order}\label{sec:so-dcx}
We now wish to consider the directionally convex order; this is a weaker order than the supermodular order.  This stochastic ordering is provided in much greater detail in Chapter 7.A.8 of~\cite{SS07stochastic}.
\begin{definition}\label{defn:dcx}
Let $X,Z$ be two $n$-dimensional random vectors such that
    \[\E[\phi(X)] \leq \E[\phi(Z)]\]
    for all directionally convex functions $\phi: \bbr^n \to \bbr$ (provided the expectations exist).\footnote{{$\phi: \bbr^n \to \bbr$ is directionally convex if $\phi(x_2) + \phi(x_3) \leq \phi(x_1) + \phi(x_4)$ for any $x_1 \leq x_2 \leq x_4$ and $x_3 := x_1 + x_4 - x_2$.  $\phi$ is directionally concave if $-\phi$ is directionally convex.}}  Then $X$ is said to be smaller than $Z$ in the directionally convex order (denoted by $X \leq_{dcx} Z$).
\end{definition}
As will be made explicit in the next proposition, if $X$ is smaller than $Z$ in the supermodular order then it must be in the directionally convex order as well.  This is because every directionally convex function is a supermodular function.  Using this result, we are able to show that any random vector is smaller than its comonotonic copula in the directionally convex order in Corollary~\ref{cor:sm-dcx}.
\begin{proposition}[Proposition 7.A.27 of~\cite{SS07stochastic}]\label{prop:dcx}
The following statements are equivalent:
\begin{enumerate}
\item The function $\phi: \bbr^n \to \bbr$ is directionally convex.
\item The function $\phi: \bbr^n \to \bbr$ is supermodular and coordinate-wise convex.
\item For any $x_1,x_2,y \in \bbr^n$ such that $x_1 \leq x_2$ and $y \geq 0$, one has \[\phi(x_1 + y) - \phi(x_1) \leq \phi(x_2 + y) - \phi(x_2).\]
\end{enumerate}
\end{proposition} 
Though implied by the supermodular order, one advantage of the directional convex ordering is that the composition of directionally convex functions is guaranteed to result in a directionally convex function.  We take advantage of this composition structure extensively in Sections~\ref{sec:insensitive} and~\ref{sec:sensitive} due to the structure of systemic risk measures as the composition of functions.
\begin{proposition}[Proposition 7.A.28(a) of~\cite{SS07stochastic}]\label{prop:dcx-comp}
If $\psi: \bbr^m \to \bbr^k$ is nondecreasing and directionally convex (concave) and $\phi: \bbr^n \to \bbr^m$ is nondecreasing and directionally convex (concave), then the composition $\psi \circ \phi$ is nondecreasing and directionally convex (concave).
\end{proposition}
\begin{corollary}\label{cor:sm-dcx}
Let $X,Z$ be two $n$-dimensional random vectors.  If $X \leq_{sm} Z$ then $X \leq_{dcx} Z$.
\end{corollary}
\begin{proof}
As every directionally convex function $\phi$ is supermodular (by Proposition~\ref{prop:dcx}), the result follows trivially.
\end{proof}

\subsection{Scalar systemic risk measures}\label{sec:insensitive}
Consider now the scalar systemic risk measures as defined in, e.g.,~\cite{chen2013axiomatic,kromer2013systemic}.  Such systemic risk measures are defined to be the composition of a monetary risk measure $\rho: L^\infty \to \bbr$ and a nondecreasing aggregation function $\bar\Lambda: \bbr^d \to \bbr$
\begin{equation*}%\label{eq:insensitive}
\rho^{sys}(Z) := \rho(\bar\Lambda(Z))~\quad \forall Z \in (L^\infty)^d
\end{equation*}
for a system of $d \geq 1$ components.
These scalar systemic risk measures are called ``insensitive'' in~\cite{AR16} since the aggregation is insensitive to the capital injection.

Within this discussion, we will focus on those aggregation functions $\bar\Lambda := \Lambda \circ V$ that are constructed as the composition of some nondecreasing aggregation function $\Lambda: \bbr^{n+1} \to \bbr$ and the system wealths $V$.  That is, specifically within this section we consider those systemic risk measures satisfying
\begin{equation}\label{eq:insensitive}
\rho^{sys}(X;\Pi,\bar p,\alpha_x,\alpha_L) := \rho(\Lambda \circ V(X;\Pi,\bar p,\alpha_x,\alpha_L))
\end{equation}
for any financial network $(X;\Pi,\bar p,\alpha_x,\alpha_L)$.
With this construction, we can immediately provide Lemma~\ref{lemma:insensitive} which provides comonotonic bounds on these systemic risk measures.

\begin{lemma}\label{lemma:insensitive}
Let $\rho: L^\infty \to \bbr$ be a convex risk measure.
Let $\Lambda: \bbr^{n+1} \to \bbr$ be a nondecreasing, directionally concave aggregation function.  %concave, and submodular aggregation function.
Define $\rho^{sys}$ by~\eqref{eq:insensitive} to be the associated scalar systemic risk measure.
Let $X \in (L^\infty_+)^n$ and $Z$ be its comonotonic copula, i.e., $Z = (F_{X_1}^{-1}(U),...,F_{X_n}^{-1}(U))$ for uniform random variable $U$ on the support $[0,1]$ and marginal distributions $F_{X_1},...,F_{X_n}$ for $X_1,...,X_n$ respectively, then \[\rho^{sys}(X;\Pi,\bar p,\alpha_x,\alpha_L) \leq \rho^{sys}(\alpha_x Z; \alpha_L \Pi,\bar p,1,1).\] %\[\rho(\Lambda(V(X))) \leq \rho(\Lambda(V(Z))).\]

Further, let $\bbg = \left\{\gcal \subseteq \fcal \; | \; \gcal \text{ is a $\sigma$-algebra}, \; \E[X \; | \; \gcal] \text{ is comonotonic}\right\}$ be the set of sub-$\sigma$-algebras such that $\E[X \; | \; \gcal]$ is a comonotonic projection of $X$. 
If, additionally, $\rho$ is either an optimized certainty equivalent~\eqref{eq:oce} or utility-based shortfall risk measure~\eqref{eq:ubsr} and $\Lambda$ is concave, then \[\rho^{sys}(X;\Pi,\bar p,\alpha_x,\alpha_L) \geq \sup_{\gcal \in \bbg} \rho^{sys}(\E[X | \gcal];\Pi,\bar p,1,1) \geq \rho^{sys}(\E[X];\Pi,\bar p,1,1).\] %\[\rho(\Lambda(V(X))) \geq \rho(\Lambda(V(\E[X]))).\]
\end{lemma}
\begin{proof}
First, we will consider the comonotonic upper bound for these scalar systemic risk measures.  Consider the case of full recovery as in~\cite{EN01}, i.e., $\alpha_x = \alpha_L = 1$.  %For ease of notation, we will leave off the parameters in $V$ for now.
By Proposition~\ref{prop:dcx-comp}, and noting that $V$ is directionally concave by application of Proposition~\ref{prop:concave-submodular} and Proposition~\ref{prop:dcx}, $\Lambda \circ V$ is nondecreasing and directionally concave.
Note that $X \leq_{dcx} Z$ with respect to the directionally convex order by application of Theorem~\ref{thm:sm} and Corollary~\ref{cor:sm-dcx}.
Therefore, $-\E^\Q[\Lambda(V(X;\Pi,\bar p,1,1))] \leq -\E^\Q[\Lambda(V(Z;\Pi,\bar p,1,1))]$ for any measure $\Q \in \mcal$ (i.e., $\Q \ll \P$).  In particular, by applying the robust representation presented in Theorem~\ref{thm:dual}, this implies:
\begin{align*}
\rho^{sys}(X;\Pi,\bar p,1,1) &= \sup_{\Q\in\mcal}\left(-\alpha(\Q) - \E^\Q[\Lambda(V(X;\Pi,\bar p,1,1))]\right)\\
    &\leq \sup_{\Q\in\mcal}\left(-\alpha(\Q) - \E^\Q[\Lambda(V(Z;\Pi,\bar p,1,1))]\right) = \rho^{sys}(Z;\Pi,\bar p,1,1).
\end{align*}
Now, consider $\alpha_x,\alpha_L \in [0,1]$ arbitrary.  By Proposition~\ref{prop:alpha-bound}, monotonicity of the risk measure and aggregation function, and results from the full recovery setting:
\begin{align*}
\rho^{sys}(X;\Pi,\bar p,\alpha_x,\alpha_L) &\leq \rho^{sys}(\alpha_x X;\alpha_L \Pi,\bar p,1,1) \leq \rho^{sys}(\alpha_x Z;\alpha_L \Pi,\bar p,1,1).
\end{align*}

We will now consider Jensen's inequality for the lower bound on these specific classes of scalar systemic risk measures.  As before we will first consider the case of full recovery $\alpha_x = \alpha_L = 1$.  By concavity of the wealth mapping $V$ as demonstrated in Proposition~\ref{prop:concave-submodular}, $\Lambda \circ V$ is concave.  By iterated application of Jensen's inequality 
\begin{align*}
\E[u(\Lambda(V(X;\Pi,\bar p,1,1))+y)] &= \E[\E[u(\Lambda(V(X;\Pi,\bar p,1,1))+y)|\gcal]]\\ 
&\leq \E[u(\E[\Lambda(V(X;\Pi,\bar p,1,1))|\gcal]+y)]\\
&\leq \E[u(\Lambda(V(\E[X|\gcal];\Pi,\bar p,1,1))+y)]
\end{align*}
for any $y \in \bbr$ and $\gcal \in \bbg$.  We will now consider these two classes of risk measures separately.
\begin{itemize}
\item Let $\rho$ be an optimized certainty equivalent.  By the above inequality
\[y - \E[u(\Lambda(V(X;\Pi,\bar p,1,1))+y)] \geq y - \E[u(\Lambda(V(\E[X|\gcal];\Pi,\bar p,1,1))+y)]\]
for any $y \in \bbr$ and $\gcal \in \bbg$.  Therefore, taking an infimum over $y \in \bbr$ retains this inequality and, thus, $\rho^{sys}(X;\Pi,\bar p,1,1) \geq \rho^{sys}(\E[X|\gcal];\Pi,\bar p,1,1)$ for any $\gcal \in \bbg$.
\item Let $\rho$ be a utility-based shortfall risk measure. By the above inequality, for any $y \in \bbr$, if $\E[u(\Lambda(V(X;\Pi,\bar p,1,1))+y)] \geq c$ then $\E[u(\Lambda(V(\E[X|\gcal];\Pi,\bar p,1,1))+y)] \geq c$ as well and, thus, $\rho^{sys}(X;\Pi,\bar p,1,1) \geq \rho^{sys}(\E[X|\gcal];\Pi,\bar p,1,1)$ for any $\gcal \in \bbg$.  
\end{itemize}
(Note that taking the trivial $\sigma$-algebra $\gcal := \{\emptyset,\Omega\}$ provides the final lower bound.)
Now, consider $\alpha_x,\alpha_L \in [0,1]$ arbitrary.  By Proposition~\ref{prop:alpha-bound}, monotonicity of the risk measure and aggregation function, and results from the full recovery setting:
\begin{align*}
\rho^{sys}(X;\Pi,\bar p,\alpha_x,\alpha_L) &\geq \rho^{sys}(X;\Pi,\bar p,1,1) \geq \rho^{sys}(\E[X|\gcal];\Pi,\bar p,1,1)
\end{align*}
for any $\gcal \in \bbg$.
\end{proof}

\begin{remark}
Though the lower bound in Lemma~\ref{lemma:insensitive} is provided only for two specific classes of risk measures, we wish to recall that the expected shortfall, as described in Example~\ref{ex:ES}, is an example of an optimized certainty equivalent.  The expected shortfall with systematic (i.e., comonotonic) shocks was studied as a systemic risk measure in~\cite{amini2021optimal}.  (These bounds are demonstrated numerically in Example~\ref{ex:rm-bound}.)
%In fact, any systemic risk measure that can be written as an expectation w.r.t.\ a single probability measure can be shown to satisfy this lower bound similarly.
\end{remark}

\begin{example}\label{ex:agg-society}
A common aggregation function is the societal wealth, i.e., $\Lambda(V) := V_{n+1}$ (see, e.g.,~\cite{amini2021optimal}).  With this aggregation function, we can immediately apply Lemma~\ref{lemma:insensitive} to bound the systemic risk.
Note that there is nothing special about the societal node herein and any other bank's wealth $\Lambda(V) := V_i$ can equally be considered.  Such an approach will be taken in the proof of Lemma~\ref{lemma:bound} (within Appendix~\ref{sec:proofs}) so as to consider that result a corollary of the bounds for these scalar systemic risk measures.
%We will prove this result as a corollary to Lemma~\ref{lemma:insensitive}.  For this approach we will consider two types of aggregation functions $\Lambda_i^V(V) = V_i$ and $\Lambda_i^p(V) = \bar p_i - V_i^-$ for any bank $i$.  Note that both classes of aggregation functions are concave and submodular (and therefore directionally concave by Proposition~\ref{prop:dcx}).  Further note that $p_i(x;\Pi,\bar p,\alpha_x,\alpha_L) = \Lambda_i^p(V(x;\Pi,\bar p,\alpha_x,\alpha_L))$ for any financial system by construction.
%The bounds follow immediately from Lemma~\ref{lemma:insensitive} by applying the utility-based shortfall risk measure with linear utility function $u(z) := z$ noting that this risk measure can be defined as $\rho(Y) = -\E[Y]$.
\end{example}

We wish to conclude this discussion of the scalar systemic risk measures by providing an analytical description of the expected utility that appears in optimized certainty equivalents and utility-based shortfall risk measures in the comonotonic setting.  (Though presented as the expectation of the utility of an aggregation function, such results equally apply to the expectation of any function of the clearing wealths; in such a way, this representation can be used for, e.g., stochastic root finding for the optimized certainty equivalent as presented in~\cite{duwe08long,HSW}.)  As such, these systemic risk measures have computationally tractable upper and lower bounds as provided in Lemma~\ref{lemma:insensitive}.
\begin{corollary}\label{cor:comonotonic-rm}
Let the endowments be defined by $X = f(q)$ satisfy Assumption~\ref{ass:comonotonic} be uniformly bounded.
%Let $\rho: L^\infty \to \bbr$ be an optimized certainty equivalent or utility-based shortfall risk measure.
Let $\Lambda: \bbr^{n+1} \to \bbr$ be a nondecreasing aggregation function.
%Define $\rho^{sys}$ by~\eqref{eq:insensitive} to be the associated scalar systemic risk measure.
The expected utility which can be used to characterize either an optimized certainty equivalent or utility-based shortfall risk measure can be computed as
\begin{align*}
\E[u(\Lambda(V(X))+y)] &= \sum_{k = 0}^n \E\left[u\left(\Lambda(\Delta_k f(q) - \delta_k)+y\right)\ind{q \in [q_{k+1}^*,q_k^*)}\right] 
\end{align*}
for any $y \in \bbr$
where we define $\Delta_k$ and $\delta_k$ as in Theorem~\ref{thm:comonotonic} by:
\begin{align*}
\Delta_k &:= \begin{cases}\Delta\left(\sum_{j = 1}^k e_j\right) &\text{if } k = 1,2,...,n\\ I &\text{if } k = 0\end{cases} \text{ and } \delta_k := \begin{cases}\delta\left(\sum_{j = 1}^k e_j\right) &\text{if } k = 1,...,n\\ (I - \Pi^\T)\bar p &\text{if } k = 0\end{cases}
\end{align*}
with $\Delta,\delta$ defined in~\eqref{eq:Delta} and~\eqref{eq:delta} respectively and $q^*$ as in Definition~\ref{defn:q*}.
\end{corollary}
\begin{proof}
Using the comonotonicity of $X$, the piecewise linear construction of $V$ in the fictitious default algorithm of Corollary~\ref{cor:FDA}, and the construction of $q^*$:
\begin{align*}
\E[u(\Lambda(V(X))+y)] &= \E[u(\Lambda(V(f(q)))+y)]\\
    &= \sum_{k = 0}^n \E\left[u\left(\Lambda(V(f(q)))+y\right)\ind{q \in [q_{k+1}^*,q_k^*)}\right] \\
    &= \sum_{k = 0}^n \E\left[u\left(\Lambda(\Delta_k f(q) - \delta_k)+y\right)\ind{q \in [q_{k+1}^*,q_k^*)}\right].
\end{align*}
%This follows directly from the piecewise linear construction of $V$ in the fictitious default algorithm of Corollary~\ref{cor:FDA} and the construction of $q^*$.
\end{proof}

\subsection{Set-valued systemic risk measures}\label{sec:sensitive}
Set-valued systemic risk measures were studied in, e.g.,~\cite{feinstein2014measures,AR16}. These risk measures are called ``sensitive'' in~\cite{AR16} as the aggregated value is sensitive to the specific capital injections.  Though presented more generally in~\cite{feinstein2014measures}, herein we follow the structure of~\cite{AR16} so that these set-valued risk measures are constructed explicitly from a scalar systemic risk measure.  In particular, taking the scalar systemic risk measures for a financial network~\eqref{eq:insensitive}, herein we will consider a set-valued systemic risk measure to be the mapping $R^{sys}$ defined by
\begin{equation}\label{eq:sensitive}
R^{sys}(X;\Pi,\bar p,\alpha_x,\alpha_L) := \{y \in \bbr^n \; | \; \rho^{sys}(X+y;\Pi,\bar p,\alpha_x,\alpha_L) \leq 0, \; X+y \in (L^\infty_+)^n\}
\end{equation}
for any financial network $(X;\Pi,\bar p,\alpha_x,\alpha_L)$.
Due to this construction, we are able to bound these set-valued systemic risk measures by application of the bounds provided in Lemma~\ref{lemma:insensitive}.

\begin{corollary}\label{cor:sensitive}
Let $\rho: L^\infty \to \bbr$ be a convex risk measure.
Let $\Lambda: \bbr^{n+1} \to \bbr$ be a nondecreasing, directionally concave aggregation function.  %concave, and submodular aggregation function.
Define $R^{sys}$ by~\eqref{eq:sensitive} to be the associated set-valued systemic risk measure.
Let $X \in (L^\infty_+)^n$ and $Z$ be its comonotonic copula, i.e., $Z = (F_{X_1}^{-1}(U),...,F_{X_n}^{-1}(U))$ for uniform random variable $U$ on the support $[0,1]$ and marginal distributions $F_{X_1},...,F_{X_n}$ for $X_1,...,X_n$ respectively, then \[R^{sys}(X;\Pi,\bar p,\alpha_x,\alpha_L) \supseteq \{y \in \bbr^n \; | \; \alpha_x y \in R^{sys}(\alpha_x Z;\alpha_L \Pi,\bar p,1,1), \; y \geq -Z \text{ a.s.}\}.\]

Further, let $\bbg = \left\{\gcal \subseteq \fcal \; | \; \gcal \text{ is a $\sigma$-algebra}, \; \E[X \; | \; \gcal] \text{ is comonotonic}\right\}$ be the set of sub-$\sigma$-algebras such that $\E[X \; | \; \gcal]$ is a comonotonic projection of $X$. 
If, additionally, $\rho$ is either an optimized certainty equivalent or a utility-based shortfall risk measure and $\Lambda$ is concave, then \[R^{sys}(X;\Pi,\bar p,\alpha_x,\alpha_L) \subseteq \bigcap_{\gcal \in \bbg} R^{sys}(\E[X|\gcal];\Pi,\bar p,1,1) \subseteq R^{sys}(\E[X];\Pi,\bar p,1,1).\]
\end{corollary}
\begin{proof}
Let $y \in \bbr^n$ such that $\alpha_x y \in R^{sys}(\alpha_x Z;\alpha_L \Pi,\bar p,1,1)$ and $y \geq -Z$ a.s.  By construction, this implies $\rho^{sys}(\alpha_x[Z+y];\alpha_L \Pi,\bar p,1,1) \leq 0$.  By simple application of Lemma~\ref{lemma:insensitive}, this implies $\rho^{sys}(X+y;\Pi,\bar p,\alpha_x,\alpha_L) \leq 0$ and $X+y \in (L^\infty_+)^n$, i.e., $y \in R^{sys}(X;\Pi,\bar p,\alpha_x,\alpha_L)$.
Note that, if $\alpha_x = 0$, then $0 \cdot y \in R^{sys}(0 \cdot Z;\alpha_L \Pi,\bar p,1,1)$ if and only if $\rho^{sys}(\vec{0};\alpha_L \Pi,\bar p,1,1) := \rho(\Lambda(\vec{0})) \leq 0$; therefore, the comonotonic bound is either the empty set or $\{y \in \bbr^n \; | \; y \geq -X \text{ a.s.}\}$.  

Consider now the bound for an optimized certainty equivalent or utility-based shortfall risk measure.  By simple application of Lemma~\ref{lemma:insensitive} with fixed $\gcal \in \bbg$, $\rho^{sys}(X+y;\Pi,\bar p,\alpha_x,\alpha_L) \geq \rho^{sys}(\E[X|\gcal]+y;\Pi,\bar p,1,1)$ for any $y \in \bbr^n$ such that $X+y \in (L^\infty_+)^n$.  Therefore $y \in R^{sys}(X;\Pi,\bar p,\alpha_x,\alpha_L)$ implies $y \in R^{sys}(\E[X|\gcal];\Pi,\bar p,1,1)$ as well.
\end{proof}

\section{Strictness of bounds in Lemma~\ref{lemma:bound}}\label{sec:strict-bound}
Herein, we wish to provide two simple examples that demonstrate that the upper and lower bounds provided within Lemma~\ref{lemma:bound} are strict in the sense that there exist networks in which either of these bounds is binding.
\begin{example}\label{ex:counterexample}
Consider a general network with interbank obligations summarized by $(\Pi,\bar p)$ with recovery rates $\alpha_x,\alpha_L \in [0,1]$.  
\begin{enumerate}
\item Let $X = x \in \bbr^n_+$ a.s.\ -- which is comonotonic -- with financial network $(\Pi,\bar p)$ such that $x + \Pi^\T \bar p - \bar p \in -\bbr^n_{++}$, i.e., such that all banks are (almost surely) in default.  By the clearing procedure, the clearing wealths (and therefore also the clearing payments) must satisfy:
    $V(x;\Pi,\bar p,\alpha_x,\alpha_L) = \left(I - \alpha_L\Pi^\T\right)^{-1}\alpha_x x - \bar p = V(\alpha_x x;\alpha_L \Pi,\bar p,1,1)$.
    That is, the expectation over $X$ equals the lower bound provided in Lemma~\ref{lemma:bound}.
\item Let $X = x \in \bbr^n_+$ a.s.\ -- which is such that $\E[X] = x$ as well -- with financial network $(\Pi,\bar p)$ such that $x + \Pi^\T \bar p - \bar p \in \bbr^n_+$, i.e., such that (almost surely) no banks are in default.  By the clearing procedure, the clearing wealths (and therefore also the clearing payments) must satisfy:
    $V(x;\Pi,\bar p,\alpha_x,\alpha_L) = x + \Pi^\T \bar p - \bar p = V(x;\Pi,\bar p,1,1)$.
    That is, the expectation over $X$ equals the upper bound provided in Lemma~\ref{lemma:bound}.
\end{enumerate}
\end{example}

\section{Merton model for the pricing of debt and equity}\label{sec:merton}
In this section we wish to give a formalized description of the Merton model for debt and equity as defined in~\cite{merton1974}.  This model is utilized as a baseline to compare to the network valuation results presented within this work.  

In contrast to the financial networks generally utilized in this work, the Merton model studies the single firm setting.  This firm has only one kind of debt $\bar p \geq 0$ with a single maturity date $T > 0$.  On the other side of the balance sheet, the firm has assets worth $b e^{rT} + sq$ at maturity $T$ where the firm has invested $b \in \bbr$ in a risk-free bond with interest rate $r \geq 0$ and $s \geq 0$ in a risky asset with stochastic payout $q \geq 0$ almost surely. 
\begin{assumption}\label{ass:merton}
The risky asset has \emph{lognormal} payout $q = \exp([\mu - \frac{\sigma^2}{2}]T + \sigma Z \sqrt{T})$ with drift $\mu \in \bbr$ and volatility $\sigma > 0$ where $Z$ is some standard normal random variable.
\end{assumption}

The primary question of the Merton model is in regards to the price of debt and equity at a time prior to maturity and, in particular, at time $0$.  
The total payments made at maturity is 
\[p := \bar p \ind{b e^{rT} + sq \geq \bar p} + \alpha\left(b e^{rT} + sq\right)\ind{b e^{rT} + sq < \bar p}\] 
for recovery rate $\alpha \in [0,1]$, i.e., the firm pays its total liabilities $\bar p$ if its assets are worth at least as much as the liabilities and otherwise it pays out all of its assets (less a proportional bankruptcy cost $1-\alpha$).  We wish to note that~\cite{merton1974} considers only the case of full recovery, i.e., $\alpha = 1$.
The equity at maturity has the payoff of the European call option 
\[E := (b e^{rT} + sq - \bar p)^+,\] 
i.e., the surplus assets after paying out all liabilities.
Pricing of these claims is accomplished under the risk-neutral measure $\Q$ for the risky asset, i.e.
\[\frac{d\Q}{d\P} = \exp\left(-\frac{1}{2}\left(\frac{\mu - r}{\sigma}\right)^2 T - \left(\frac{\mu - r}{\sigma}\right)Z\sqrt{T}\right).\]
{(Though presented herein as a single-period setting, we assume that the lognormal payoff of $q$ is generated from a geometric Brownian motion.)} 
That is, the price of debt is the discounted $\Q$-expectation of the payments $p$ normalized by the total liabilities $\bar p$ and the market capitalization for the firm is the discounted $\Q$-expectation of the equity $E$.
\begin{lemma}\label{lemma:merton}
Under the risk-neutral measure $\Q$, the discounted price of the firm's debt and market capitalization are given, respectively, by:
\begin{align*}
\E^\Q\left[e^{-rT}\frac{p}{\bar p}\right] &= \begin{cases} e^{-rT} &\text{if } be^{rT} \geq \bar p \\
    e^{-rT} + \frac{\alpha s}{\bar p} \Phi(-d^1) - \left[e^{-rT} - \frac{\alpha b}{\bar p}\right]\Phi(-d^2) &\text{if } be^{rT} < \bar p \end{cases}\\ 
\E^\Q\left[e^{-rT}E\right] &= \begin{cases} b + s - e^{-rT}\bar p &\text{if } be^{rT} \geq \bar p \\
    s\Phi(d^1) - \left[\bar p e^{-rT} - b\right]\Phi(d^2)&\text{if } be^{rT} < \bar p \end{cases}
\end{align*}
where $\Phi: \bbr \to [0,1]$ is the standard normal cumulative distribution function and
\begin{align*}
d^1 &= \frac{\log\left(\frac{s}{\bar p - be^{-rT}}\right) + \left(r + \frac{\sigma^2}{2}\right)T}{\sigma\sqrt{T}}\\
d^2 &= d^1 - \sigma\sqrt{T}.
\end{align*}
\end{lemma}
\begin{proof}
First, in the case that $be^{rT} \geq \bar p$ then the firm will always pay its obligations in full and any excess is equity.  Consider the case with $be^{rT} < \bar p$.
Consider the price of debt from the definition of the payment $p$
\begin{align*}
\E^\Q\left[e^{-rT}\frac{p}{\bar p}\right] &= e^{-rT}\Q(sq \geq \bar p - be^{rT}) + \frac{\alpha}{\bar p}\E^\Q[e^{-rT}(be^{rT}+sq)\ind{sq < \bar p - be^{rT}}]\\
    &= e^{-rT}\Q(sq \geq \bar p - be^{rT}) + \frac{\alpha b}{\bar p}\Q(sq < \bar p - be^{rT}) + s\widehat\Q(sq < \bar p - be^{rT})
\end{align*}
where $\frac{d\widehat\Q}{d\Q} = e^{-rT}q$.  By construction of the measures $\Q,\widehat\Q$, these probabilities are explicitly given with respect to $d^1,d^2$, i.e.,
\begin{align*}
\Q(sq \geq \bar p - be^{rT}) &= \Phi(d^2) \quad \text{ and } \quad \widehat\Q(sq \geq \bar p - be^{rT}) = \Phi(d^1).
\end{align*}
Therefore, by symmetry of the normal distribution, it can easily be proven that
\[\E^\Q\left[e^{-rT}\frac{p}{\bar p}\right] = e^{-rT} + \frac{\alpha s}{\bar p} \Phi(-d^1) - \left[e^{-rT} - \frac{\alpha b}{\bar p}\right]\Phi(-d^2).\]
Likewise, the market capitalization for the firm can be decomposed as
\begin{align*}
\E^\Q\left[e^{-rT}E\right] &= \E^\Q[e^{-rT}(be^{rT}+sq)\ind{sq \geq \bar p - be^{rT}}] - e^{-rT}\bar p\Q(sq \geq \bar p - be^{rT})\\
    &= s\widehat\Q(sq \geq \bar p - be^{rT}) - \left[\bar p e^{-rT} - b\right]\Q(sq \geq \bar p - be^{rT})\\
    &= s\Phi(d^1) - \left[\bar p e^{-rT} - b\right]\Phi(d^2).
\end{align*}
\end{proof}

\section{Pricing under CAPM model with idiosyncratic risk}\label{sec:capm}

Herein we wish to consider the setting in which firms have idiosyncratic risks that are independent from the market portfolio, but may have complicated correlation structures between firm idiosyncratic risks.  That is, {firm endowments follow the CAPM structure}.
CAPM is a single factor model for the returns of assets; under this structure, bank $i$ has returns based on its correlation to the market portfolio $q$ and additional idiosyncratic risk.  The CAPM structure was first introduced in~\cite{sharpe1964,lintner1965}.  This model provides bank $i$'s logarithmic returns $r_i$ up to maturity $T$ as a function of the risk-free rate $r \geq 0$ and the market returns $r_M = \log(q)$ through the relation
\[r_i = r T + \beta_i\left(r_M - r T\right) +  \epsilon_i\]
where $\beta_i$ is the correlation of the bank $i$'s return to the market multiplied by the ratio of the standard deviation of bank $i$'s investment to the standard deviation of the market portfolio (i.e., the market beta for bank $i$) and $\epsilon_i$ is some idiosyncratic return for bank $i$.
Notably, this setting is \emph{not} comonotonic, though we find comonotonic upper and lower bounds along the lines of Lemma~\ref{lemma:bound} and Remark~\ref{rem:sfm} for this setting.  
%For simplicity, we will focus solely on the full recovery setting of~\cite{EN01} herein, i.e., $\alpha_x = \alpha_L = 1$.  
We provide short numerical {examples} to demonstrate the value of all bounds provided in this setting as a function of market betas.

\begin{assumption}\label{ass:capm}
%Let the financial markets allow for full recovery $\alpha_x = \alpha_L = 1$ in case of default. 
Let $r \geq 0$ be the risk-free rate.  Let $s \in \bbr^n_+$ denote the investment size of the firms.  Fix $\beta \in \bbr^n_+$ to be the vector of market betas for each firm and $\gamma \in \bbr^n_+$ to be the vector of standard deviations for the idiosyncratic risks.  Consider a terminal time $T > 0$.  Define {endowments} $X := \diag(s)\eta$ almost surely where, for each bank $i$,
\[\eta_i := \exp\left(\left[(1-\beta_i)r + \beta_i\mu_M - \frac{\beta_i^2\sigma_M^2 + \gamma_i^2}{2}\right]T + \left(\beta_i\sigma_M Z_M + \gamma_i Z_i\right)\sqrt{T}\right)\]
for standard normal random variables $Z_M,Z_1,...,Z_n$ such that $Z_M$ is independent from each $Z_i$ and $(Z_1,...,Z_n)$ follows a multivariate normal distribution.
Further, consider the market portfolio that follows the lognormal distribution 
\[q := \exp\left(\left[\mu_M - \frac{\sigma_M^2}{2}\right]T + \sigma_M Z_M \sqrt{T}\right)\]
with market drift $\mu_M$ and standard deviation $\sigma_M \geq 0$.
That is, in the CAPM language presented previously, bank $i$'s return is provided by $r_i = \log(\eta_i)$ with market return $r_M = \log(q)$ and idiosyncratic returns $\epsilon_i = -\frac{\gamma_i^2}{2}T + \gamma_i Z_i \sqrt{T}$.
\end{assumption}
Throughout we will also consider the individual firm volatilities $\sigma_i = \sqrt{\beta_i^2 \sigma_M^2 + \gamma_i^2}$.  With this we can define the market beta by the firm's correlation to the market portfolio $\rho_{iM}$, i.e.\ $\beta_i = \rho_{iM} \frac{\sigma_i}{\sigma_M}$, and we can determine the idiosyncratic volatility from this correlation as well, i.e.\ $\gamma_i = \sqrt{1-\rho_{iM}^2}\sigma_i$.  We note that the setting in which $\rho_{iM} = 1$ for all firms $i$ provides already a comonotonic setting.

Throughout this section, though a single-period setting is constructed, we assume that this structure is generated from a multivariate geometric Brownian motion with dynamic trading strategies (but a static network of obligations).  As such, we will be considering the expectations with respect to the (unique) risk-neutral measure $\Q$ for the market, i.e.\
\[\frac{d\Q}{d\P} = \exp\left(-\frac{1}{2}\left(\frac{\mu_M-r}{\sigma_M}\right)^2 T - \left(\frac{\mu_M-r}{\sigma_M}\right)Z_M\sqrt{T}\right).\]
Thus the market portfolio, under $\Q$, has value at maturity $T$ of
\[q = \exp\left((r - \frac{\sigma_M^2}{2})T + \sigma_M \tilde{Z}_M \sqrt{T}\right)\]
where $\tilde{Z}_M$ is a standard normal random variable under the risk-neutral measure $\Q$.
Additionally the portfolios for each firm $i$ has value at maturity $T$ of 
\begin{align}
\label{eq:eta-1} \eta_i &= \exp\left((r - \frac{\sigma_i^2}{2})T + \beta_i \sigma_M \tilde{Z}_M \sqrt{T} + \gamma_i Z_i \sqrt{T}\right)\\ 
\label{eq:eta-2} &= \exp\left(\left[(1-\beta_i)r - \frac{\gamma_i^2}{2}\right]T + \gamma_i Z_i \sqrt{T}\right)\exp\left((1-\beta_i)\frac{\beta_i \sigma_M^2}{2}T\right)q^{\beta_i}.
\end{align}

For this setting we now wish to consider analytical bounds from Lemma~\ref{lemma:bound}. These bounds permit the quantification of the effects of diversity of investment strategies on pricing and, therefore, also systemic risk.  By deconstructing the returns of any bank into the systematic or market component $q$ and the idiosyncratic component, the CAPM setting can suitably recover the tradeoffs between systematic risk and idiosyncratic risks on systemic risk.  In particular, by utilizing the comonotonic bounds from Lemma~\ref{lemma:bound}, we can analytically consider, e.g., the sensitivity of the price of debt -- a proxy for systemic risk -- to the various network parameters; an example of such a result is displayed in Example~\ref{ex:bound} below.
\begin{lemma}\label{lemma:capm}
Under the risk-neutral measure $\Q$, the discounted price of firm $i$'s debt $\EQ[e^{-rT}p_i(\diag(s)\eta)/\bar p_i]$ is bounded above and below by:
\begin{align}
\label{eq:capm-bound} \EQ\left[e^{-rT}\frac{p_i({\alpha_x}\diag(s) \hat \eta(\sigma) {; \alpha_L \Pi,\bar p,1,1})}{\bar p_i}\right] &\leq \EQ\left[e^{-rT}\frac{p_i(\diag(s) \eta {; \Pi,\bar p,\alpha_x,\alpha_L})}{\bar p_i}\right]\\ 
\nonumber &\leq \EQ\left[e^{-rT}\frac{p_i(\diag(s) \hat \eta(\beta \sigma_M) {; \Pi,\bar p,1,1})}{\bar p_i}\right] 
\end{align}
where $\hat \eta: \bbr^n_+ \to (L^1_+)^n$ is defined by 
\begin{align*}
\hat \eta_i(z) &= \exp\left((r - \frac{z_i^2}{2})T + z_i \tilde Z_M \sqrt{T}\right) = \exp\left((1-\frac{z_i}{\sigma_M})(r + \frac{z_i\sigma_M}{2})T\right)q^{\frac{z_i}{\sigma_M}}.
\end{align*}
That is, $\hat \eta(\sigma)$ and $\hat \eta(\beta \sigma_M)$ are the comonotonic version and conditional expected version (w.r.t.\ $q$) of $\eta$ respectively.
In fact, the comonotonic bounds can be computed via the general structure:
\begin{align}
\nonumber \EQ\left[e^{-rT}\frac{p_i(\diag(s) \hat \eta(z) {; \Pi,\bar p,1,1})}{\bar p_i}\right] &= e^{-rT} + \frac{1}{\bar p_i}e_i^\T \sum_{k = i}^n \left[\Delta_k\diag(s)\left(\Phi(-\hat d_k^1(z)) - \Phi(-\hat d_{k+1}^1(z))\right)\right.\\
\label{eq:capm-price} &\qquad\qquad\quad \left.- e^{-rT}\delta_k \left(\Phi(-\hat d_k^2(z)) - \Phi(-\hat d_{k+1}^2(z))\right) \right]
\end{align}
where the {threshold market prices $\hat q^*(z)$ as in Definition~\ref{defn:q*} with the} ordering of banks provided as in Assumption~\ref{ass:q*order}, $\Phi: \bbr \to [0,1]$ is the standard normal cumulative distribution function, 
\begin{align*}
\hat d_k^1(z) &= \frac{\log(1/\hat q_k^*(z))\vec{1} + (r\vec{1} - \frac{1}{2}(\sigma_M\vec{1} - 2z)\sigma_M)T}{\sigma_M \sqrt{T}} \quad \forall k = 0,1,...,n+1\\
\hat d_k^2(z) &= \frac{\log(1/\hat q_k^*(z)) + (r - \frac{1}{2}\sigma_M^2)T}{\sigma_M \sqrt{T}} \quad \forall k = 0,1,...,n+1,
\end{align*}
{and $\Delta_k,\delta_k$ are defined as in Theorem~\ref{thm:comonotonic}, i.e.,
\begin{align*}
\Delta_k &:= \begin{cases}\Delta\left(\sum_{j = 1}^k e_j\right) &\text{if } k = 1,2,...,n\\ I &\text{if } k = 0\end{cases} \text{ and } \delta_k := \begin{cases}\delta\left(\sum_{j = 1}^k e_j\right) &\text{if } k = 1,...,n\\ (I - \Pi^\T)\bar p &\text{if } k = 0\end{cases}
\end{align*}
{with $\Delta,\delta$ defined in~\eqref{eq:Delta} and~\eqref{eq:delta} respectively based on the ordering of $\hat q^*(z)$.}
}
\end{lemma}
\begin{proof}
{We will prove this result in two parts.  First, we will demonstrate the upper and lower bounds of~\eqref{eq:capm-bound} hold.  Second, we will prove the pricing structure for a given comonotonic structure as in~\eqref{eq:capm-price}.
\begin{enumerate}
\item Through an application of Lemma~\ref{lemma:bound}, we will prove the upper and lower bounds by showing that $\hat \eta(\sigma)$ is the comonotonic copula of $\eta$ and $\hat \eta(\beta \sigma_M) = \EQ[\eta | q]$ is a comonotonic projection of $\eta$. First, consider the marginal distribution of $\eta_i$ for bank $i$ w.r.t.\ $\Q$: a lognormal distribution (for Gaussian mean $(r - \frac{\sigma_i^2}{2})T$ and variance $\sigma_i^2 T$) as depicted in~\eqref{eq:eta-1}.  Therefore the comonotonic copula of $\eta$ is a random vector of lognormals with some consistent underlying standard normal distribution.  Therefore $\hat \eta(\sigma)$ is the comonotonic copula of $\eta$ w.r.t.\ the standard normal $\tilde Z_M$ (under the measure $\Q$).  Now, we will consider $\EQ[\eta | q]$ by taking advantage of the form given in~\eqref{eq:eta-2}:
    \begin{align*}
    \EQ[\eta_i | q] &= \EQ\left[\exp\left(\left[(1-\beta_i)r - \frac{\gamma_i^2}{2}\right]T + \gamma_i Z_i \sqrt{T}\right)\exp\left((1-\beta_i)\frac{\beta_i \sigma_M^2}{2}T\right)q^{\beta_i} \; | \; q\right] \\
%    &= \EQ\left[\exp\left(\left[(1-\beta_i)r - \frac{\gamma_i^2}{2}\right]T + \gamma_i Z_i \sqrt{T}\right) \; | \; q\right] \exp\left((1-\beta_i)\frac{\beta_i \sigma_M^2}{2}T\right)q^{\beta_i} \\
    &= \EQ\left[\exp\left(\left[(1-\beta_i)r - \frac{\gamma_i^2}{2}\right]T + \gamma_i Z_i \sqrt{T}\right)\right] \exp\left((1-\beta_i)\frac{\beta_i \sigma_M^2}{2}T\right)q^{\beta_i} \\
%    &= \exp\left((1-\beta_i)r\right) \exp\left((1-\beta_i)\frac{\beta_i \sigma_M^2}{2}T\right)q^{\beta_i} \\
    &= \exp\left((1-\beta_i)\left[r + \frac{\beta_i \sigma_M^2}{2}\right]T\right) q^{\beta_i} \\
    &= \hat \eta_i(\beta_i\sigma_M).
    \end{align*}
Furthermore, $\hat\eta_i(\beta_i\sigma_M)$ is comonotonic by the positivity of the exponential (i.e., $\exp\left((1-\beta_i)(r + \frac{\beta_i\sigma_M^2}{2})T\right)$) and the nonnegativity of the power of $q$ (i.e., $\beta_i$) for each bank $i$.
\item We now wish to prove the pricing structure for a given comonotonic structure as provided in~\eqref{eq:capm-price}.  For this purpose, we will consider the generating function $f(q) := \diag(s)\hat\eta(z)$ (which is implicitly a function of $q$) for the comonotonic endowments as considered in Proposition~\ref{prop:comonotonic-defn}.  Thus, following the pricing structure from Theorem~\ref{thm:comonotonic}, we find
    \begin{align*}
    \EQ&\left[e^{-rT}\frac{p_i(\diag(s)\hat\eta(z);\Pi,\bar p,1,1)}{\bar p_i}\right] \\
    &= \frac{e^{-rT}}{\bar p_i}\left(\bar p_i + e_i^\T \sum_{k = i}^n\left[\Delta_k\EQ\left[f(q)\ind{q \in [\hat q_{k+1}^*(z),\hat q_k^*(z))}\right] - \delta_k \Q(q \in [\hat q_{k+1}^*(z),\hat q_k^*(z))) \right]\right) \\
    &= e^{-rT} + \frac{1}{\bar p_i}e_i^\T \sum_{k = i}^n\Big[\Delta_k\diag(s)\EQ\left[e^{-rT}\hat\eta(z)\ind{q \in [\hat q_{k+1}^*(z),\hat q_k^*(z))}\right]\\
    &\qquad\qquad\qquad - e^{-rT}\delta_k \Q(q \in [\hat q_{k+1}^*(z),\hat q_k^*(z))) \Big].
%    &=  ...\\
%    &= e^{-rT} + \frac{1}{\bar p_i}e_i^\T \sum_{k = i}^n \left[\Delta_k\diag(s)\left(\Phi(-\hat d_k^1(z)) - \Phi(-\hat d_{k+1}^1(z))\right)- e^{-rT}\delta_k \left(\Phi(-\hat d_k^2(z)) - \Phi(-\hat d_{k+1}^2(z))\right) \right]
    \end{align*}
Consider now the structures of $\EQ[e^{-rT}\hat\eta(z)\ind{q \leq \hat q_k^*(z)}]$ and $\Q(q \leq \hat q_k^*(z))$.  These expectations will be computed comparably to that undertaken in Lemma~\ref{lemma:merton} for the (single-firm) Merton model, i.e., let $\hat\Q$ denote the vector of probability measures such that $\frac{d\hat\Q_i}{d\Q} = e^{-rT}\hat\eta_i(z)$ for each bank $i$.  By construction of the measures $\Q,\hat\Q$ these expectations are explicitly given with respect to $\hat d^1(z),\hat d^2(z)$, i.e.,
    \begin{align*}
    \EQ[e^{-rT}\hat\eta_i(z)\ind{q \leq \hat q_k^*(z)}] = \hat\Q_i(q \leq \hat q_k^*(z)) &= \Phi(-\hat d_k^1(z)),\\
    \Q(q \leq \hat q_k^*(z)) &= \Phi(-\hat d_k^2(z)).
    \end{align*}
Therefore, by construction, we can immediately recover the form of~\eqref{eq:capm-price}.
\end{enumerate}
}
%The upper and lower bounds follow directly from Lemma~\ref{lemma:bound} (with the conditional expectation upper bound as discussed in Remark~\ref{rem:sfm}).  The pricing structures follow from Theorem~\ref{thm:comonotonic} (with $f(q) = sq$) and the typical results from, e.g.,~\cite{merton1974} or Online Appendix~\ref{sec:merton}.
\end{proof}
Though not presented in the above lemma, similar bounds can be provided for the effective interest rate and system-wide market capitalization as in Remark~\ref{rem:bound-R} and Corollary~\ref{cor:bound} respectively. {Similar bounds can likewise be applied to systemic risk measures as presented in Appendix~\ref{sec:sysrisk}; this is numerically computed in Example~\ref{ex:rm-bound} below.}

{We wish to note that the upper bound for the price of debt as provided in Lemma~\ref{lemma:capm} is provided by the comonotonic projection of the endowments $X$ onto the market portfolio $q$ (as all banks are long in the market portfolio by $\beta_i \geq 0$ for all $i$). In fact, as the CAPM model is a single factor model, this upper bound follows the proposed structure as considered in Remark~\ref{rem:sfm}.  This demonstrates that a single factor model can readily lead to an interpretable comonotonic projection.}
%\begin{remark}\label{rem:sfm}
%For simplicity of exposition, consider the full recovery setting $\alpha_x = \alpha_L = 1$ of~\cite{EN01}.
%Consider a single factor model for endowments that allows for errors or idiosyncratic terms in which every firm is long the factor: $X = f(q,\epsilon)$ where $f$ and $q$ are as in Assumption~\ref{ass:comonotonic} and $\epsilon \in (L^0)^n$ is independent from $q$ ($f$ need not be monotonic in $\epsilon$).  The conditional upper bound considered in Lemma~\ref{lemma:bound} implies $\E[V_i(X)] \leq \E[V_i(\E[X | q])]$ and $\E[p_i(X)] \leq \E[p_i(\E[X | q])]$ for all banks $i$ ({i.e., with $\gcal = \sigma(q)$}).  We wish to note that this upper bound is again with respect to comonotonic endowments and, thus, can be computed via the formulations in Theorem~\ref{thm:comonotonic}.
%\end{remark}

\begin{remark}\label{rem:q*CAPM}
In the above representation we note that the ordering of banks may change based on the parameter $z$ (i.e., the effective order used will differ for the lower bound and upper bound).  
The computation of $\hat q^*(z)$ is simplified significantly if all firms have the same risk profiles, i.e., $\beta_i = \beta_j$ and $\gamma_i = \gamma_j$ for all firms $i,j$.  In that case 
\[\hat q_k^*(z\vec{1}) = \exp\left((1-\frac{\sigma_M}{z})(r + \frac{z\sigma_M}{2})T\right)\hat q_k^*(\sigma_M\vec{1})^{\frac{\sigma_M}{z}}\] 
where $\hat q_k^*(\sigma_M\vec{1})$ is computed exactly as in Proposition~\ref{prop:q*-gen} {(with $f(q) = s\hat\eta(\sigma_M\vec{1})$ implicitly depends on $q$)}.
\end{remark}

As briefly discussed above, the bounds on the price of debt presented in Lemma~\ref{lemma:capm}, permit a quantification of the tradeoff for systemic risk from the diversity versus diversification problem.  This is a well-studied problem in \emph{price-mediated contagion} (see, e.g.,~\cite{CW19}); as far as the authors are aware, this problem has not previously been studied within the context of \emph{default contagion}.  As evidenced by Lemma~\ref{lemma:capm}, with fixed risks for each bank $\sigma \in \bbr^n_+$, the purely diversified portfolio by investing purely in the market portfolio provides the greatest systemic risk; however, the improvements through diversity of investments (i.e., through idiosyncratic risks) are bounded by the exposure of each bank to the market portfolio.  Therefore, from the perspective of default contagion, diversity of investments through idiosyncratic risks and returns appears to generically outperform diversification.  We wish to note that, in the price-mediated contagion literature, a similar conclusion is drawn on the downsides of diversification on systemic risk in~\cite{detering2020suffocating} in which the comparable term is portfolio similarity.

We wish to conclude this section by providing simple numerical {examples} demonstrating the bounding properties first proposed in Lemma~\ref{lemma:bound} {(and Lemma~\ref{lemma:insensitive})} applied to the CAPM setting considered within this section (Lemma~\ref{lemma:capm}).  In this example we will focus on the simple 2 bank network with societal node as depicted in Figure~\ref{fig:2bank}.
\begin{example}\label{ex:bound}
Consider the financial system with 2 banks and an additional societal node presented in Figure~\ref{fig:2bank} and discussed in detail in Section~\ref{sec:2bank}; that is, such that bank $1$ has external assets $s_1 = 3$ with interbank obligations $L_{12} = 7$ and societal obligations $L_{13} = 3$ and such that bank $2$ has external assets $s_2 = 4$ with interbank obligations $L_{21} = 3$ and societal obligations $L_{23} = 3$. For simplicity, we will consider this financial system with full recovery $\alpha_x = \alpha_L = 1$.  Herein we wish to visualize the bounds found in Lemma~\ref{lemma:capm} over varying market betas.  
To simplify this setting we will take $\sigma_M = \sigma_1 = \sigma_2 = 1$ and $\beta_1 = \beta_2$ throughout (hereafter just written as $\beta$).  
We first draw our attention to Figure~\ref{fig:bound-2}.  In this figure we see that the comonotonic lower bound and Jensen upper bound are insensitive to the choice of $\beta \in [0,1]$.  However, the conditional upper bound approaches the lower bound as market beta $\beta$ tends to $1$ (i.e.\ the comonotonic setting).  Thus for high correlations we consider the bounds to be tight, and therefore of great use.  However, as discussed in Section~\ref{sec:equity}, the market capitalization exceeds the values under the conditional and full Jensen upper bound scenarios (for banks 1 and 2). 
Finally, in Figure~\ref{fig:bound-q*} we plot the minimal solvency prices $\hat q^*(\sigma)$ and $\hat q^*(\beta\sigma_M)$ from the comonotonic and conditional expectation settings.  Though $\hat q^*(\sigma)$ is insensitive to the market beta $\beta$, we find that $\hat q_1^*(\beta\sigma_M)$ is monotonically decreasing but $\hat q_2^*(\beta\sigma_M)$ is monotonically increasing in $\beta$.  This makes it clear that the interaction of defaults between firms can greatly complicate the minimal solvency thresholds for the two extreme scenarios. 
Furthermore, these minimal solvency thresholds are one-to-one with the probability of default as summarized in Theorem~\ref{thm:comonotonic}.  As such, Figure~\ref{fig:bound-q*} demonstrates the counterintuitive result that, for bank 2, the probability of default is \emph{smaller} under the comonotonic setting than in the conditional setting.  This is contrary to what might be expected due to the relation with the price of debt and demonstrates that these bounds do not hold for the probability of default in general.
\begin{figure}[h!]
\centering
\begin{subfigure}[b]{.46\textwidth}
\includegraphics[width=\textwidth]{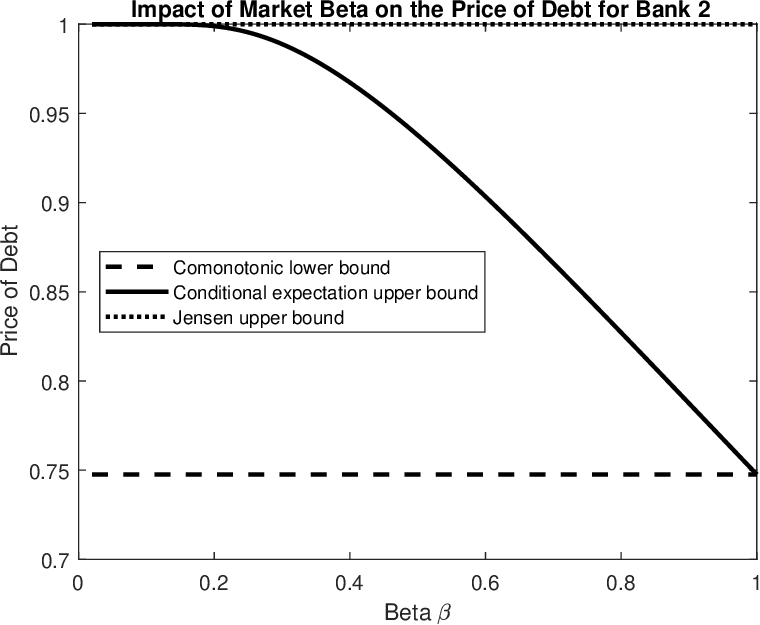}
\caption{Price of debt for bank 2 under changes in market beta from the upper (solid and dotted lines) and lower bounds (dashed line) from comonotonicity.}
\label{fig:bound-2}
\end{subfigure}
~
\begin{subfigure}[b]{.46\textwidth}
\includegraphics[width=\textwidth]{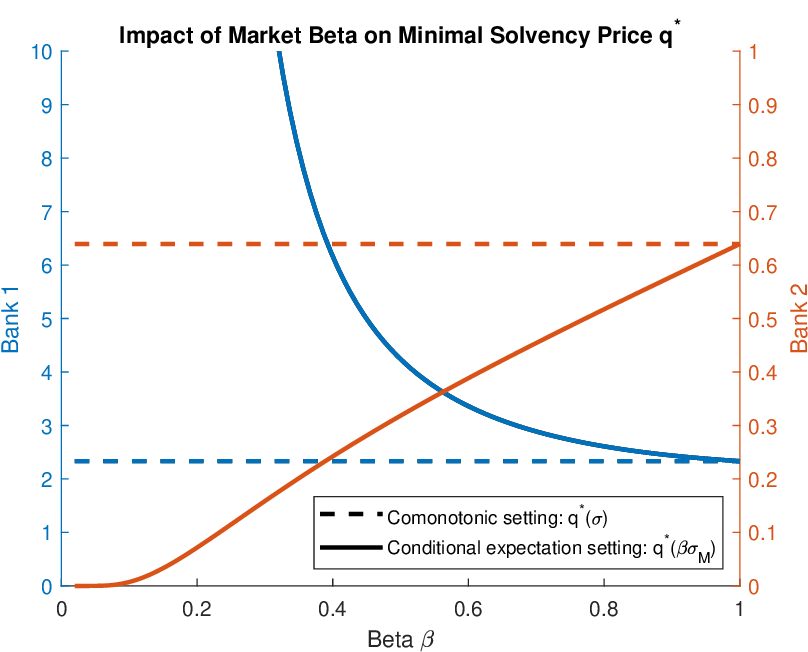}
\caption{Minimal solvency price under changes in market beta from the conditional (solid line) $\hat q^*(\beta\sigma_M)$ and comonotonic settings (dashed line) $\hat q^*(\sigma)$. Note that these no longer form upper and lower bounds.}
\label{fig:bound-q*}
\end{subfigure}
\caption{Example~\ref{ex:bound}: Demonstration of the bounds for the price of debt and minimal solvency prices as functions of the market beta.}
\label{fig:bound}
\end{figure}
\end{example}

\begin{example}\label{ex:rm-bound}
Consider again the financial system described in Example~\ref{ex:bound}.  Herein, rather than considering the price of debt, we will consider the bounds on the scalar systemic risk measures as described in Appendix~\ref{sec:sysrisk}. In particular, we will consider the systemic risk measures of Appendix~\ref{sec:insensitive} with coherent risk measure $\rho := ES_{10\%}$ (see Example~\ref{ex:ES}) and aggregation function given by the societal wealth $\Lambda(V) := V_{n+1}$.  The bounds, as provided in Lemma~\ref{lemma:insensitive}, are displayed in Figure~\ref{fig:ES_scalar}.  As for the price of debt in the above example, the comonotonic \emph{upper} bound and Jensen \emph{lower} bound are insensitive to the choice of $\beta \in [0,1]$, and the conditional \emph{lower} bound  approaches the comonotonic upper bound as market beta $\beta$ tends to $1$.  Therefore the conclusions regarding the bounds from the prior example still hold.  Notably, in contrast to Example~\ref{ex:bound}, the conditional bound is strictly better than the Jensen bound for market betas $\beta > 0$. 
\begin{figure}[h!]
\centering
\includegraphics[width=0.45\textwidth]{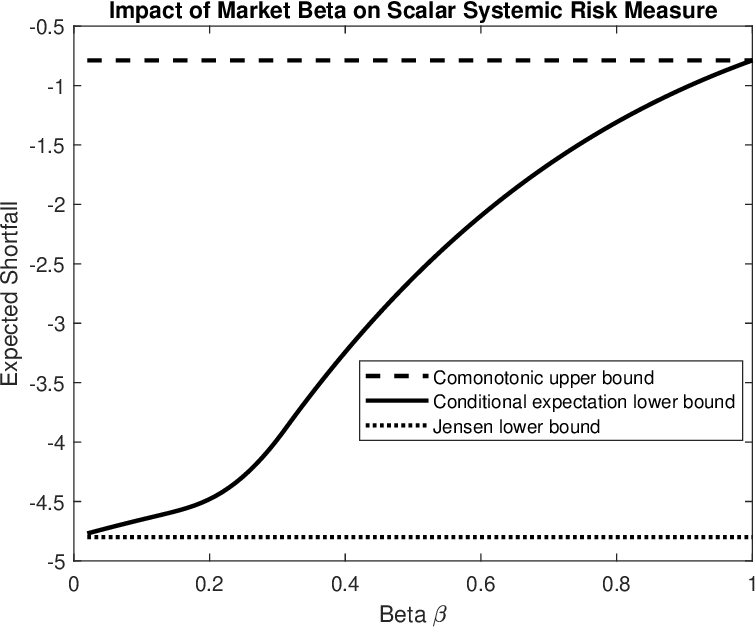}
\caption{Example~\ref{ex:rm-bound}: Demonstration of the bounds for a scalar systemic risk measure as functions of the market beta.}
\label{fig:ES_scalar}
\end{figure}
\end{example}

\section{Calibrating the interbank network for Section~\ref{sec:EBA}}\label{sec:calibration}
For {the network calibration in Section~\ref{sec:EBA}}, we consider a stylized balance sheet for each bank.  We consider banks with only two types of assets: \emph{interbank assets} $\sum_{j = 1}^n L_{ji}$ and \emph{external (risky) assets} $s_i$. Similarly, we consider three types of liabilities for each bank: \emph{interbank liabilities} $\sum_{j = 1}^n L_{ij}$, \emph{external liabilities} $L_{i,n+1}$, and \emph{capital} $C_i$.
In contrast, the EBA dataset provides the total assets $A_i$, capital $C_i$, and interbank liabilities $\sum_{j = 1}^n L_{ij}$ for each bank $i$.  

Therefore, to calibrate our interbank network, we will need to make a few simplifying assumptions and take advantage of techniques from prior literature.  In particular, as in~\cite{feinstein2017currency,CLY14,GY14}, the external (risky) assets are the difference between the total assets and interbank assets, the external obligations (owed to the societal node $L_{i,n+1}$) are equal to the total liabilities less the interbank liabilities and capital, and the interbank assets will be assumed equal to the interbank liabilities, i.e., $\sum_{j = 1}^n L_{ij} = \sum_{j = 1}^n L_{ji}$ for all banks $i$.  Thus, we can construct the remainder of our stylized balance sheet through the system of equations
\begin{align*}
s_i &= A_i - \sum_{j = 1}^n L_{ij}, \quad L_{i,n+1} = A_i - \sum_{j = 1}^n L_{ij} - C_i, \quad \bar p_i = L_{i,n+1} + \sum_{j = 1}^n L_{ij}.
\end{align*}
To verify the consistency of this calibration, we note that firm $i$'s net worth is equal to its capital, i.e., $C_i = A_i - \bar p_i$.

Finally, for our calibration, we need to consider the full nominal liabilities matrix $L \in \bbr_+^{87 \times 87}$ and not just the total interbank assets and liabilities.  In order to accomplish this task we consider the methodology of~\cite{GV16}.  That paper presents an MCMC methodology to construct the nominal liabilities matrix consistent with the total interbank assets and liabilities and which allows for a (randomized) sparsity structure.  As noted {within Section~\ref{sec:EBA}}, this example is for illustrative purposes only and thus we will consider only a single calibration of the interbank network.

\section{Details on the B\"uhlmann equilibrium and comonotonicity}\label{sec:buhlmann}
In Section~\ref{sec:motivation}, we introduced the basic details of the B\"uhlmann equilibrium.  In this section we wish to present the general construction of the B\"uhlmann equilibrium $(Y^*,\Q)$ in order to provide the comonotonicity argument that motivates this work.  In particular, we use this construction to demonstrate that every portfolio holdings solution $Y^*$ will be a comonotonic random vector.  The construction presented here follows directly from~\cite{buhlmann1984general}.  

First, recall the definition of a B\"uhlmann equilibrium from Section~\ref{sec:motivation}.  That is, consider a market of $n$ participants with concave and nondecreasing utility functions $u_i: \bbd\to \bbr$ with common domain $\bbd \subseteq \bbr$ and endowments $X_i \in L^\infty$ (such that $X_i \in \bbd$ a.s.).  A B\"uhlmann equilibrium is a pair $(Y^*,\Q)$ that is both
\begin{enumerate}
\item \emph{utility maximizing}: for every agent $i$, $Y_i^* \in \argmax_{Y_i \in L^\infty} \left\{\E\left[u_i\left(Y_i\right)\right] \; | \; \E^\Q[Y_i] \leq \E^\Q[X_i]\right\}$ and
\item \emph{market clearing}: $\sum_{i = 1}^n Y_i^* = \sum_{i = 1}^n X_i =: \xcal$.
\end{enumerate}
Note that the B\"uhlmann equilibrium problem reduces to the Arrow-Debreu equilibrium~\cite{arrow1954existence} under a finite probability space.

Since every market participant $i$ has a nondecreasing utility function $u_i$, without loss of generality we can assume $\E^\Q[Y_i^*] = \E^\Q[X_i]$ given the equilibrium measure $\Q$.  With such a reformulation, the utility maximizing condition is, trivially, equivalent to the reformulation
\begin{align*}
Y_i^* &= X_i + \hat Y_i^* - \E^\Q[\hat Y_i^*]\\
\hat Y_i^* &\in \argmax_{\hat Y_i \in L^\infty} \E\left[u_i\left(X_i + \hat Y_i - \E^\Q[\hat Y_i]\right)\right]
\end{align*}
in equilibrium such that market clearing is reformulated as $\sum_{i = 1}^n \hat Y_i^* = 0$.  With this reformulation, we have an unconstrained concave maximization problem.  As such we can consider the first order conditions for this utility maximizing problem, i.e.,
\begin{equation}\label{eq:foc}
u_i'(X_i + \hat Y_i^* - \E^\Q[\hat Y_i^*]) = \underbrace{\E\left[u_i'(X_i + \hat Y_i^* - \E^\Q[\hat Y_i^*])\right]}_{=: C_i \in \bbr_{++}} \frac{d\Q}{d\P} \quad \text{a.s.}
\end{equation}
In fact, any equilibrium must depend on the probability space only through $\xcal$; this is formalized within~\cite{borch1960,borch1962}.  Therefore, we will let $\ycal_i(\gamma) := X_i + \hat Y_i^* - \E^\Q[\hat Y_i^*]$ for every $i$ with $\sum_{i = 1}^n \ycal_i(\gamma) = \gamma$ and $\phi(\gamma)$ is the associated Radon-Nikodym derivative of $\Q$ w.r.t.\ $\P$.  (In this way $Y^* = \ycal(\xcal)$.)
With this notation, the first order conditions~\eqref{eq:foc} can be written as $u_i'(\ycal_i(\gamma)) = C_i \phi(\gamma)$ for any $\gamma$.  

Consider now the derivative (w.r.t.\ $\gamma$) of the logarithm of the first order condition~\eqref{eq:foc}, i.e., for every market participant $i$
\begin{equation}\label{eq:foc-2}
\underbrace{\frac{u_i''(\ycal_i(\gamma))}{u_i'(\ycal_i(\gamma))}}_{= -\rho_i(\ycal_i(\gamma))} \ycal_i'(\gamma) = \frac{\phi'(\gamma)}{\phi(\gamma)}
\end{equation}
where $\rho_i$ denotes the risk aversion of bank $i$.
Note that, by construction, $\sum_{i = 1}^n \ycal_i'(\gamma) = 1$, therefore~\eqref{eq:foc-2} can be exploited to provide the relation
\begin{equation}\label{eq:buhlmann-2}
1 = -\frac{\phi'(\gamma)}{\phi(\gamma)}\underbrace{\sum_{i = 1}^n \frac{1}{\rho_i(\ycal_i(\gamma))}}_{=: n \rho(\gamma)^{-1}}
\end{equation}
where $\rho$ is the harmonic average of the risk aversions $\rho_i$ for $i = 1,...,n$.

The B\"uhlmann equilibrium $(Y^*,\Q)$, through the formulation $(\ycal(\xcal),\Q)$, can be constructed from~\eqref{eq:buhlmann-2}.  
First, let us consider the equilibrium measure $\Q$ through its Radon-Nikodym derivative $\phi(\gamma)$.  Through a reformulation of~\eqref{eq:buhlmann-2}, $\phi'(\gamma) = -\frac{1}{n}\rho(\gamma)\phi(\gamma)$.  As we wish for $\phi$ to define a measure and recalling that $\frac{d\Q}{d\P} = \phi(\xcal)$, we impose the additional condition that $\E[\phi(\xcal)] = 1$.  As such, we can define the measure $\Q$ via
\[\frac{d\Q}{d\P} = \frac{\exp\left(-\frac{1}{n}\int_{c}^\xcal \rho(\xi)d\xi\right)}{\E\left[\exp\left(-\frac{1}{n}\int_{c}^\xcal \rho(\xi)d\xi\right)\right]}\]
for arbitrary $c \in \bbd$.
Let us now consider $\ycal_i(\xcal)$.  By studying both~\eqref{eq:foc-2} and~\eqref{eq:buhlmann-2}, we can construct the differential system
\begin{equation}\label{eq:ycal}
\ycal_i'(\gamma) = \frac{1}{n}\frac{\rho(\gamma)}{\rho_i(\ycal_i(\gamma))}
\end{equation}
with initial condition $\ycal_i(c) \in \bbr$ such that $\E^\Q[\ycal_i(\xcal)] = \E^\Q[X_i]$ for every bank $i$ (at the equilibrium measure $\Q$ constructed above).  
In fact, the existence and uniqueness of the B\"uhlmann equilibrium problem is equivalent to the existence and uniqueness of such an initial condition for $\ycal$; this was the approach taken in~\cite{buhlmann1984general}.  Such a result on existence and uniqueness is presented in Theorem~\ref{thm:buhlmann} below.

\begin{remark}
By construction in~\eqref{eq:ycal}, the portfolio holdings of all banks $\ycal_i$ is nondecreasing in $\gamma$.  As the equilibrium portfolio holdings are $Y^* = \ycal(\xcal)$, the comonotonicity of these holdings (as presented in Section~\ref{sec:motivation}) immediately follows.
\end{remark}

We conclude this section by presenting the existence and uniqueness results for the B\"uhlmann equilibrium.
\begin{theorem}\label{thm:buhlmann}
There exists a unique B\"uhlmann equilibrium $(Y^*,\Q)$ if either:
\begin{enumerate}
\item $\bbd = \bbr$ and the absolute risk aversions $z \mapsto \rho_i(z) := -u_i''(z)/u_i'(z) > 0$ are Lipschitz continuous for every $i=1, ...,n$; or 
\item $\bbd = \bbr_{++}$, the Inada conditions are satisfied (i.e., $\lim_{z \to 0} u_i'(z) = \infty$ and $\lim_{z \to \infty} u_i'(z) = 0$), and $z \mapsto z u_i'(z)$ are nondecreasing for every $i = 1,...,n$.
\end{enumerate}
\end{theorem}
\begin{proof}
This is proven in \cite{buhlmann1984general} if $\bbd = \bbr$ and \cite{aase1993equilibrium} if $\bbd = \bbr_{++}$.
\end{proof}

\section{Proofs of the Main Results}\label{sec:proofs}
%\subsection{Proofs for Section~\ref{sec:price}}\label{sec:proofs-price}

\begin{proof}[Proof of Theorem~\ref{thm:comonotonic}]
{As in the proof of Corollary~\ref{cor:comonotonic-rm},} this follows directly from {the comonotonicity of $X$,} the {piecewise linear} construction of $V$ in the fictitious default algorithm of Corollary~\ref{cor:FDA} and the construction of $q^*$.
\end{proof}

%\subsection{Proofs for Section~\ref{sec:bounds}}\label{sec:proofs-bounds}

\begin{proof}[Proof of Lemma~\ref{lemma:bound}]
{
We will prove this result as a corollary to Lemma~\ref{lemma:insensitive}.  Specifically, consider the utility-based shortfall risk measure with linear utility function $u(z) := z$ with threshold utility $c = 0$ (see Example~\ref{ex:ubsr}).  By construction this risk measure can be written as the negative expectation $\rho(Y) = -\E[Y]$.  We will consider two classes of aggregation functions, each indexed by the banks $i \in \{1,2,...,n+1\}$; let $\Lambda_i^V(V) = V_i$ and $\Lambda_i^p(V) = \bar p_i - V_i^-$ for any bank $i$.  Note that both classes of aggregation functions are concave and submodular (and therefore directionally concave by Proposition~\ref{prop:dcx}).  Further note that $p_i(x;\Pi,\bar p,\alpha_x,\alpha_L) = \Lambda_i^p(V(x;\Pi,\bar p,\alpha_x,\alpha_L))$ for any financial system by construction.
Therefore the bounds immediately follow from an application of Lemma~\ref{lemma:insensitive}.
}
\end{proof}

\begin{proof}[Proof of Corollary~\ref{cor:bound}]
%As with the proof of Lemma~\ref{lemma:bound}, we will first prove these results in the full recovery setting, i.e., $\alpha_x = \alpha_L = 1$.
As with the proof of Lemma~\ref{lemma:insensitive}, we will first prove these results in the full recovery setting, i.e., $\alpha_x = \alpha_L = 1$.
First, consider deterministic endowments $x \in \bbr^n_+$ and notice that $E_i(x) = x_i + \sum_{j = 1}^n \pi_{ji} p_j(x) - p_i(x)$ as shown in~\cite{EN01}.  As such the total equity in the financial system is given by the total system endowments, i.e., $\sum_{i = 1}^{n+1} E_i(x) = \sum_{i = 1}^{n+1} x_i$ since $\sum_{i = 1}^{n+1} \pi_{ji} = 1$ for any bank $j$.  For our considerations, since the societal node has no obligations back into the system, $x_{n+1}$ is arbitrary and only appears as a shift for the equity of that node.
Second, consider the random endowments $X \in (L^1_+)^n$ and its comonotonic copula $Z$.  By the above system-wide equity, it immediately follows that $\sum_{i = 1}^{n+1} \E[E_i(X)] = \sum_{i = 1}^{n+1} \E[E_i(Z)]$.
However, for the societal node $n+1$, $\E[E_{n+1}(X)] \in [\E[E_{n+1}(Z)] \, , \, \E[E_{n+1}(\E[X \; | \; {\gcal}])]]$ for any {$\gcal \in \bbg$} because $E_{n+1} \equiv V_{n+1}$ as the societal node has no obligations and thus cannot default by construction.  This immediately provides the desired bounds.

Now let $\alpha_x,\alpha_L \in [0,1]$ arbitrary. Following {Proposition~\ref{prop:alpha-bound}}, for any $x \in \bbr^n_+$ and any bank $i$,
$V_i(x;\Pi,\bar p,\alpha_x,\alpha_L) \in [V_i(\alpha_x x;\alpha_L \Pi,\bar p,1,1) \, , \, V_i(x;\Pi,\bar p,1,1)]$.
As equity is the positive part of the wealths, 
\begin{align*}
E_i(x;\Pi,\bar p,\alpha_x,\alpha_L) &\in [E_i(\alpha_x x;\alpha_L \Pi,\bar p,1,1) \, , \, E_i(x;\Pi,\bar p,1,1)].
\end{align*}
As both upper and lower bounds with bankruptcy costs are clearing solutions with full recovery, the bounds on the expected system equity follows by the above argument.
\end{proof}

\begin{proof}[Proof of Corollary~\ref{cor:comonotonic}]
First, in the grand coalition, by Corollary~\ref{cor:bound} all banks are invested comonotonically; this follows from a simple contradiction argument as the comonotonic copula of any endowment scheme provides an upper bound on the total market capitalization.
Second, we wish to show that the Shapley value is an imputation.  As provided in Chapter II.2 of~\cite{driessen2013cooperative}, this is true if $v: 2^\ncal \to \bbr_+$ is superadditive.  Let $\ccal,\dcal \subseteq \ncal$ be disjoint sets:
    \begin{align*}
    v(\ccal \cup \dcal) &= \sup_{X_{\ccal \cup \dcal} \in \prod_{i \in \ccal \cup \dcal} \bcal_i} \inf_{X_{-(\ccal \cup \dcal)} \in \prod_{j \not\in \ccal \cup \dcal} \bcal_j} \sum_{i \in \ccal \cup \dcal} \E[E_i(X_{\ccal \cup \dcal} , X_{-(\ccal \cup \dcal)})] \\
        &= \sup_{X_{\ccal \cup \dcal} \in \prod_{i \in \ccal \cup \dcal} \bcal_i} \inf_{X_{-(\ccal \cup \dcal)} \in \prod_{j \not\in \ccal \cup \dcal} \bcal_j} \left(\sum_{i \in \ccal} \E[E_i(X_{\ccal \cup \dcal} , X_{-(\ccal \cup \dcal)})] \right.\\
        &\qquad\qquad \left.+ \sum_{i \in \dcal} \E[E_i(X_{\ccal \cup \dcal},X_{-(\ccal \cup \dcal)})]\right)\\
        &\geq \sup_{X_{\ccal \cup \dcal} \in \prod_{i \in \ccal \cup \dcal} \bcal_i} \left(\inf_{X_{-(\ccal \cup \dcal)} \in \prod_{j \not\in \ccal \cup \dcal} \bcal_j} \sum_{i \in \ccal} \E[E_i(X_{\ccal \cup \dcal} , X_{-(\ccal \cup \dcal)})] \right.\\
        &\qquad\qquad \left.+ \inf_{X_{-(\ccal \cup \dcal)} \in \prod_{j \not\in \ccal \cup \dcal} \bcal_j} \sum_{i \in \dcal} \E[E_i(X_{\ccal \cup \dcal},X_{-(\ccal \cup \dcal)})]\right)\\
        &= \sup_{X_\ccal \in \prod_{i \in \ccal} \bcal_i} \sup_{X_\dcal \in \prod_{i \in \dcal} \bcal_i} \left(\inf_{X_{-(\ccal \cup \dcal)} \in \prod_{j \not\in \ccal \cup \dcal} \bcal_j} \sum_{i \in \ccal} \E[E_i(X_\ccal , X_\dcal , X_{-(\ccal \cup \dcal)})] \right.\\
        &\qquad\qquad \left.+ \inf_{X_{-(\ccal \cup \dcal)} \in \prod_{j \not\in \ccal \cup \dcal} \bcal_j} \sum_{i \in \dcal} \E[E_i(X_\ccal , X_\dcal , X_{-(\ccal \cup \dcal)})]\right)\\
        &\geq \sup_{X_\ccal \in \prod_{i \in \ccal} \bcal_i} \left(\inf_{X_\dcal \in \prod_{i \in \dcal} \bcal_i} \inf_{X_{-(\ccal \cup \dcal)} \in \prod_{j \not\in \ccal \cup \dcal} \bcal_j} \sum_{i \in \ccal} \E[E_i(X_\ccal , X_\dcal , X_{-(\ccal \cup \dcal)})] \right.\\
        &\qquad\qquad \left.+ \sup_{X_\dcal \in \prod_{i \in \dcal} \bcal_i} \inf_{X_{-(\ccal \cup \dcal)} \in \prod_{j \not\in \ccal \cup \dcal} \bcal_j} \sum_{i \in \dcal} \E[E_i(X_\ccal , X_\dcal , X_{-(\ccal \cup \dcal)})]\right)\\
        &\geq \sup_{X_\ccal \in \prod_{i \in \ccal} \bcal_i} \inf_{X_\dcal \in \prod_{i \in \dcal} \bcal_i} \inf_{X_{-(\ccal \cup \dcal)} \in \prod_{j \not\in \ccal \cup \dcal} \bcal_j} \sum_{i \in \ccal} \E[E_i(X_\ccal , X_\dcal , X_{-(\ccal \cup \dcal)})] \\
        &\qquad\qquad + \inf_{X_\ccal \in \prod_{i \in \ccal} \bcal_i} \sup_{X_\dcal \in \prod_{i \in \dcal} \bcal_i} \inf_{X_{-(\ccal \cup \dcal)} \in \prod_{j \not\in \ccal \cup \dcal} \bcal_j} \sum_{i \in \dcal} \E[E_i(X_\ccal , X_\dcal , X_{-(\ccal \cup \dcal)})]\\
        &\geq \sup_{X_\ccal \in \prod_{i \in \ccal} \bcal_i} \inf_{X_\dcal \in \prod_{i \in \dcal} \bcal_i} \inf_{X_{-(\ccal \cup \dcal)} \in \prod_{j \not\in \ccal \cup \dcal} \bcal_j} \sum_{i \in \ccal} \E[E_i(X_\ccal , X_\dcal , X_{-(\ccal \cup \dcal)})] \\
        &\qquad\qquad + \sup_{X_\dcal \in \prod_{i \in \dcal} \bcal_i} \inf_{X_\ccal \in \prod_{i \in \ccal} \bcal_i} \inf_{X_{-(\ccal \cup \dcal)} \in \prod_{j \not\in \ccal \cup \dcal} \bcal_j} \sum_{i \in \dcal} \E[E_i(X_\ccal , X_\dcal , X_{-(\ccal \cup \dcal)})]\\
        &= \sup_{X_\ccal \in \prod_{i \in \ccal} \bcal_i} \inf_{X_{-\ccal} \in \prod_{j \not\in \ccal} \bcal_j} \sum_{i \in \ccal} \E[E_i(X_\ccal , X_{-\ccal})] + \sup_{X_\dcal \in \prod_{i \in \dcal} \bcal_i} \inf_{X_{-\dcal} \in \prod_{j \not\in \dcal} \bcal_j} \sum_{i \in \dcal} \E[E_i(X_\dcal , X_{-\dcal})]\\
        &= v(\ccal) + v(\dcal)
    \end{align*}
    and the proof is complete.
\end{proof}

%%\bibliographystyle{plain}
%\bibliographystyle{plainnat}
%%\bibliographystyle{informs2014}
%\bibliography{bibtex2}

\end{document}